\documentclass[10pt,journal,compsoc,twocolumn]{IEEEtran}
\pdfoutput=1

\usepackage[nocompress]{cite}
\usepackage{threeparttable}
\usepackage{booktabs}
\usepackage{float}
\usepackage{extarrows}
\usepackage{graphicx}
\usepackage{epstopdf}
\usepackage{xcolor}
\usepackage{epsfig}

\usepackage{amsmath}
\usepackage{amssymb}
\usepackage{amsthm}
\usepackage{color}
\usepackage{url}
\usepackage{array}
\usepackage{multirow}
\usepackage{algorithm}
\usepackage{algorithmic}
\usepackage{bigstrut}
\usepackage[T1]{fontenc}
\usepackage[justification=centering]{caption}

\usepackage{rotating}

\renewcommand{\algorithmicrequire}{\quad \textbf{Input:}} 
\renewcommand{\algorithmicensure}{\quad \textbf{Output:}} 

\newtheorem{definition}{Definition}
\newtheorem{theorem}{Theorem}
\newtheorem{lemma}{Lemma}
\newtheorem{corollary}{Corollary}
\newtheorem{proposition}{Proposition}

\newtheorem{remark}{Remark}
\newtheorem{assumption}{Assumption}

\numberwithin{equation}{section}
\numberwithin{definition}{section}
\numberwithin{theorem}{section}
\numberwithin{lemma}{section}
\numberwithin{corollary}{section}
\numberwithin{proposition}{section}
\numberwithin{remark}{section}
\numberwithin{assumption}{section}
\newlength{\figurewidth}
\begin{document}

\title{A Generalized Tensor Formulation for Hyperspectral Image Super-Resolution Under General Spatial Blurring}

\author{Yinjian~Wang,
	Wei~Li,~\IEEEmembership{Senior Member,~IEEE}, 
	Yuanyuan~Gui,
	Qian Du,~\IEEEmembership{Fellow,~IEEE},\\
	and James~E.~Fowler,~\IEEEmembership{Fellow,~IEEE}
	\IEEEcompsocitemizethanks{\IEEEcompsocthanksitem This paper is supported by the National Key R\&D Program of China (Grant no. 2021YFB3900502). (Corresponding author: Wei Li.)
	}
	\IEEEcompsocitemizethanks{\IEEEcompsocthanksitem Y. Wang, W. Li, and Y. Gui are with the School of Information and Electronics,
		Beijing Institute of Technology, and Beijing Key Laboratory of Fractional
		Signals and Systems, 100081 Beijing, China (e-mail: yinjw@bit.edu.cn, liwei089@ieee.org, 953647315@qq.com).
	}
	\IEEEcompsocitemizethanks{\IEEEcompsocthanksitem Q.~Du and J.~E.~Fowler are with the Department of Electrical and Computer Engineering, Mississippi State University, Starkville, MS 39762 USA (e-mail: du@ece.msstate.edu, fowler@ece.msstate.edu).
}}

{}

\IEEEcompsoctitleabstractindextext{%

\begin{abstract}
Hyperspectral super-resolution is commonly accomplished by the fusing
of a hyperspectral imaging of low spatial resolution with a
multispectral image of high spatial resolution, and many tensor-based
approaches to this task have been recently proposed. Yet, it is
assumed in such tensor-based methods that the spatial-blurring
operation that creates the observed hyperspectral image from the
desired super-resolved image is separable into independent horizontal
and vertical blurring. Recent work has argued that such separable
spatial degradation is ill-equipped to model the operation of real
sensors which may exhibit, for example, anisotropic blurring.  To
accommodate this fact, a generalized tensor formulation based on a
Kronecker decomposition is proposed to handle any general
spatial-degradation matrix, including those that are not separable as
previously assumed. Analysis of the generalized formulation reveals
conditions under which exact recovery of the desired super-resolved
image is guaranteed, and a practical algorithm for such recovery,
driven by a blockwise-group-sparsity regularization, is proposed.
Extensive experimental results demonstrate that the proposed
generalized tensor approach outperforms not only traditional
matrix-based techniques but also state-of-the-art tensor-based
methods; the gains with respect to the latter are especially
significant in cases of anisotropic spatial blurring.
\end{abstract}
\begin{IEEEkeywords}
Image fusion, hyperspectral super-resolution, tensor factorization, recoverability, sparse coding, nonconvex surrogate.
\end{IEEEkeywords}}
\maketitle
\IEEEdisplaynotcompsoctitleabstractindextext
\IEEEpeerreviewmaketitle

\IEEEraisesectionheading{\section{Introduction}\label{sec:introduction}}
\IEEEPARstart{H}{yperspectral} imagery is designed to capture a
densely-sampled spectral signature for each pixel in an image,
providing much finer spectral information than other imaging
modalities. However, limited by the trade-off in current optical
imaging systems, detailed spectral information comes at the cost of
low spatial resolution in a hyperspectral image (HSI) \cite{7946218,10252045}.  This fact has
severely restricted the use of HSI in such applications as
classification, anomaly detection, and object tracking. To deal with
this issue, HSI super-resolution (HSR) has been the
subject of substantial recent work, and HSR can be most effectively
achieved through the fusion of an HSI with an image
possessing high spatial resolution, such as a multispectral image (MSI).

From the perspective of data reconstruction, HSR aims to recover the
super-resolution HSI (SRI) from the observed HSI and MSI.  Since the
earliest HSR studies (e.g., \cite{MKM1999,HEW2004}), the spatial
degradation from the SRI to the HSI, as well as the spectral
degradation from the SRI to the MSI, are modeled in a manner of a
matrix operation. In this article, we formally refer to such modeling
as the matrix formulation of HSR (i.e., MF-HSR). In the literature,
MF-HSR is the most common framework for HSR, and it explicitly
reflects the inherent ill-posedness of the HSR problem, since the
degradation equations therein are undetermined. To alleviate such
ill-posedness, much existing HSR work has taken into account various
priors induced from the intrinsic spatial and spectral correlations of
the SRI to develop MF-HSR methods through various means such as
spectral unmixing \cite{KMW2011,YYI2012,LBS2015,WBD2016,LMC2018},
sparse/low-rank representation
\cite{ASM2014,DFS2016,HSZ2018,SBA2015,HYX2020,XZB2021,DL2019,LLL2021},
non-local similarities \cite{HYL2022}, Bayesian learning
\cite{ASM2015,WDT2015,ZCZ2022,DQX2023}, and deep learning
\cite{NUS2022,DGL2023,XZZ2022}.

Although MF-HSR fusion models can be effective, they all rely on the
reshaping of the HSI into a 2D matrix, or 1D vector, for
processing. However, there is increasing interest
\cite{HYL2022,DFL2017,CHH2020,KFS2018} in instead considering the HSI
as a 3D tensor in order to exploit higher-order characteristics.
Accordingly, many efforts have been devoted to adopting
tensor-analysis tools to build tensor-based HSR models (e.g.,
\cite{LDF2018,KFS2018,PUC2020,DFH2021,BZX2021,WGH2020,XWC2020a,YXZ2023}).
To
adequately exploit the low-dimensional structure of high-dimensional
HSI data, the resulting tensor formulation of HSR (TF-HSR) usually
performs factorization in every single dimension. Thus, the spatial
domain is decomposed into two orthogonal dimensions which precludes
the design of TF-based methods following the framework of MF-HSR,
namely, the expressing of the spatial degradation as a single matrix
operation.  To address this problem, \cite{DFL2017,KFS2018} made the
additional assumption that the spatial-degradation matrix is separable
in the sense of being a Kronecker product of two independent vertical
and horizontal degradation matrices.  Using such Kronecker-separable
spatial degradation, a multitude of TF-HSR methods (e.g.,
\cite{LDF2018,KFS2018,PUC2020,DFH2021,BZX2021,WGH2020,XWC2020,YXZ2023})
have appeared in recent literature and achieved state-of-the-art
performance.

Yet, the success of a TF-HSR is highly dependent on the assumption of
separability of the spatial-degradation matrix.  As pointed out in
\cite{KFS2018}, the separable assumption holds for the most commonly
used average and Gaussian blurring kernels.  However, it has been
argued recently \cite{YZX2022,HCP2023,GX2023} that such a simple
average kernel or isotropic Gaussian kernel (IGK) is insufficient to
characterize blurring processes exhibited by real sensors. For HSI in
particular, sensor motion may result in anisotropic blurring
ill-captured by the usual IGK (e.g., \cite{ZVC2017,WM2018,ZSJ2014}).
In such cases, more general kernels, such as an anisotropic Gaussian
kernel (AGK), may more accurately model real-world spatial blurring.
Under this observation, several questions naturally arise.  Does the
separable assumption still hold under more realistic (anisotropic)
blurring kernels? If it does not, is there a criterion according to
which one can tell when TF-HSR is suitable, and when is not?
And, for more realistic blurring kernels for which TF-HSR fails, what
should be done to still take advantage of tensors in HSR?

In addition, another open issue for TF-HSR is blind HSR, i.e., in
which HSR is conducted without knowledge of the spatial- and
spectral-degradation matrices.  MF-HSR models have achieved this goal
by various techniques, since MF-HSR is based on the physical
interpretation of the blurring process. However, such is not the case
for TF-HSR models, wherein modeling of spatial blurring is
mathematically-driven. As a result, it is not known how to best
estimate the pair of matrices that constitute separable spatial
blurring, and most TF-HSR techniques (e.g.,
\cite{KFS2018,PUC2020,DFH2021,BZX2021}) resort to simply entirely
ignoring degradation information such that they remain agnostic to it.

In this paper, we examine the TF-HSR framework under the
assumption of general blurring and determine
that TF-HSR is not feasible under such conditions.
Thus, we propose a generalization of the TF-HSR
problem, which we refer to as generalized TF-HSR, or GTF-HSR.
The conditions for exact recovery of the
proposed GTF-HSR determined, and  
a two-stage optimization strategy is devised for
iterative estimation under a block-sparsity prior. 
Specifically, the primary contributions of this paper are:
\begin{itemize}
	\item
	We present GTF-HSR, which is based on the Kronecker
	decomposition. Leveraging the property of such decomposition, we
	obtain a generalized separable condition on the
	spatial-degradation matrix for GTF and establish an equivalency
	between the proposed GTF-HSR and MF-HSR. That is, the proposed
	GTF-HSR can be applied to as many cases as can MF-HSR, and, most
	importantly, when TF-HSR cannot.
	\item
	The proposed capacity of GTF-HSR to achieve exact SRI recovery
	(i.e., its \textit{recoverability}) is analyzed, and the conditions for
	exact recovery are determined. We also deduce the conditions when
	exact recovery by TF-HSR is impossible, which further emphasizes
	the benefit of our generalized approach.
	\item
	We invoke a blockwise group sparsity (BGS) as a new, higher-level
	prior.  The proposed BGS characterizes the grouped property of the
	sparsity in the core tensor of a Tucker decomposition of the
	SRI. With the help of a tensor-unfolding strategy as well as a
	nonconvex surrogate, BGS is easily imposed upon the core tensor
	to regularize the GTF-HSR problem, managing to explore the
	multi-linear properties of higher-order data in a compact form
	in an algorithmic framework we call BGS-GTF-HSR.
	\item
	To tackle the resulting regularized large-scale nonconvex
	BGS-GTF-HSR optimization problem, a two-stage optimization
	strategy consisting of subspace identification and BGS coding is
	devised. In the first stage, we cascade the two spatial subspace
	bases extracted from the MSI and HSI through singular-value
	decompostion (SVD) and sparse-dictionary learning, estimating the
	spectral subspace from the HSI alone via SVD. In the second stage,
	the BGS coefficients are estimated by an alternating-directions
	optimization.
	\item
	Extensive experiments on both synthetic and real-world datasets
	for blind and non-blind HSR problems demonstrate the superiority
	of the proposed GTF-HSR method compared to not only traditional MF-HSR
	methods but also to state-of-the-art TF-HSR methods.
\end{itemize}

The remainder of the paper is organized as follows.  First in
Sec.~\ref{sec:background}, we overview requisite background on tensors
as well as existing MF-HSR and TF-HSR frameworks.  Then, in
Sec.~\ref{sec:gtfhsr}, we introduce GTF-HSR, analyze it and its
recoverabiity, and compare it analytically to the TF-HSR
strategy. Sec.~\ref{sec:bgsgtfhsr} then introduces the BGS-GTF-HSR
optimization algorithm, while Sec.~\ref{sec:results} presents a body
of experimental results.  Finally, Sec.~\ref{sec:conclusions}
concludes the manuscript.

\section{Background}
\label{sec:background}
In this section, we overview pertinent aspects of tensor mathematics
and the Tucker decomposition, as well as formulations of HSR found in prior
literature. {Also, we briefly introduce AGK.}
\subsection{Tensors and Notation}
\label{sec:notation}
In this paper, a scalar is denoted as $a$, a vector is denoted as
$\mathbf{a}$, a matrix is denoted as $\mathbf{A}$, and an $N$-order
tensor is denoted as $\mathcal{A}$. For two tensors $\mathcal{A}$ and
$\mathcal{B}$, we denote their inner product as
$\left\langle\mathcal{A},\,\mathcal{B}\right\rangle=\sum_{i_1,i_2,\dots,
	i_N}\mathcal{A}_{i_1,i_2,\dots,i_N}\mathcal{B}_{i_1,i_2,\dots,i_N}$
where $\mathcal{A}_{i_1,i_2,\cdots,i_N}$ is the element at location
$(i_1,i_2,\dots, i_N)$ of $\mathcal{A}$. The Frobenius norm of a
tensor is then defined by $\left\Vert\mathcal{A}\right\Vert_F =
\sqrt{\left\langle\mathcal{A},\,\mathcal{A}\right\rangle}$. The
Kronecker product between matrices is denoted by $\otimes$, while
$\mathsf{col}\left\{A\right\}$ and $\mathsf{null}\left\{A\right\}$
denote the column and null spaces of a matrix, respectively. If a
variable is drawn from some absolutely continuous distribution, we
call it ``generic'' after \cite{Del2008}.  For an $N$-order tensor
$\mathcal{A}\in\mathbb{R}^{I_1\times I_2\times\cdots\times I_N}$, its
mode-$n$ unfolding is a matrix denoted by
$\mathbf{A}_{[n]}\in\mathbb{R}^{I_n\times I_1\cdots I_{n-1} 
	I_{n+1} \cdots I_N}.$ The mode-$n$ product (denoted $\times_n$)
between tensor $\mathcal{A}\in\mathbb{R}^{I_1\times
	I_2\times\cdots\times I_N}$ and matrix
$\mathbf{U}\in\mathbb{R}^{J\times I_n}$ is a tensor defined such that
$\left(\mathcal{A}\times_n\mathbf{U}\right)_{[n]}=\mathbf{U}\mathbf{A}_{[n]}.$

In the sequel, we use the notation $\mathtt{reshape}(\cdot)$,
$\mathtt{permute}(\cdot)$, $\mathtt{svd}(\cdot)$, and
$\mathtt{soft}(\cdot)$ to denote tensor operators with the corresponding
functionality of their \textsc{Matlab} namesakes.  Accordingly, we
define the vectorization and unvectorization operators as
$\mathsf{Vec}\left(\cdot\right)=\mathtt{reshape}\left(\cdot,[\,],1\right)$
and
$\mathsf{Unv}\left(\cdot|J,K\right)=\mathtt{reshape}\left(\cdot,J,K\right)$,
respectively.

Finally, we define the Tucker decomposition of an arbitrary $N$-order
tensor $\mathcal{Z}\in\mathbb{R}^{I_1\times I_2\times\cdots\times
	I_N}$ as
\begin{equation}
	\mathcal{Z}=\mathcal{G}\times_1 \mathbf{U}_1\times_2
	\mathbf{U}_2\times _3\cdots\times_N \mathbf{U}_N ,
	\label{eq:td}
\end{equation}
where $\mathcal{G}\in\mathbb{R}^{J_1\times J_2\times\cdots\times J_N}$
is the core tensor, and $\left\{\mathbf{U}_n\in\mathbb{R}^{I_n\times
	J_n}\right\}_{n=1}^N$ are the factor matrices. In line with this
decomposition, the Tucker-rank of $\mathcal{Z}$ is defined in a
multi-rank form as
\begin{equation}
	\operatorname{rank}_{T}\left\{\mathcal{Z}\right\} =
	\left(\operatorname{rank}\left\{\mathbf{Z}_{[1]}\right\},\,
	\operatorname{rank}\left\{\mathbf{Z}_{[2]}\right\},\,\cdots,\,
	\operatorname{rank}\left\{\mathbf{Z}_{[N]}\right\}\right) .
\end{equation}
Note that Tucker decomposition exists if and only if $J_n\geq
\operatorname{rank}\left\{\mathbf{Z}_{[n]}\right\},\,\forall n$. Two
important calculation rules are
\begin{equation}
	\mathbf{Z}_{[n]}=\mathbf{U}_n\mathbf{G}_{[n]}\left(\mathbf{U}_N\otimes\cdots\otimes\mathbf{U}_{n+1}\otimes\mathbf{U}_{n-1}\otimes\cdots\otimes\mathbf{U}_1\right)^T,
\end{equation}
and
\begin{equation}
	\mathsf{Vec}\left(\mathcal{Z}\right)=\left(\mathbf{U}_N\otimes\cdots\otimes\mathbf{U}_1\right)\mathsf{Vec}\left(\mathcal{G}\right).
\end{equation}
The reader is referred to, e.g., \cite{KB2009}, for greater
elaboration on the definitions and notations presented in this
section.
\subsection{Existing Formulations of HSR}
We now describe mathematically the two main frameworks for HSR as
existing in prior literature: MF-HSR and TF-HSR.  These formulations
are built on the assumption---common in the literature---of a
spectral-degradation matrix $\mathbf{R}$ that comprises a
spectral-downsampling process, along with $\mathbf{D}$, a matrix that
encapsulates the hyperspectral sensor's spatial blurring coupled with
the subsampling entailed by the imaging process.
\begin{definition}[MF-HSR]
	\label{Def: Matrix Formulation}
	{\it Given HSI $\mathcal{X}\in \mathbb{R} ^{m_1\times
			m_2\times S}$, MSI $\mathcal{Y}\in \mathbb{R} ^{M_1\times
			M_2\times s}$, spectral-degradation matrix $\mathbf{R}\in
		\mathbb{R}^{s\times S}$, and spatial-degradation matrix
		$\mathbf{D}\in \mathbb{R}^{M_1M_2\times m_1m_2}$, with $m_1<
		M_1$, $m_2< M_2$, and $s< S$, the MF-HSR problem seeks the
		most appropriate SRI $\mathcal{Z}\in \mathbb{R} ^{M_1\times
			M_2\times S}$, such that
		\begin{equation}
			\begin{aligned}
				\mathbf{X}_{[3]} & =\mathbf{Z}_{[3]}\mathbf{D}, \\
				\mathbf{Y}_{[3]} & =\mathbf{R}\mathbf{Z}_{[3]}.
			\end{aligned}
			\label{eq: Matrix Formulation}
	\end{equation}}
\end{definition}

\begin{definition}[TF-HSR]
	{\it Given HSI $\mathcal{X}\in \mathbb{R} ^{m_1\times m_2\times S}$,
		MSI $\mathcal{Y}\in \mathbb{R} ^{M_1\times M_2\times s}$,
		spectral-degradation matrix $\mathbf{R}\in \mathbb{R}^{s\times
			S}$, and spatial-degradation matrices $\mathbf{P}_1\in
		\mathbb{R}^{m_1\times M_1}$ and $\mathbf{P}_2\in
		\mathbb{R}^{m_2\times M_2}$, with $m_1< M_1$, $m_2< M_2$, and $s<
		S$, the TF-HSR problem seeks the most appropriate SRI
		$\mathcal{Z}\in \mathbb{R} ^{M_1\times M_2\times S}$, such that
		\begin{equation}
			\begin{aligned}
				\mathcal{X} & =\mathcal{Z}\times_1\mathbf{P}_1\times_2\mathbf{P}_2, \\
				\mathcal{Y} & =\mathcal{Z}\times_3\mathbf{R}.
			\end{aligned}
			\label{eq: Tensor Formulation}
	\end{equation}}
	\label{Def: Tensor Formulation}
\end{definition}
Additionally, TF-HSR implicitly requires a separable
spatial-degradation operator:
\begin{assumption}
	{\it Suppose $\mathbf{D}\in \mathbb{R}^{M_1M_2\times m_1m_2}$ is the
		spatial-degradation matrix in MF-HSR, then TF-HSR assumes that
		there exist $\mathbf{P}_1\in \mathbb{R}^{m_1\times M_1}$ and
		$\mathbf{P}_2\in \mathbb{R}^{m_2\times M_2}$ such
		that
		\begin{equation}
			\mathbf{D}=(\mathbf{P}_2\otimes \mathbf{P}_1)^T.
		\end{equation}
		\label{asum: Kronecker separable assumption}}
\end{assumption}
\subsection{{AGK}}
\label{sec:AGK}
{Let $\mathbf{\Phi}\in\mathbb{R}^{(2r+1)\times(2r+1)}$ denote an AGK. Then each of its elements, $\mathbf{\Phi}_{i,j}$, is calculated as
	\begin{equation}
		\begin{aligned}
			\mathbf{\Phi}_{i,j}=\frac{1}{2\pi}&\sqrt{\vert\mathbf{\Lambda}\vert}\exp\left(-\frac{1}{2}[i\quad j]\mathbf{\Lambda}[i\quad j]^T\right),\\&i,j\in\left\{-r,\cdots,r\right\}
		\end{aligned}
	\end{equation}
	where $\mathbf{\Lambda}=\begin{bmatrix}
		\cos\theta& -\sin\theta\\
		\sin\theta&\cos\theta
	\end{bmatrix}\begin{bmatrix}
		a& \\
		&b
	\end{bmatrix}\begin{bmatrix}
		\cos\theta& \sin\theta\\
		-\sin\theta&\cos\theta
	\end{bmatrix}$. Thus AGK is determined by three parameters $\theta,a,b.$ To guarantee the positive-definiteness of $\mathbf{\Lambda}$, we require $a,b$ to be positive. We also note that when $a=b$, AGK degrades to IGK.}

\section{A Generalized Tensor Formulation}
\label{sec:gtfhsr}
To propose a generalized formulation for HSR that is more appropriate
when spatial degradations are anisotropic, we first examine the
feasibility of TF-HSR, specifically, the validity of Asm.~\ref{asum:
	Kronecker separable assumption}.  Then, in
Sec.~\ref{sec:gtfhsr_framework}, we present the proposed GTF-HSR
framework that generalizes TF-HSR in order to handle anisotropic
degradation. We close this discussion with an examination of the
potential of GTF-HSR to exactly recover the desired SRI as well as
ramifications of blind HSR on this recovery in
Secs.~\ref{sec:recoverability} and \ref{sec:blind}, respectively.
\subsection{Feasibility of TF-HSR}
\label{sec: Feasibility of TeF}
To begin, we introduce the Kronecker decomposition (KD):
\begin{theorem}[Kronecker Decomposition]
	\textit{For any matrix $\mathbf{W}\in\mathbb{R}^{J_1J_2\times K_1K_2}$, there exist two sets of matrices, $\left\{\mathbf{M}_1^{(r)}\right\}_{r=1}^R\subseteq\mathbb{R}^{J_1\times K_1}$ and $\left\{\mathbf{M}_2^{(r)}\right\}_{r=1}^R\subseteq\mathbb{R}^{J_2\times K_2}$, such that
		\begin{equation}
			\mathbf{W}=\sum_{r=1}^R\mathbf{M}_1^{(r)}\otimes\mathbf{M}_2^{(r)}.
			\label{eq:kd}
	\end{equation}}
	\label{Th: Kronecker Decomposition}
\end{theorem}
We note that Thm.~\ref{Th: Kronecker Decomposition} is a direct
composition of the Kronecker-product SVD described in
\cite[Thm.~12.3.1]{GV2013}.

According to Thm.~\ref{Th: Kronecker Decomposition}, it is clear
that the spatial degradation matrix $\mathbf{D}$ can be decomposed
into the sum of $R$ Kronecker products. The validity of Asm.~\ref{asum: Kronecker separable assumption} then reduces as to whether $R$
can equal 1. In other words, we need to determine the minimum
value of $R$ in the KD of $\mathbf{D}$. To this end, we define
Kronecker rank:
\begin{definition}[Kronecker Rank]
	\textit{The Kronecker rank of matrix
		$\mathbf{W}\in\mathbb{R}^{J_1J_2\times K_1K_2}$, denoted as 
		{$\mathsf{kr}(J_1,K_1)_{\mathbf{W}}$}, is defined as the minimal $R$ such that
		$R$ pairs of matrices generate the KD of $\mathbf{W}$ as in \eqref{eq:kd}.}
	\label{Def: Kronecker Rank}
\end{definition}
The issue then becomes how to determine $\mathsf{kr}_{\mathbf{D}}$\footnote{{For a better presentation, we shorten the notation for the Kronecker rank of $\mathbf{D}\in\mathbb{R}^{M_1M_2\times m_1m_2}$. In the sequel, $\mathsf{kr}_{\mathbf{D}}\triangleq\mathsf{kr}(M_1,m_1)_{\mathbf{D}}$.}}. In
response, we recall that the spatial-degradation matrix $\mathbf{D}$
models the process of blurring and downsampling; thus, the spatial
degradation in \eqref{eq: Matrix Formulation} can be further
elaborated as
\begin{equation}
	\mathbf{X}_{[3]} =\mathbf{Z}_{[3]}\mathbf{D}=(\left({\cal Z}\ast\mathbf{\Phi}\right)_{\downarrow})_{[3]} ,
	\label{eq:Physical Expression}
\end{equation}
where $\mathbf{\Phi}\in\mathbb{R}^{\phi\times\phi}$ is the
spatial-blurring kernel, $\ast$ denotes the periodic 2D convolution,
and the subscript $\downarrow$ is uniform
downsampling.
As such, the degradation matrix $\mathbf{D}$ is endowed
with significant structure and satisfies a wealth of properties. Accordingly,
we have the following proposition.
\begin{proposition}\footnote{{We present its proof assuming the convolution operator in \eqref{eq:Physical Expression} is periodic. However, one can verify the same conclusion holds subject to aperiodic convolutions such as that with zero-padding strategy.}}
	\textit{For the spatial-degradation matrix $\mathbf{D}$ in
		\eqref{eq: Matrix Formulation}, since it is physically modeled as
		\eqref{eq:Physical Expression}, we
		have
		\begin{equation}
			\mathsf{kr}_{\mathbf{D}}=\operatorname{rank}\left\{\mathbf{\Phi}\right\}.
		\end{equation}
	}
	\label{Prop:Kr-Rank}
\end{proposition}
The proof of Prop.~\ref{Prop:Kr-Rank} can be found in
the supplemental material.
We conclude from Prop.~\ref{Prop:Kr-Rank}
that Asm.~\ref{asum: Kronecker separable assumption} holds true
if and only if $\operatorname{rank}\left\{\mathbf{\Phi}\right\}=1$,
i.e., the feasibility of TF-HSR rests solely on the rank of
the blurring kernel.

To gauge the likelihood of having unity
$\operatorname{rank}\left\{\mathbf{\Phi}\right\}$, we consider
Fig.~\ref{fig:KernelSVD} which visualizes the distribution of the
singular values of different blurring kernels.  We note that the
second singular values for the isotropic-Gaussian and average kernels drop
sharply to zero, meaning these kernels are rank-1, thereby supporting
the the use of TF-HSR with them. However, for the more complicated
AGK, the curve descends much more slowly, thus
$\operatorname{rank}\left\{\mathbf{\Phi}\right\}>1$, and it is no
longer reasonable to apply TF-HSR when such anisotropic kernels are in
effect. Indeed, Sec.~\ref{sec:results}
shows empirically that the fusion performance of current
TF-HSR-based approaches deteriorates significantly under AGK
blurring. Thus, in the next section, we reformulate TF-HSR to
accommodate more general blurring processes.

\begin{figure}[t]
	\centering
	\setlength{\tabcolsep}{0.15mm}
	\begin{tabular}{cm{0.9\figurewidth}}\rotatebox[origin=c]{90}{\footnotesize{Magnitude of singular values}}&
		\includegraphics[width=0.9\figurewidth]{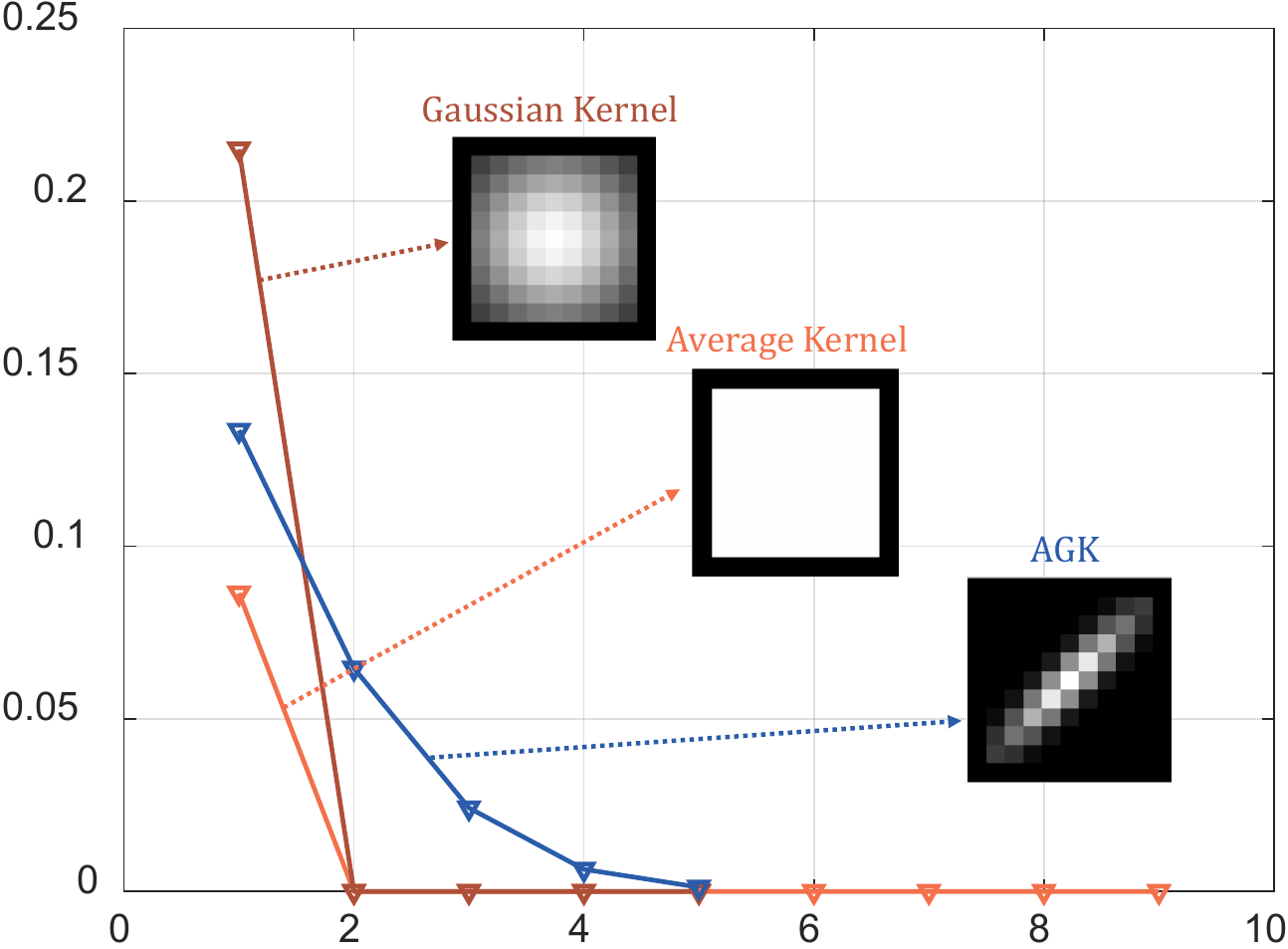}\\\quad&\multicolumn{1}{c}{\footnotesize{Number of singular values}}
	\end{tabular}
	\caption{\label{fig:KernelSVD} {The distribution singular values for
			some representative $9\times9$ blurring kernels.}}
\end{figure}

\subsection{The Proposed GTF-HSR Framework}
\label{sec:gtfhsr_framework}
Thm.~\ref{Th: Kronecker Decomposition} and
Proposition~\ref{Prop:Kr-Rank} have proved that the $\mathbf{P}_1$ and
$\mathbf{P}_2$ in Asm.~\ref{asum: Kronecker separable
	assumption} do not necessarily exist. What exists instead is the KD
of spatial-degradation matrix $\mathbf{D}\in \mathbb{R}^{M_1M_2\times
	m_1m_2}$, which means that there exist collections
$\left\{\mathbf{P}_1^{(r)}\right\}_{r=1}^{\mathsf{kr}_{\mathbf{D}}}\subseteq\mathbb{R}^{m_1\times
	M_1}$ and
$\left\{\mathbf{P}_2^{(r)}\right\}_{r=1}^{\mathsf{kr}_{\mathbf{D}}}\subseteq\mathbb{R}^{m_2\times
	M_2}$ such that
\begin{equation}
	\mathbf{D}=\sum_{r = 1}^{\mathsf{kr}_{\mathbf{D}}}  (\mathbf{P}_2^{(r)}\otimes \mathbf{P}_1^{(r)})^T.
	\label{eq:generalizedD}
\end{equation}
Incorporating this expansion into \eqref{eq: Matrix Formulation}, we have
\begin{equation}
	\begin{aligned}
		\mathbf{X}_{[3]}& =\mathbf{Z}_{[3]}\mathbf{D}\\
		&=\mathbf{Z}_{[3]}\sum_{r = 1}^{\mathsf{kr}_{\mathbf{D}}}  (\mathbf{P}_2^{(r)}\otimes \mathbf{P}_1^{(r)})^T\\
		&=\sum_{r = 1}^{\mathsf{kr}_{\mathbf{D}}}\mathbf{Z}_{[3]}  (\mathbf{P}_2^{(r)}\otimes \mathbf{P}_1^{(r)})^T\\
		&=\sum_{r = 1}^{\mathsf{kr}_{\mathbf{D}}}\left(\mathcal{Z}\times_1\mathbf{P}_1^{(r)}\times_2\mathbf{P}_2^{(r)}\right)_{[3]}.
	\end{aligned}
\end{equation}
Thus, we formulate a generalized version of TF-HSR,
which we call GTF-HSR:
\begin{definition}[GTF-HSR]
	{\it Given HSI $\mathcal{X}\in \mathbb{R} ^{m_1\times m_2\times S}$,
		MSI $\mathcal{Y}\in \mathbb{R} ^{M_1\times M_2\times s}$,
		spectral-degradation matrix $\mathbf{R}\in \mathbb{R}^{s\times S}$,
		and spatial-degradation matrices
		$\left\{\mathbf{P}_1^{(r)}\right\}_{r=1}^{\mathsf{kr}_{\mathbf{D}}}\subseteq\mathbb{R}^{m_1\times
			M_1}$ and
		$\left\{\mathbf{P}_2^{(r)}\right\}_{r=1}^{\mathsf{kr}_{\mathbf{D}}}\subseteq\mathbb{R}^{m_2\times
			M_2}$ with $m_1< M_1$, $m_2< M_2$, and $s< S$, the GTF-HSR problem
		seeks the most appropriate SRI $\mathcal{Z}\in \mathbb{R}
		^{M_1\times M_2\times S}$ such that
		\begin{equation}
			\begin{aligned}
				\mathcal{X} & =\sum_{r = 1}^{\mathsf{kr}_{\mathbf{D}}}  \mathcal{Z}\times_1\mathbf{P}_1^{(r)}\times_2\mathbf{P}_2^{(r)}, \\
				\mathcal{Y} & =\mathcal{Z}\times_3\mathbf{R}
			\end{aligned}
			\label{eq:Reformulation}
	\end{equation}}
	\label{Def: Tensor ReFormulation}
\end{definition}
The proposed GTF-HSR differs from the existing TF-HSR approaches in
that it extends the modeling of spatial degradation from SRI to HSI
into a summation form. Sec.~\ref{sec: Feasibility of TeF} guarantees
that such extension enables the GTF-HSR to accurately capture the real
spatial-degradation process. Strictly speaking, GTF-HSR is equivalent
to MF-HSR in the sense that any degradation process modeled by
\eqref{eq: Matrix Formulation} can also be modeled by
\eqref{eq:Reformulation}, and vice versa. However, such equivalence is
not guaranteed between the TF-HSR and MF-HSR due to the failure of
Asm.~\ref{asum: Kronecker separable assumption} to hold for general
spatial-blurring of rank greater than unity.
\subsection{Recoverability in GTF-HSR}
\label{sec:recoverability}
Because of the ill-posedness of HSR, a solution satisfying \eqref{eq:
	Matrix Formulation}, \eqref{eq: Tensor Formulation}, or the proposed
\eqref{eq:Reformulation} does not necessarily recover the desired
original SRI $\mathcal{Z}$. Thus, recoverability---i.e., the
conditions for the solution to MF-HSR, TF-HSR, or GTF-HSR to surely
obtain the ground-truth SRI $\mathcal{Z}$---plays a pivotal role in
HSR.  Consequently, although MF-HSR has been quite successful from an
algorithmic perspective, one of the key motivations for proposing
GTF-HSR is that the algebraic properties of tensors facilitate the
establishing of recoverability conditions.  That said, previous
analyses (e.g., \cite{KFS2018,PUC2020,DFH2021}) establishing
recoverability of TF-HSR no longer apply due to the generalization of
\eqref{asum: Kronecker separable assumption} as
\eqref{eq:generalizedD}.  Thus, we present a new recoverability
analysis tailored to the proposed GTF-HSR based on the block-term
decomposition outlined in \cite{Del2008}:
\begin{theorem}
	\textit{Suppose the SRI $\mathcal{Z}\in\mathbb{R}^{M_1\times
			M_2\times S}$, HSI $\mathcal{X}\in \mathbb{R} ^{m_1\times
			m_2\times S}$, and MSI $\mathcal{Y}\in \mathbb{R} ^{M_1\times
			M_2\times s}$ satisfy relationship \eqref{eq: Matrix
			Formulation}. Suppose further that the Tucker decomposition of
		$\mathcal{Z}$ and the KD of spatial-degradation matrix
		$\mathbf{D}\in\mathbb{R}^{M_1M_2\times m_1m_2}$ are 
		\begin{align}
			\mathcal{Z}&=\mathcal{G}\times_1 \mathbf{U}_1\times_2 \mathbf{U}_2\times_3 \mathbf{U}_3,\\
			\mathbf{D}&=\sum_{r = 1}^{\mathsf{kr}_{\mathbf{D}}}  (\mathbf{P}_2^{(r)}\otimes \mathbf{P}_1^{(r)})^T ,
		\end{align}
		where $\mathcal{G}\in \mathbb{R}^{L_1\times L_2\times C}$ is drawn from an absolutely continuous distribution; and
		$\mathbf{U}_1$, $\mathbf{U}_2$, and $\mathbf{U}_3$ have
		full column rank. Then, if it is true that
		\begin{equation}
			\begin{aligned}
				&L_1\leq L_2C, \,L_2\leq L_1C,\, S\geq 3\\
				& \operatorname{rank}\left\{\left[\mathbf{P}_1^{(1)}\mathbf{U}_1,\cdots,\mathbf{P}_1^{(\mathsf{kr}_{\mathbf{D}})}\mathbf{U}_1\right]\right\}=L_1\mathsf{kr}_{\mathbf{D}},\\
				& \operatorname{rank}\left\{\left[\mathbf{P}_2^{(1)}\mathbf{U}_2,\cdots,\mathbf{P}_2^{(\mathsf{kr}_{\mathbf{D}})}\mathbf{U}_2\right]\right\}=L_2\mathsf{kr}_{\mathbf{D}},\\
				& \operatorname{rank}\left\{\mathbf{Y}_{[1]}\right\}=L_1,                                                                                                                                                                                                       \\
				& \operatorname{rank}\left\{\mathbf{Y}_{[2]}\right\}=L_2,
			\end{aligned}
			\label{eq:Rank Conditions 2}
		\end{equation}
		any solution of Tucker-rank at most $(L_1,L_2,S)$ to GTF-HSR
		recovers SRI $\mathcal{Z}$ with probability 1.}
	\label{Th:Recovery Guarantee}
\end{theorem}
A proof of Thm.~\ref{Th:Recovery Guarantee} built primarily on the
properties of the block-term decomposition described in \cite{Del2008}
can be found in the supplemental material. We note that, when
$\mathsf{kr}_{\mathbf{D}}=1$, this theorem addresses the special case
of the recoverability of TF-HSR. The main suppositions in this case
include the rank constraints on the spatial dimensions of both the SRI
and HSI images, as well as the genericity of the Tucker factors, which
are common conditions that have been considered before
\cite{KFS2018,PUC2020,DFH2021}.  Thus, from this perspective,
Thm.~\ref{Th:Recovery Guarantee} does not rely on unrealistic
assumptions.  More importantly, however, when
$\mathsf{kr}_{\mathbf{D}}>1$, Thm.~\ref{Th:Recovery Guarantee}
addresses the recoverability of GTF-HSR, which is unprecedented in the
literature.

While Thm.~\ref{Th:Recovery Guarantee} guarantees that the proposed
GTF-HSR can exactly recover SRI $\mathcal{Z}$, it can be shown that
TF-HSR cannot recover $\mathcal{Z}$ under the same conditions.
That is, although Sec.~\ref{sec: Feasibility of TeF} has already established
that, when $\mathsf{kr}_{\mathbf{D}}>1$, it is impossible for
Asm.~\ref{asum: Kronecker separable assumption} to hold,
since HSR is a severely ill-posed problem, one might think that,
as long as
\begin{equation}
	\mathsf{col}\left\{\mathbf{D}- (\mathbf{P}_2\otimes \mathbf{P}_1)^T\right\}\subseteq\mathsf{null}\left\{\mathbf{Z}_{[3]}\right\} ,
	\label{eq:Estimation Criterion}
\end{equation}
there might still be hope of recovering $\mathcal{Z}$ via TF-HSR.  In
other words, from the perspective of fusion performance, it might not
be strictly necessary to precisely model the real spatial-degradation
process. However, the following corollary, a simple
consequence of Thm.~\ref{Th:Recovery Guarantee}, indicates that this
is not the case.
\begin{corollary}
	\textit{Under the conditions of Thm.~\ref{Th:Recovery Guarantee},
		if it is true that
		\begin{equation}
			\begin{aligned}
				&L_1\leq L_2C, \,L_2\leq L_1C,\, S\geq 3\\
				& \operatorname{rank}\left\{\left[\mathbf{P}_1^{(1)}\mathbf{U}_1,\cdots,\mathbf{P}_1^{(\mathsf{kr}_{\mathbf{D}})}\mathbf{U}_1\right]\right\}=L_1\mathsf{kr}_{\mathbf{D}},\\
				& \operatorname{rank}\left\{\left[\mathbf{P}_2^{(1)}\mathbf{U}_2,\cdots,\mathbf{P}_2^{(\mathsf{kr}_{\mathbf{D}})}\mathbf{U}_2\right]\right\}=L_2\mathsf{kr}_{\mathbf{D}},
			\end{aligned}
		\end{equation}
		then any solution to TF-HSR recovers SRI $\mathcal{Z}$ with
		probability 0 when $\mathsf{kr}_{\mathbf{D}}>1$.}
	\label{Cro:TeF Inability}
\end{corollary}
A proof of Cor.~\ref{Cro:TeF Inability} can be found in the
supplemental material. We note that Cor.~\ref{Cro:TeF Inability}
indicates that, under a subset of the conditions of
Thm.~\ref{Th:Recovery Guarantee}, TF-HSR will fail to recover
$\mathcal{Z}$.  Indeed, even in those cases wherein the conditions in
Cor.~\ref{Cro:TeF Inability} do not hold and applying TF-HSR is
still theoretically possible, there exists no clear practical route to
obtaining appropriate $\mathbf{P}_1$ or $\mathbf{P}_2$ that satisfy
\eqref{eq:Estimation Criterion}.  Hence, the proposed GTF-HSR is at an
advantage over TF-HSR, both in theory and in practice.
\subsection{Blind HSR}
\label{sec:blind}
An additional key advantage of the proposed GTF-HSR is its ability to
facilitate blind HSR. As mentioned above, blind HSR, wherein the
spatial-degradation matrix is unknown, is more realistic in certain
settings but is an as-yet unresolved problem for TF-oriented methods,
due to the lack of a process for acquiring a pair of appropriate
$\mathbf{P}_1$ and $ \mathbf{P}_2$. Several strategies have been
proposed in an attempt to circumvent this TF-HSR limitation---for
example, \cite{KFS2018,PUC2020,DFH2021,BZX2021,DLF2020} propose
absorbing the spatial-degradation matrices into the factors to be
estimated, inevitably resulting in suboptimal performance due to the
resulting information loss. Alternatively, \cite{BZX2021,XWC2020a,CZH2022}
conduct trial-and-error estimation. {They directly utilize the rank-1 estimation of the blurring kernel to generate $\mathbf{P}_1$ and $
	\mathbf{P}_2$ to perform blind HSR. Yet, it is unclear if such scheme delivers the best estimation on the spatial degradation process. Moreover, there exists no means to determine if the spatial degradation process can be precisely modeled following TF-HSR.}

In contrast, GTF-HSR {confirms that precisely modeling the spatial degradation process following TF-HSR is impossible, and realizes this goal by invoking KD}.  That is, once we
obtain an estimate of $\mathbf{D}$ (using, e.g., appropriate methods
developed for the MF-HSR problem), we can perform KD on $\mathbf{D}$
to derive
$\left\{\mathbf{P}_1^{(r)}\right\}_{r=1}^{\mathsf{kr}_{\mathbf{D}}}\subseteq\mathbb{R}^{m_1\times
	M_1}$, and
$\left\{\mathbf{P}_2^{(r)}\right\}_{r=1}^{\mathsf{kr}_{\mathbf{D}}}\subseteq\mathbb{R}^{m_2\times
	M_2}$. Alternatively, we could develop methods to directly estimate
the sets of these matrices since GTF-HSR is equivalent to MF-HSR in
the sense of spatial-degradation modeling.

\section{A Group-Sparse Solution for GTF-HSR}
\label{sec:bgsgtfhsr}
While Thm.~\ref{Th:Recovery Guarantee} guarantees perfect SRI
recovery is possible within the GTF-HSR framework, it does not
actually indicate how one goes about effectuating the same.
Consequently, we now proceed to develop an algorithmic procedure to
solve GTF-HSR in the form of a factor-identification problem.  While
various tensor-decomposition frameworks could be used for this, we
adopt the Tucker decomposition.  That is, by applying the Tucker
decomposition of \eqref{eq:td} to \eqref{eq:Reformulation}, GTF-HSR
becomes the problem of estimating the most appropriate $\mathcal{G}\in
\mathbb{R}^{L_1\times L_2\times C}$,
$\mathbf{U}_{1}\in\mathbb{R}^{M_1\times L_1}$,
$\mathbf{U}_{2}\in\mathbb{R}^{M_2\times L_2}$, and
$\mathbf{U}_{3}\in\mathbb{R}^{S\times C}$ such that
\begin{equation}
	\begin{aligned}
		\mathcal{X} & =\sum_{r = 1}^{\mathsf{kr}_{\mathbf{D}}}  \mathcal{G}\times_1\mathbf{P}_1^{(r)}\mathbf{U}_1\times_2\mathbf{P}_2^{(r)}\mathbf{U}_2\times_3\mathbf{U}_3, \\
		\mathcal{Y} & =\mathcal{G}\times_1\mathbf{U}_1\times_2\mathbf{U}_2\times_3\mathbf{R}\mathbf{U}_3,
	\end{aligned}
	\label{eq:bare objective}
\end{equation}
given HSI $\mathcal{X}\in \mathbb{R} ^{m_1\times m_2\times S}$, MSI
$\mathcal{Y}\in \mathbb{R} ^{M_1\times M_2\times s}$, spatial
degradations
$\left\{\mathbf{P}_1^{(r)}\right\}_{r=1}^{\mathsf{kr}_{\mathbf{D}}}\subseteq\mathbb{R}^{m_1\times
	M_1}$ and
$\left\{\mathbf{P}_2^{(r)}\right\}_{r=1}^{\mathsf{kr}_{\mathbf{D}}}\subseteq\mathbb{R}^{m_2\times
	M_2}$, and spectral degradation
$\mathbf{R}\in\mathbb{R}^{s\times S}$.
\subsection{Tensor Blockwise Group Sparsity}
\begin{figure*}[t]
	\centering
	\setlength{\tabcolsep}{1cm}
	\begin{tabular}{m{1\figurewidth}m{0.4\figurewidth}}\includegraphics[width=1\figurewidth]{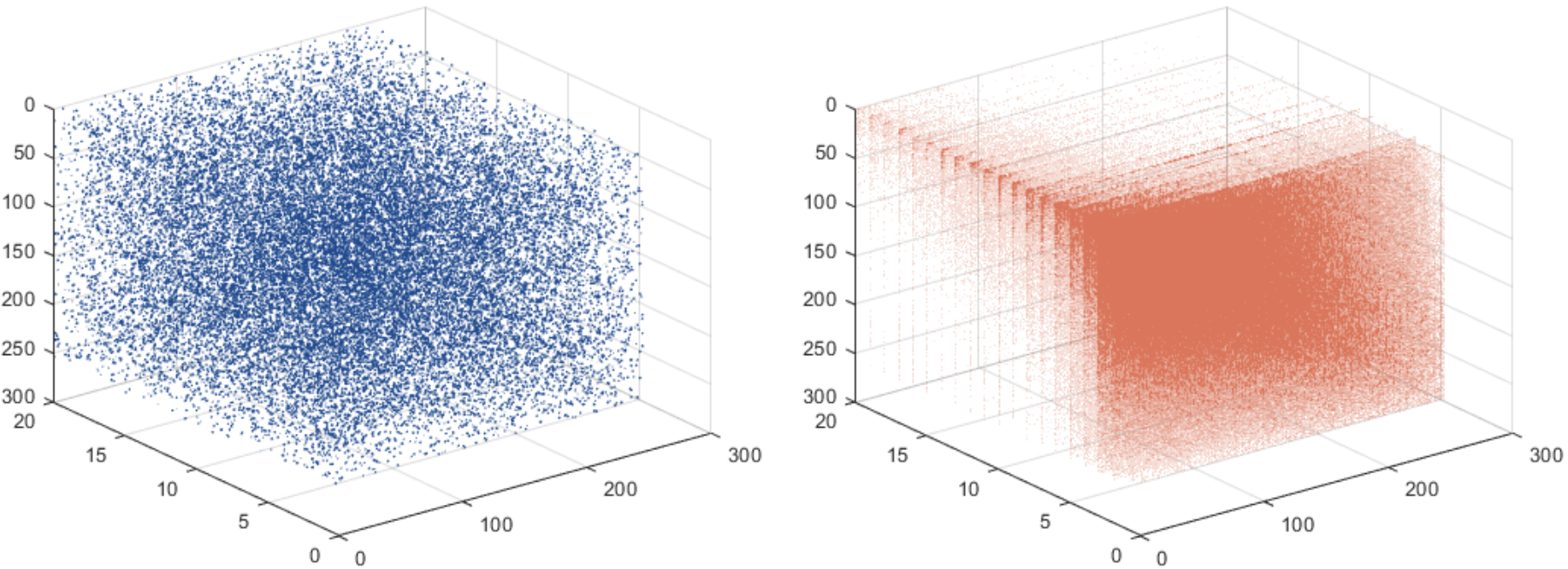}&
		\includegraphics[width=0.4\figurewidth]{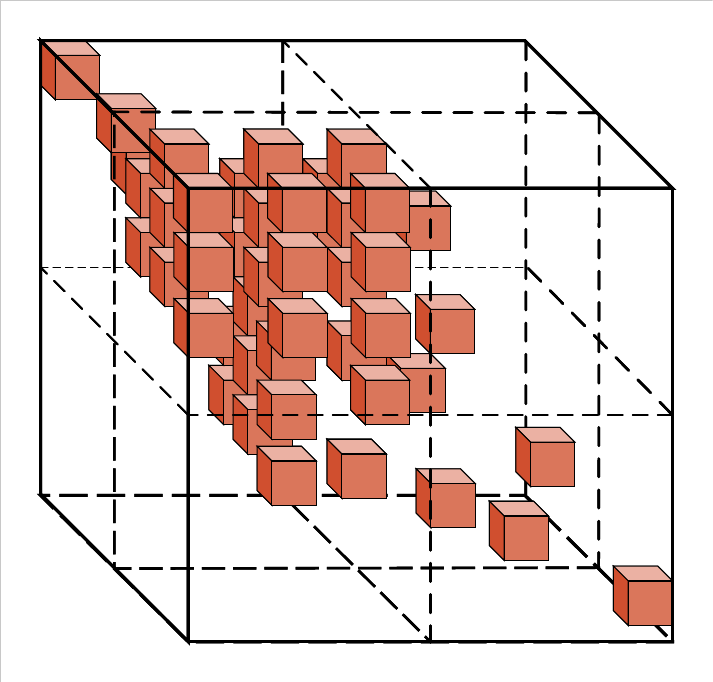}\\\multicolumn{1}{c}{(a)}&\multicolumn{1}{c}{(b)}
	\end{tabular}
	\caption{\label{fig:BGS}
		(a) A sparse random tensor (left) vs.\ the core tensor from
		the Tucker decomposition of a real HSI (right); the core tensor (right)
		exhibits a prominent structured-sparsity characteristic.
		(b) A BGS tensor wherein sparsity takes on a blockwise structure.}
\end{figure*}

While the ill-posedness of HSR is the greatest obstacle to obtaining
the ideal solution, under the Tucker decomposition, the SRI exhibits a
wealth of properties that can be exploited to design
regularizations that narrow down the solution set. Particularly,
sparsity of the core tensor $\mathcal{G}$ has been considered for HSR in
the past (e.g., \cite{DFL2017,LDF2018,DLF2020}). However, it has been
observed (e.g., \cite{ZZW2023}) that real-world data often exhibits
structured sparsity due to relationships contained within the data.
Indeed, such structured sparsity can be seen prominently in the core
tensor of a real HSI in Fig.~\ref{fig:BGS}(a).  Consequently, we adopt
blockwise group sparsity (BGS)---illustrated in
Fig.~\ref{fig:BGS}(b)---as a prior for solving GTF-HSR rather than the
simple sparsity used previously for HSR in
\cite{DFL2017,LDF2018,DLF2020}. Briefly, in BGS, the overall tensor is
divided into smaller subblocks under the supposition that only a
relative few of the subblocks contain nonzero samples; this is likely
to be true if the subblocks are of sufficiently small size.  We note
that, while BGS has been used in the past \cite{WGH2020,PMX2014}, {\cite{WGH2020} defines BGS for 4-order image cube cluster, which is difficult to be generalized towards an $N$-order tensor, including the $3$-order tensor in our case. The semi-algebraic \cite{PMX2014} is hard to deploy in our multi-sourced reconstruction problem. In response, we present a generalized BGS pattern for any $N$-order tensor achieved by an optimization-based algorithmic framework.}

\begin{figure}[bthp!]
	\centering
	\setlength{\tabcolsep}{1mm}
	\begin{tabular}{m{0.8\figurewidth}}\includegraphics[width=0.8\figurewidth]{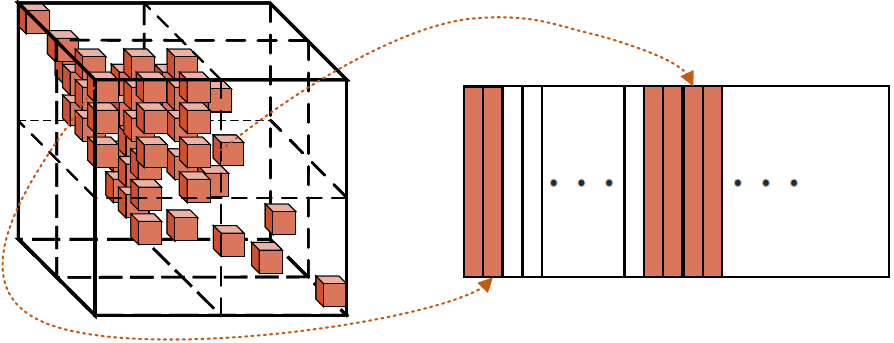}
	\end{tabular}
	\caption{\label{fig:B-unfolding}B-unfolding, reorganizing each
		tensor block into a single column of a matrix. }
\end{figure}

Our solution procedure effectively consists of imposing a BGS
constraint onto the core tensor and casting the restoration task as a
regularized optimization.  To this end, we first propose an unfolding
strategy appropriate for BGS which is defined below and also
illustrated in Fig.~\ref{fig:B-unfolding}.
\begin{definition}[Blockwise Unfolding (B-Unfolding)]
	\textit{For an $N$-order tensor $\mathcal{T}\in\mathbb{R}^{T_1\times T_2\times\cdots\times T_N}$ with $T_n=t_ns_n$, $n=1,2,\dots,N$, its Blockwise unfolding, denoted by $\mathbf{T}_{[\mathbf{t}]}$, is defined as
		\begin{multline}
			\mathbf{T}_{[\mathbf{t}]}=\mathtt{reshape}\Big(\mathtt{permute}\big(\mathtt{reshape}(\mathcal{T},\\
			[t_1,s_1,\dots,t_N,s_N],
			[1,3,\dots,2N-1,2,4,\cdots,2N])\big), \\
			\prod_{n=1}^N t_n,[\,]\Big)
		\end{multline}
		where $\mathbf{t}\triangleq[t_1,t_2,\dots,t_N].$}
	\label{Def:B-Unfolding}
\end{definition}
Effectively, the proposed B-unfolding reorganizes each subblock of
size $t_1\times t_2\times\cdots\times t_N$ into a single column of a
matrix.  B-unfolding thus allows the transformation a BGS constraint
on a tensor into a column-wise sparsity constraint on its B-unfolding.
We note that the latter has been well-studied extensively in prior
literature and incorporated into optimizations in the form of the
$\ell_{2,0}$-norm, the number of nonzero columns in a matrix.

Thus, our proposed BGS-based approach yields the optimization
\begin{equation}
	\begin{gathered}
		\min_{{\mathcal{G},\mathbf{U}_1,\mathbf{U}_2,\mathbf{U}_3}} \left\lVert \mathbf{G}_{[\mathbf{t}]}\right\rVert_{2,0}                                                \\
		\begin{aligned}
			\text{s.t.}\quad \mathcal{X} & =\sum_{r = 1}^{\mathsf{kr}_{\mathbf{D}}}  \mathcal{G}\times_1\mathbf{P}_1^{(r)}\mathbf{U}_1\times_2\mathbf{P}_2^{(r)}\mathbf{U}_2\times_3\mathbf{U}_3, \\
			\mathcal{Y}       & =\mathcal{G}\times_1\mathbf{U}_1\times_2\mathbf{U}_2\times_3\mathbf{R}\mathbf{U}_3.
		\end{aligned}
	\end{gathered}
	\label{eq:Regularized TeRF}
\end{equation}
One major benefit of using B-unfolding is that it enables the
exploration of the multi-linear structure of higher-order data through
a single compact norm. In spite of the consensus on the existence of
multi-linearity in higher-order tensors, previous studies (e.g.,
\cite{XZL2020,QZZ2022,CJL2021}) have resorted to imposing constraints in the
form of a summation or product of multiple norms, which inevitably
increases model complexity and poses additional challenges for
subsequent optimization procedures. In contrast, the proposed BGS
constraint can easily balance the multi-linear structure hidden in the
data through manipulating the parameter $\mathbf{t}$ while reducing
optimization burden, since it can be achieved by a single matrix norm.

That said, since the $\ell_{2,0}$-norm is discontinuous and nonconvex,
the optimization in
\eqref{eq:Regularized TeRF} is NP-hard \cite{ZZW2023}.
While a straightforward solution would be to adopt a
convex relaxation via the $\ell_{2,1}$-norm, this would be
intrinsically suboptimal since the $\ell_{2,1}$-norm is the convex
envelope of the $\ell_{2,0}$-norm.
Rather, we propose to use a nonconvex surrogate if the form of
the Laplace function {\cite{ZZW2023}}; i.e.,
\begin{equation}
	\begin{aligned}
		\left\lVert \mathbf{G}_{[\mathbf{t}]}\right\rVert_{2,0}&\approx\left\lVert \mathbf{G}_{[\mathbf{t}]}\right\rVert_{2,\gamma}
		\triangleq \sum_{i}\left[1-\exp\left(-\left\lVert\mathbf{g}_{[\mathbf{t}]}^{(i)}\right\rVert_2/\gamma\right)\right] ,
	\end{aligned}
\end{equation}
where $\gamma>0$ is a parameter to adjust the position of this
surrogate function, $\mathbf{g}_{[\mathbf{t}]}^{(i)}$ is column~$i$ of
matrix $\mathbf{G}_{[\mathbf{t}]}$, and we note that
$\lim_{\gamma\rightarrow 0^+}\left\lVert
\mathbf{G}_{[\mathbf{t}]}\right\rVert_{2,\gamma}=\left\lVert
\mathbf{G}_{[\mathbf{t}]}\right\rVert_{2,0}$.  Despite the fact that
nonconvex regularized tensor approaches for data restoration have been,
to some extent, studied before (e.g.,
\cite{XZM2017,CJL2021,YLL2022a,CGS2019}), these past efforts have
largely aimed at single-source restoration such as tensor completion,
deconvolution, and denoising. In the multi-source case considered in
this work, the optimization problem must be conducted at a much larger
scale due to the increased amount of data.  Such large-scale
optimization is ill-handled by existing nonconvex schemes, and,
consequently, nonconvex surrogates have not been used for tensor-based
HSR before now, to the best of our knowledge.

\subsection{The Proposed BGS-GTF-HSR Algorithm}
Using the Laplace surrogate, the HSR problem of \eqref{eq:Regularized TeRF}
is relaxed into
\begin{equation}
	\begin{gathered}
		\min_{{\mathcal{G},\mathbf{U}_1,\mathbf{U}_2,\mathbf{U}_3}} \left\lVert \mathbf{G}_{[\mathbf{t}]}\right\rVert_{2,\gamma}                                                \\
		\begin{aligned}
			\text{s.t.}\quad \mathcal{X} & =\sum_{r = 1}^{\mathsf{kr}_{\mathbf{D}}}  \mathcal{G}\times_1\mathbf{P}_1^{(r)}\mathbf{U}_1\times_2\mathbf{P}_2^{(r)}\mathbf{U}_2\times_3\mathbf{U}_3, \\
			\mathcal{Y}       & =\mathcal{G}\times_1\mathbf{U}_1\times_2\mathbf{U}_2\times_3\mathbf{R}\mathbf{U}_3.
		\end{aligned}
	\end{gathered}
	\label{eq:Relaxed TeRF}
\end{equation}
This is a problem of a nonconvex objective function with
multi-variable nonconvex constraints. Though alternating direction
method of multipliers (ADMM) has seen success \cite{XZM2017} in
dealing with such problems, due to a larger scale here, some of our
initial trials have indicated that directly applying ADMM to
\eqref{eq:Relaxed TeRF} may lead to a suboptimal solution and weak
convergence. As an alternative, we devise a two-stage framework, which
we refer to as BGS-GTF-HSR, to solve \eqref{eq:Relaxed TeRF} in a
divide-and-conquer manner: first, a subspace-identification problem
which identifies $\mathbf{U}_1$, $\mathbf{U}_2$, and $\mathbf{U}_3$,
followed by a coding problem to identify $\mathcal{G}$. We develop
BGS-GTF-HSR throughout the rest of this section.
\subsubsection{Subspace Identification}
We first seek a practical technique to identify $\mathbf{U}_{1}$,
$\mathbf{U}_2$, and $\mathbf{U}_3,$ the subspace factors.  Since the
observed MSI $\mathcal{Y}$ is degraded only spectrally, we derive
$\mathbf{U}_1$ and $\mathbf{U}_2$ primarily from $\mathcal{Y}$ while
including a few supplementary basis vectors from HSI
$\mathcal{X}$. That is, the spatial factors
$\mathbf{U}_{i}\in\mathbb{R}^{M_i\times L_i},\,i=1,2$ are partitioned
into
$\mathbf{U}_{i}=\begin{bmatrix}\mathbf{U}_{i}^\mathcal{Y}&\mathbf{U}_i^\mathcal{X}\end{bmatrix}$
where $\mathbf{U}_{i}^\mathcal{Y}\in\mathbb{R}^{M_i\times K_i}$ with
$K_i\gg L_i-K_i$. Through mode-$i$ unfolding, we have
\begin{equation}
	\mathbf{Y}_{[i]}=\begin{bmatrix}\mathbf{U}_{i}^\mathcal{Y}&\mathbf{U}_i^\mathcal{X}\end{bmatrix}\mathbf{G}_{[i]}\left(\mathbf{R}\mathbf{U}_3\otimes
	\mathbf{U}_{[3-i]}\right)^T .
\end{equation}
Now, suppose the columns of $\mathbf{U}_{i}^\mathcal{Y}$ span
$\mathsf{col}\left\{\mathbf{Y}_{[i]}\right\}$. We can then estimate
$\mathbf{U}_{i}^\mathcal{Y}$ via
\begin{equation}
	\left[\mathbf{U}_i^\mathcal{Y},\thicksim, \thicksim
	\right]=\mathtt{svd}\left(\mathbf{Y}_{[i]},K_i\right).
\end{equation}
We extract $\mathbf{U}_i^\mathcal{X}$ from HSI $\mathcal{X}$.  Again by
mode-$i$ unfolding, we have
\begin{equation}
	\mathbf{X}_{[i]}=\begin{bmatrix}\mathbf{P}_i^{(1)} & \cdots & \mathbf{P}_i^{(\mathsf{kr}_{\mathbf{D}})}\end{bmatrix}\left(\mathbf{I}_{\mathsf{kr}_{\mathbf{D}}}\otimes \begin{bmatrix}\mathbf{U}_{i}^\mathcal{Y}&\mathbf{U}_i^\mathcal{X}\end{bmatrix}\right)\mathbf{A}_i
\end{equation}
where
\begin{equation}
	\mathbf{A}_i=\begin{bmatrix}
		\mathbf{G}_{[i]}\left(\mathbf{U}_3\otimes\mathbf{P}_{3-i}^{(1)}\mathbf{U}_{3-i}\right)\\
		\mathbf{G}_{[i]}\left(\mathbf{U}_3\otimes\mathbf{P}_{3-i}^{(2)}\mathbf{U}_{3-i}\right)\\
		\vdots\\
		\mathbf{G}_{[i]}\left(\mathbf{U}_3\otimes\mathbf{P}_{3-i}^{(\mathsf{kr}_{\mathbf{D}})}\mathbf{U}_{3-i}\right)
	\end{bmatrix}.
\end{equation}
\begin{algorithm}[t]
	\renewcommand{\algorithmicrequire}{\textbf{Input:}}
	\renewcommand{\algorithmicensure}{\textbf{Output:}}
	\caption{BGS-GTF-HSR: Subspace Identification}
	\label{alg:BGS-TeRF Part 1}
	\begin{algorithmic}[1]
		\REQUIRE The observed HSI $\mathcal{X}\in\mathbb{R}^{m_1\times m_2\times S}$, MSI $\mathcal{Y}\in\mathbb{R}^{M_1\times M_2\times s}$, spatial-degradation matrices $\left\{\mathbf{P}_1^{(r)}\right\}_{r=1}^{\mathsf{kr}_{\mathbf{D}}}\subseteq\mathbb{R}^{m_1\times M_1}$ and $\left\{\mathbf{P}_2^{(r)}\right\}_{r=1}^{\mathsf{kr}_{\mathbf{D}}}\subseteq\mathbb{R}^{m_2\times M_2}$, spectral-degradation matrix $\mathbf{R}\in\mathbb{R}^{s\times S}$, $\mu,\,\epsilon,\,K_1,\,K_2$
		\ENSURE Subspace factors $\mathbf{U}_1\in\mathbb{R}^{M_1\times L_1}$, $\mathbf{U}_2\in\mathbb{R}^{M_2\times L_2}$, $\mathbf{U}_3\in\mathbb{R}^{S\times C}$
		\FOR{$i=1:2$}
		\STATE $\left[\mathbf{U}_i^\mathcal{Y},\thicksim, \thicksim \right]=\mathtt{svd}\left(\mathbf{Y}_{[i]},K_i\right)$
		\STATE initialize $\rho$, $\rho_{max}$, $\nu>1$, choose random $\mathbf{A}_i$
		\STATE set $\mathbf{B}_i=\mathbf{A}_i$, $\mathbf{U}_i^\mathcal{X} = \mathbf{0}$, $\mathbf{M}_i = \mathbf{0}$
		\WHILE{$\frac{\left\Vert\mathbf{U}_i^\mathcal{X}-\mathbf{U}_i^{\mathcal{X}\text{(previous)}}\right\Vert_F}{\left\Vert\mathbf{U}_i^{\mathcal{X}\text{(previous)}}\right\Vert_F}\geq\epsilon$}
		\STATE $\widetilde{\mathbf{X}}_{[i]}\leftarrow\mathbf{X}_{[i]}-\sum_{r=1}^{\mathsf{kr}_{\mathbf{D}}}\mathbf{P}_i^{(r)}\mathbf{U}_i^\mathcal{Y}
		\mathbf{B}_i^{(r)\mathcal{Y}}$
		\STATE Update $\mathbf{U}_i^\mathcal{X}$ by solving $\nabla \mathsf{L}_{\mathbf{U}_i^\mathcal{X}}=0$  via conjugate gradient (CG) \cite{GV2013} \label{sp:cg}
		\STATE $\mathbf{D}_i\leftarrow\begin{bmatrix}\mathbf{P}_i^{(1)} & \cdots & \mathbf{P}_i^{(\mathsf{kr}_{\mathbf{D}})}\end{bmatrix}\left(\mathbf{I}_{\mathsf{kr}_{\mathbf{D}}}\otimes \begin{bmatrix}\mathbf{U}_i^\mathcal{Y} & \mathbf{U}_i^\mathcal{X}\end{bmatrix}\right)$
		\STATE  $\mathbf{B}_i\leftarrow\left(\mathbf{D}_i^T\mathbf{D}_i+\frac{\rho}{2}\mathbf{I}_{L_i\mathsf{kr}_{\mathbf{D}}}\right)^{-1}\Big(\mathbf{D}_i^T\mathbf{X}_{[i]}+$\\\hfill$\frac{\rho}{2}\left(\mathbf{A}_i+\frac{\mathbf{M}_i}{\rho}\right)\Big)$ \label{sp:inversion}
		\STATE  $\mathbf{A}_i\leftarrow \mathtt{soft}\left(\mathbf{B}_i-\frac{\mathbf{M}_i}{\rho},\frac{\mu}{\rho}\right)$
		\STATE  $\mathbf{M}_i\leftarrow\mathbf{M}_i+\rho(\mathbf{A}_i-\mathbf{B}_i)$
		\STATE  $\rho\leftarrow\min\left\{\nu\rho,\rho_{\text{max}}\right\}$
		\ENDWHILE
		\STATE   $\left[\mathbf{U}_3,\thicksim, \thicksim \right]=\mathtt{svd}\left(\mathbf{X}_{[3]},C\right)$
		\ENDFOR
	\end{algorithmic}
\end{algorithm}

The degradation matrices
$\left\{\mathbf{P}_i^{(r)}\right\}_{r=1}^{\mathsf{kr}_{\mathbf{D}}}$
hinder direct extraction of $\mathbf{U}_i^\mathcal{X}$ from
$\mathsf{col}\left\{\mathbf{X}_{[i]}\right\}$; therefore, we propose
to impose a sparse constraint on $\mathbf{A}$ to estimate
$\mathbf{U}_i^\mathcal{X}$ via sparse dictionary learning,
\begin{gather}
	\begin{aligned}
		&\min_{\mathbf{U}_i^\mathcal{X},\mathbf{A}_i}\,\Big\lVert \mathbf{X}_{[i]}-\\
		&\begin{bmatrix}\mathbf{P}_i^{(1)} & \cdots & \mathbf{P}_i^{(\mathsf{kr}_{\mathbf{D}})}\end{bmatrix}\left(\mathbf{I}_{\mathsf{kr}_{\mathbf{D}}}\otimes \begin{bmatrix}\mathbf{U}_i^\mathcal{Y} & \mathbf{U}_i^\mathcal{X}\end{bmatrix}\right)\mathbf{B}_{i}\Big\rVert_F^2+\mu\left\lVert \mathbf{A}_i\right\rVert_1 \\
	\end{aligned} \nonumber \\
	\text{s.t.} \quad \mathbf{A}_i = \mathbf{B}_i ,
	\label{eq:SubIden}
\end{gather}
where we have introduced auxiliary variables $\mathbf{B}_i$.  The
optimization \eqref{eq:SubIden} is solved via ADMM; this is described
in detail in the supplemental material.

While we could estimate 
spectral subspace factor $\mathbf{U}_3$ in a similar fashion as is done above
for the spatial factors, we instead adopt the simpler approach of
extracting $\mathbf{U}_3$ directly from the HSI $\mathcal{X}$ as it is
subject to only spatial degradation:
\begin{equation}
	\left[\mathbf{U}_3,\thicksim, \thicksim \right]=\mathtt{svd}\left(\mathbf{X}_{[3]},C\right).
\end{equation}

The subspace-identification procedure is presented as
Alg.~\ref{alg:BGS-TeRF Part 1} wherein
$\mathbf{B}_i^{(r)\mathcal{Y}}\in\mathbb{R}^{K_i\times m_{3-i}S}$
and
$\mathbf{B}_i^{(r)\mathcal{X}}\in\mathbb{R}^{(L_i-K_i)\times m_{3-i}S}$
are submatrices partitioned from $\mathbf{B}_i$ as
$  \mathbf{B}_i= \begin{bmatrix}
	(\mathbf{B}_i^{(1)\mathcal{Y}})^T &
	(\mathbf{B}_i^{(1)\mathcal{X}})^T &
	\cdots &
	(\mathbf{B}_i^{(\mathsf{kr}_{\mathbf{D}})\mathcal{Y}})^T &
	(\mathbf{B}_i^{(\mathsf{kr}_{\mathbf{D}})\mathcal{X}})^T
\end{bmatrix}^T ;$
and, in step~\ref{sp:cg},
\begin{multline}
	\nabla \mathsf{L}_{\mathbf{U}_i^\mathcal{X}}\triangleq\sum_{r_1=1}^{\mathsf{kr}_{\mathbf{D}}}\sum_{r_2=1}^{\mathsf{kr}_{\mathbf{D}}}\left(\mathbf{P}_i^{(r_1)}\right)^T\mathbf{P}_i^{(r_2)}\mathbf{U}_i^\mathcal{X}\mathbf{B}_i^{(r_2)\cal X}\left(\mathbf{B}_i^{(r_1)\cal X}\right)^T-\\\sum_{r=1}^{\mathsf{kr}_{\mathbf{D}}}\left(\mathbf{P}_i^{(r)}\right)^T\widetilde{\mathbf{X}}_{[i]}\left(\mathbf{B}_i^{(r)\cal X}\right)^T.
\end{multline}
\subsubsection{BGS Coding}
With the factor matrices $\mathbf{U}_1$, $\mathbf{U}_2$, and
$\mathbf{U}_3$ being determined by Alg.~\ref{alg:BGS-TeRF Part 1}, the
sole remaining task in problem \eqref{eq:Relaxed TeRF} is to determine
$\mathcal{G}$. In doing so,
we introduce auxiliary variables $\mathcal{G}_r$ such that
\eqref{eq:Relaxed TeRF} becomes
\begin{equation}
	\begin{gathered}  
		\min_{{\left\{\mathcal{G}_r\right\}_{r=1}^{\mathsf{kr}_{\mathbf{D}}}
			},\mathcal{G},\hat{\mathbf{G}}} \left\lVert
		\hat{\mathbf{G}}\right\rVert_{2,\gamma} \\
		\begin{aligned}
			\text{s.t.}\quad \mathcal{X} & =\sum_{r =
				1}^{\mathsf{kr}_{\mathbf{D}}}
			\mathcal{G}_r\times_1\mathbf{P}_1^{(r)}\mathbf{U}_1\times_2\mathbf{P}_2^{(r)}\mathbf{U}_2\times_3\mathbf{U}_3,
			\\ \mathcal{Y} &
			=\mathcal{G}\times_1\mathbf{U}_1\times_2\mathbf{U}_2\times_3\mathbf{R}\mathbf{U}_3,
			\\ \hat{\mathbf{G}} & =\mathbf{G}_{[\mathbf{t}]}, \\ \mathcal{G}
			& =\mathcal{G}_r, \,\, r=1,2,\dots,\mathsf{kr}_{\mathbf{D}}.
		\end{aligned}
	\end{gathered}
	\label{eq:BGS subproblem}
\end{equation}
The optimization is carried out via ADMM; thus, for brevity, we
present only the resulting algorithm here as Alg.~\ref{alg:BGS-TeRF
	Part 2}, relegating the complete details to the supplemental
material.  We do note, however, that, we solve the subproblem in
step~\ref{sp:gai}, which is critical to achieving the desired BGS
pattern, via the recently developed generalized accelerating iterative
(GAI) \cite{ZZW2023}.

\begin{algorithm}[t]
	\renewcommand{\algorithmicrequire}{\textbf{Input:}}
	\renewcommand{\algorithmicensure}{\textbf{Output:}}
	\caption{BGS-GTF-HSR: BGS Coding}
	\label{alg:BGS-TeRF Part 2}
	\begin{algorithmic}[1]
		\REQUIRE The observed HSI $\mathcal{X}\in\mathbb{R}^{m_1\times
			m_2\times S}$, MSI $\mathcal{Y}\in\mathbb{R}^{M_1\times
			M_2\times s}$, spatial-degradation matrices
		$\left\{\mathbf{P}_1^{(r)}\right\}_{r=1}^{\mathsf{kr}_{\mathbf{D}}}\subseteq\mathbb{R}^{m_1\times
			M_1}$ and $\left\{\mathbf{P}_2^{(r)}\right\}_{r=1}^{\mathsf{kr}_{\mathbf{D}}}\subseteq\mathbb{R}^{m_2\times
			M_2}$, spectral-degradation matrix
		$\mathbf{R}\in\mathbb{R}^{s\times S}$, subspace factors
		$\mathbf{U}_1\in\mathbb{R}^{M_1\times
			L_1}$, $\mathbf{U}_2\in\mathbb{R}^{M_2\times
			L_2}$, $\mathbf{U}_3\in\mathbb{R}^{S\times C}$,
		$\lambda,\,\epsilon$
		\ENSURE Core tensor
		$\mathcal{G}\in\mathbb{R}^{L_1\times L_2\times C}$
		\STATE initialize $\rho$, $\rho_{\text{max}}$, $\nu>1$, choose random
		$\mathbf{A}_i$
		\STATE set
		$\mathcal{G}$,
		$\left\{\mathcal{G}_r\right\}_{r=1}^{\mathsf{kr}_{\mathbf{D}}}$,
		$\mathcal{P}^\mathcal{X}$,
		$\mathcal{P}^\mathcal{Y}$,
		$\mathbf{W}$,
		$\left\{\mathcal{P}_r\right\}_{r=1}^{\mathsf{kr}_{\mathbf{D}}}$
		to $\mathbf{0}$
		\WHILE{$\frac{\left\Vert\mathcal{G}-\mathcal{G}^{\text{(previous)}}\right\Vert_F}{\left\Vert\mathcal{G}^{\text{(previous)}}\right\Vert_F}\geq\epsilon$}
		\FOR{$r=1:\mathsf{kr}_{\mathbf{D}}$}
		\STATE
		$\mathcal{H}\leftarrow\mathcal{X}+\frac{\mathcal{P}^{\cal{X}}}{\rho}-$\\
		\hfill$\sum_{r^*\neq
			r}
		\mathcal{G}_{r^*}\times_1\mathbf{P}_1^{(r^*)}\mathbf{U}_1\times_2\mathbf{P}_2^{(r^*)}\mathbf{U}_2\times_3\mathbf{U}_3$
		\STATE
		$\mathbf{Q}_1\leftarrow\mathbf{P}_1^{(r)}\mathbf{U}_1$,
		$\mathbf{Q}_2\leftarrow\mathbf{P}_2^{(r)}\mathbf{U}_2$,
		$\mathbf{Q}_3\leftarrow\mathbf{U}_3$
		\STATE
		$\mathcal{K}\leftarrow\mathcal{G}+\frac{\mathcal{P}_r}{\rho}$
		\FOR{$n=1:3$}
		\STATE
		$\left[\mathbf{V}_n,\sqrt{\mathbf{\Sigma}_n},
		\thicksim\right]=\mathtt{svd}\left(\mathbf{Q}_{n}^T\right)$
		\ENDFOR
		\STATE
		$\mathcal{T}\leftarrow\mathcal{H}\times_1\mathbf{Q}_1^T\times_2\mathbf{Q}_2^T\times_3\mathbf{Q}_3^T+\mathcal{K}$
		\STATE
		$\mathcal{T}'\leftarrow\mathcal{T}\times_1\mathbf{V}_1^T\times_2\mathbf{V}_2^T\times_3\mathbf{V}_3^T$
		\STATE
		$\mathsf{Vec}\left(\mathcal{T}''\right)\leftarrow\left(\mathbf{\Sigma}_3\otimes\mathbf{\Sigma}_2\otimes\mathbf{\Sigma}_1+\mathbf{I}_{L_1L_2C}\right)^{-1}\mathsf{Vec}\left(\mathcal{T}'\right)$
		\STATE
		$\mathcal{G}_r\leftarrow\mathcal{T}{''}\times_1\mathbf{V}_1\times_2\mathbf{V}_2\times_3\mathbf{V}_3$
		\ENDFOR
		\STATE
		$\hat{\mathbf{G}}\leftarrow\arg\min_{\hat{\mathbf{G}}}\,\frac{\rho}{2}\left\Vert\hat{\mathbf{G}}-\mathbf{G}_{[\mathbf{t}]}+\frac{\mathbf{W}}{\rho}\right\Vert_F^2+\left\lVert
		\hat{\mathbf{G}}\right\rVert_{2,\gamma}$ \label{sp:gai}
		\STATE
		$\mathcal{H}\leftarrow\mathcal{Y}+\frac{\mathcal{P}^{\cal{Y}}}{\rho}$,
		$\mathbf{Q}_1\leftarrow\mathbf{U}_1$,
		$\mathbf{Q}_2\leftarrow\mathbf{U}_2$,
		$\mathbf{Q}_3\leftarrow\mathbf{R}\mathbf{U}_3$
		\STATE
		$\mathcal{K}\leftarrow\frac{\left(\mathcal{G}^{\mathbf{W}}+\sum_{r=1}^{\mathsf{kr}_{\mathbf{D}}}\mathcal{G}_r-\frac{\mathcal{P}_r}{\rho}\right)}{\mathsf{kr}_{\mathbf{D}}+1}$,
		$\tau\leftarrow
		\mathsf{kr}_{\mathbf{D}}+1$
		\FOR{$n=1:3$}
		\STATE $\left[\mathbf{V}_n,\sqrt{\mathbf{\Sigma}_n},
		\thicksim\right]=\mathtt{svd}\left(\mathbf{Q}_{n}^T\right)$
		\ENDFOR
		\STATE
		$\mathcal{T}\leftarrow\mathcal{H}\times_1\mathbf{Q}_1^T\times_2\mathbf{Q}_2^T\times_3\mathbf{Q}_3^T+\tau\mathcal{K}$
		\STATE
		$\mathcal{T}'\leftarrow\mathcal{T}\times_1\mathbf{V}_1^T\times_2\mathbf{V}_2^T\times_3\mathbf{V}_3^T$
		\STATE
		$\mathsf{Vec}\left(\mathcal{T}''\right)\leftarrow\left(\mathbf{\Sigma}_3\otimes\mathbf{\Sigma}_2\otimes\mathbf{\Sigma}_1+\tau\mathbf{I}_{L_1L_2C}\right)^{-1}\mathsf{Vec}\left(\mathcal{T}'\right)$
		\STATE
		$\mathcal{G}\leftarrow\mathcal{T}{''}\times_1\mathbf{V}_1\times_2\mathbf{V}_2\times_3\mathbf{V}_3$
		\ENDWHILE
	\end{algorithmic}
\end{algorithm}
\subsubsection{Complexity and Convergence}
In Alg.~\ref{alg:BGS-TeRF Part 1}, the main complexity lies in the CG
iterations in step~\ref{sp:cg} and the matrix inversion in
step~\ref{sp:inversion}. In the CG iterations, the primary
computational burden is the multiplication of the system matrices with
factor matrix, whose complexity is
$\mathcal{O}(M_i(L_i-K_i)+M_i(L_i-K_i)^2)$, $i=1,2$. In
step~\ref{sp:inversion}, the matrix inversion has complexity
$\mathcal{O}(\mathsf{kr}_{\mathbf{D}}^3L_i^3).$ Thus the total
complexity of Alg.~\ref{alg:BGS-TeRF Part 1} is
$\sum_{i=1}^{2}[\mathcal{O}(N_{\text{CG}}(M_i(L_i-K_i)+M_i(L_i-K_i)^2))+\mathcal{O}(\mathsf{kr}_{\mathbf{D}}^3L_i^3)]$,
where $N_{\text{CG}}$ denotes the number of CG iterations.

In Alg.~\ref{alg:BGS-TeRF Part 2}, the complexity centers mostly on
the updating of $\mathcal{G}$,
$\left\{\mathcal{G}_r\right\}_{r=1}^{\mathsf{kr}_{\mathbf{D}}}$, and
$\hat{\mathbf{G}}$. Both the updating of $\mathcal{G}$ and
$\left\{\mathcal{G}_r\right\}_{r=1}^{\mathsf{kr}_{\mathbf{D}}}$ costs
the same complexity, $\mathcal{O}(L_1^2L_2C+L_1L_2^2C+L_1L_2C^2)$,
while the complexity of performing step~\ref{sp:gai} via GAI is
$\mathcal{O}(\prod_{n=1}^{3}s_n)$ where $\mathbf{s}=\begin{bmatrix}s_1&s_2&s_3\end{bmatrix}$ is
defined in Def.~\ref{Def:B-Unfolding}. As such, the whole
complexity of Alg.~\ref{alg:BGS-TeRF Part 2} is
$\mathcal{O}(\mathsf{kr}_{\mathbf{D}}(L_1^2L_2C+L_1L_2^2C+L_1L_2C^2))+\mathcal{O}(N_{\text{GAI}}\prod_{n=1}^{3}s_n)$,
where $N_{\text{GAI}}$ is the number of GAI iterations. {Besides, we'd like to note that the matrices in step 13 and 24 requiring inversion are diagonal. Thus their inversion can be calculated by element-wise inversion on their diagonals. And the subsequent multiplication can also be done element-wisely. These two steps, though involving the inversion on large-scale matrices, do not add complexity to the overall algorithm.}

Although ADMM has been widely deployed (e.g., \cite{LDF2018,DL2019,8698334}), its convergence has been
confirmed for only 2-block convex problems \cite{WM2022}. Here, due
to the larger scale of the problem, as well as the nonconvexity of
both the constraint and objective function, convergence of the
proposed BGS-GTF-HSR is not theoretically guaranteed. Nonetheless, we
have not witnessed any convergence issues in our experimental
evaluations.

\section{Experimental Study}
\label{sec:results}
\subsection{Experimental Setup}
\label{sec:resultssetup}
We now present a body of experimental results to evaluate the proposed
BGS-GTF-HSR framework.  Experiments using both simulated and real
datasets are conducted. In the simulated experiments, both degradation
by the traditional IGK as well as the more realistic AGK are
considered.  Moreover, within the experiments for each kernel, both
blind and non-blind HSR are employed to demonstrate the superiority of
the proposed BGS-GTF-HSR. Since BGS-GTF-HSR is unsupervised,
comparisons are made to only unsupervised techniques from prior
literature; specifically, we compare to Hysure \cite{SBA2015}, SURE
\cite{NUS2022}, LTMR \cite{DL2019}, LRTA,\cite{LLL2021}, and ZSL
\cite{DGL2023} as MF-HSR methods\footnote{We note that, while LTMR and
	LRTA employ certain aspects of tensors---namely, tensor rank---their
	operation is more in line with the MF-HSR framework of \eqref{eq: Matrix
		Formulation} than the TF-HSR of \eqref{eq: Tensor
		Formulation}; we thus treat them as MF-HSR techniques here.}, and to
STEREO \cite{KFS2018}, CSTF \cite{LDF2018}, and FSTRD \cite{CZH2022}
as TF-HSR methods. Note that, for STEREO, we use its blind version
(B-STE) in the blind HSR experiments. For the remaining methods, the
spatial-degradation matrices are estimated via the technique suggested
in \cite{SBA2015}. The proposed BGS-GTF-HSR is implemented in MATLAB
R2021a on Intel\textsuperscript{\textregistered}
Core\textsuperscript{TM} i7-8700 CPU @ 3.20 GHz with 32-GB RAM. We
measure performance in terms of peak signal-to-noise ratio (PSNR),
root mean square error (RMSE), spectral angle mapper (SAM), and
structural similarity metric (SSIM).
\begin{figure}[t]
	\centering
	\setlength{\tabcolsep}{0.25mm}
	\begin{tabular}{cm{0.45\figurewidth}cm{0.45\figurewidth}}\rotatebox[origin=c]{90}{\footnotesize{RMSE}}&\includegraphics[width=0.45\figurewidth]{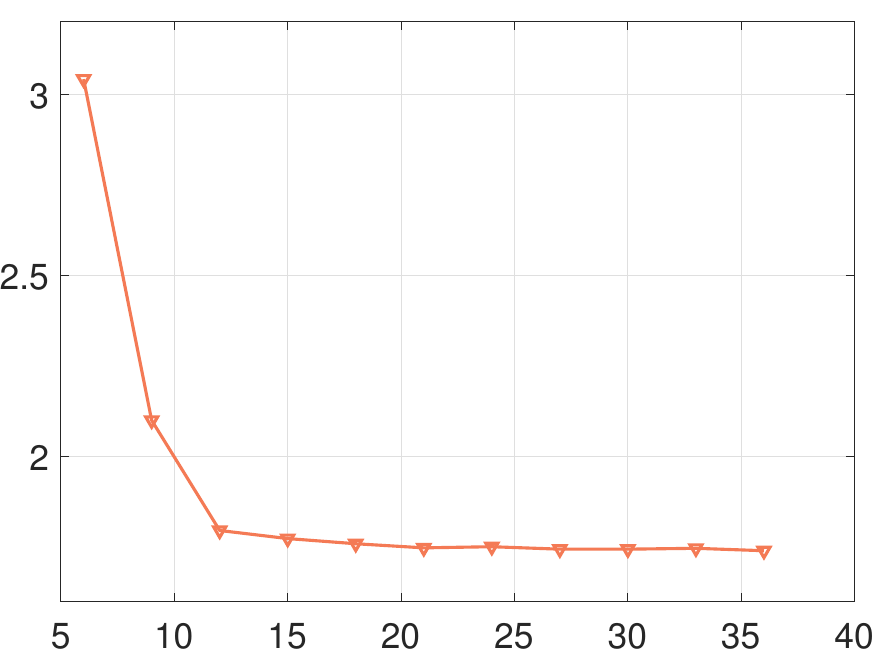}&\rotatebox[origin=c]{90}{\footnotesize{RMSE}}&
	\includegraphics[width=0.45\figurewidth]{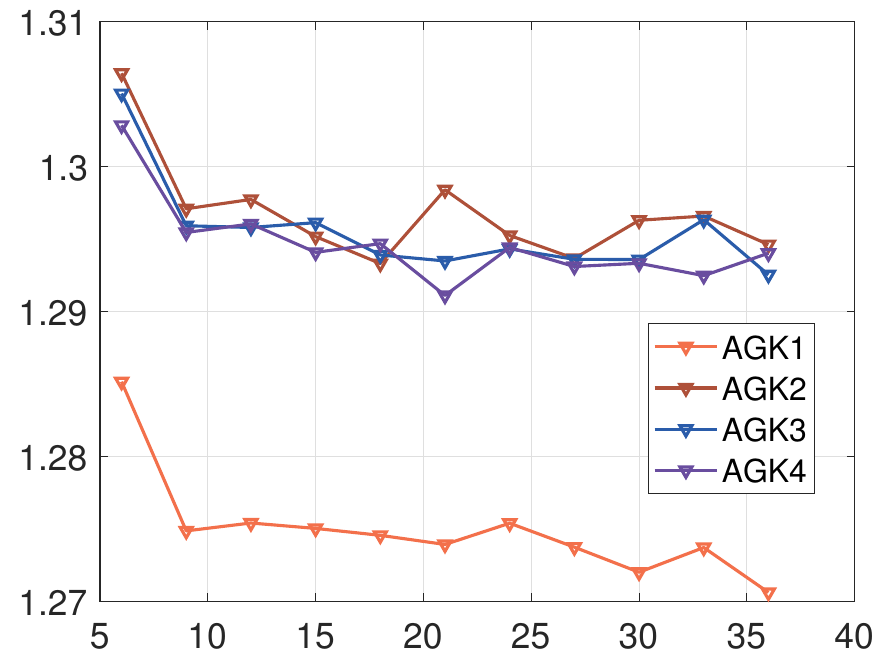}\\\quad&
		\multicolumn{1}{c}{\footnotesize{$C$}}&\quad&\multicolumn{1}{c}{\footnotesize{$C$}}\\\quad&
		\multicolumn{1}{c}{\footnotesize{(a)}}&\quad&\multicolumn{1}{c}{\footnotesize{(b)}}
	\end{tabular}
	\caption{\label{fig:Tuning process} Tuning process of spectral rank $C$ in terms of RMSE. (a) URBAN dataset. (b) Houston2013 dataset.}
\end{figure}
\begin{figure}[t]
	\centering
	\setlength{\tabcolsep}{0.3mm}
	\begin{tabular}{m{0.19\figurewidth}m{0.19\figurewidth}m{0.19\figurewidth}m{0.19\figurewidth}m{0.19\figurewidth}}
		\includegraphics[width=0.19\figurewidth]{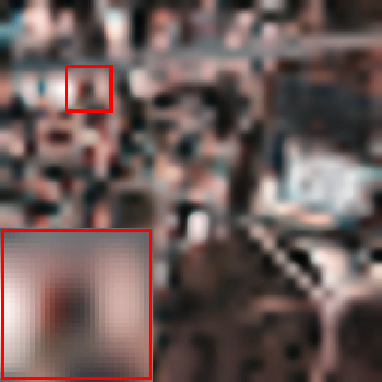}&
		\includegraphics[width=0.19\figurewidth]{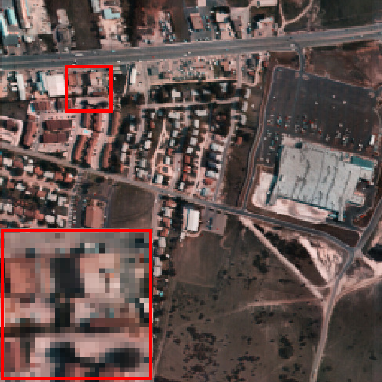}  &
		\includegraphics[width=0.19\figurewidth]{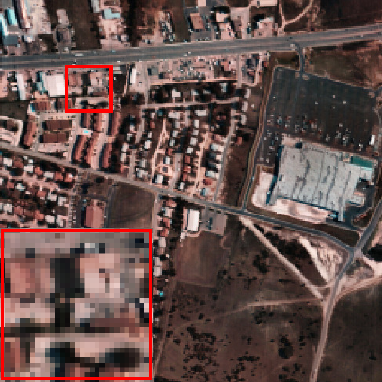}    &
		\includegraphics[width=0.19\figurewidth]{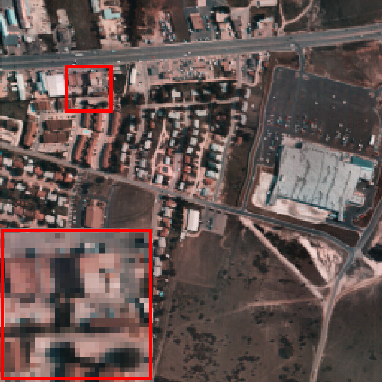}   &
		\includegraphics[width=0.19\figurewidth]{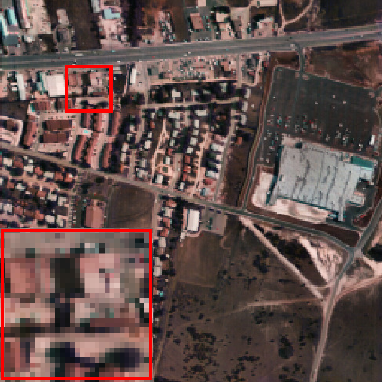}   \\
		\includegraphics[width=0.19\figurewidth]{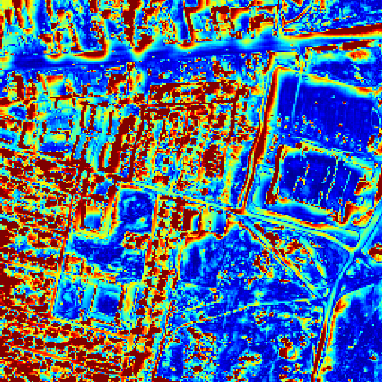}  &
		\includegraphics[width=0.19\figurewidth]{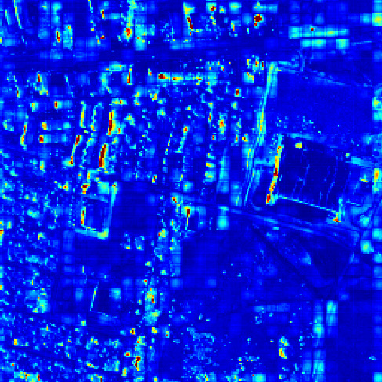}  &
		\includegraphics[width=0.19\figurewidth]{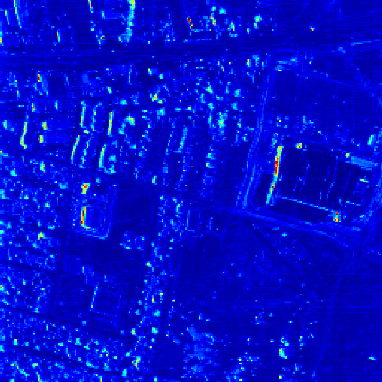}    &
		\includegraphics[width=0.19\figurewidth]{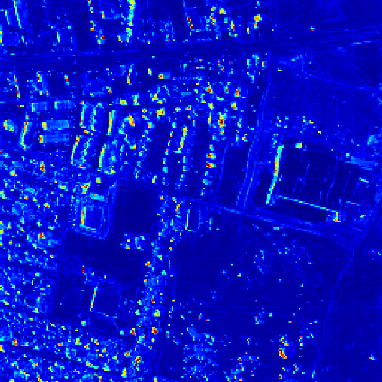}   &
		\includegraphics[width=0.19\figurewidth]{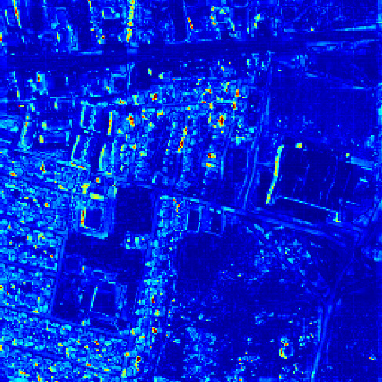} \\
		\multicolumn{1}{c}{\footnotesize{Bicubic}}
		&\multicolumn{1}{c}{\footnotesize{Hysure}}
		& \multicolumn{1}{c}{\footnotesize{{SURE}}}
		& \multicolumn{1}{c}{\footnotesize{LTMR}}
		& \multicolumn{1}{c}{\footnotesize{LRTA}} \\
		\includegraphics[width=0.19\figurewidth]{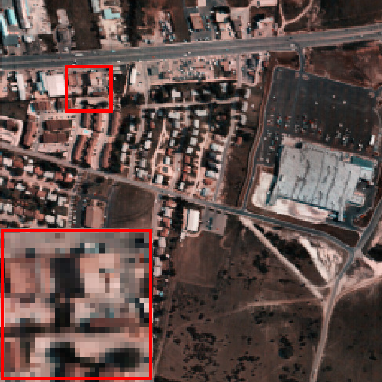} &
		\includegraphics[width=0.19\figurewidth]{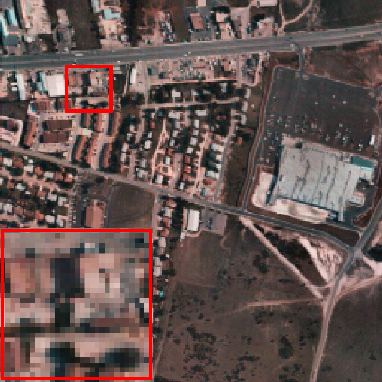}  &
		\includegraphics[width=0.19\figurewidth]{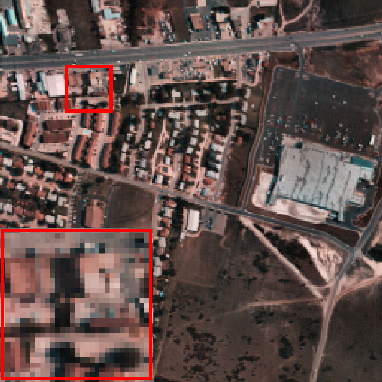}    &
		\includegraphics[width=0.19\figurewidth]{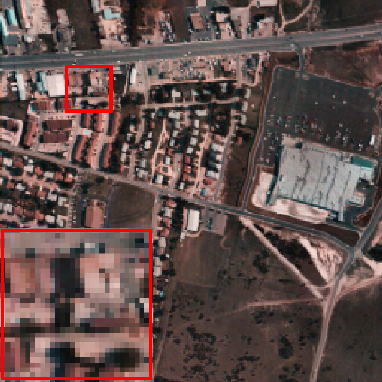}   &
		\includegraphics[width=0.19\figurewidth]{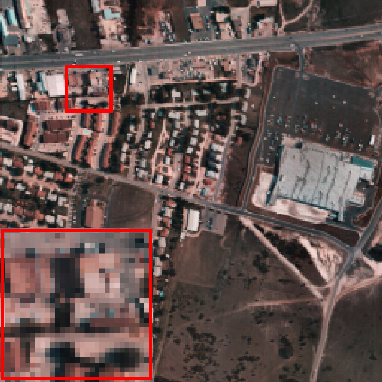}\\
		\includegraphics[width=0.19\figurewidth]{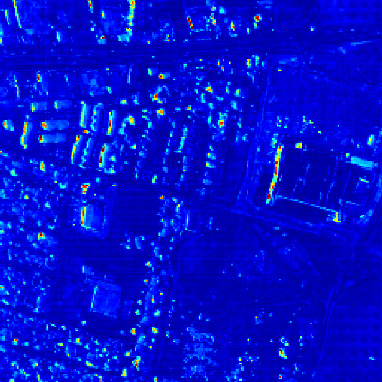} &
		\includegraphics[width=0.19\figurewidth]{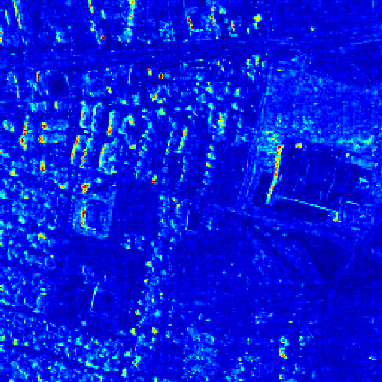}  &
		\includegraphics[width=0.19\figurewidth]{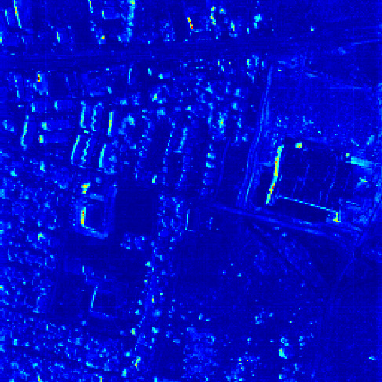}    &
		\includegraphics[width=0.19\figurewidth]{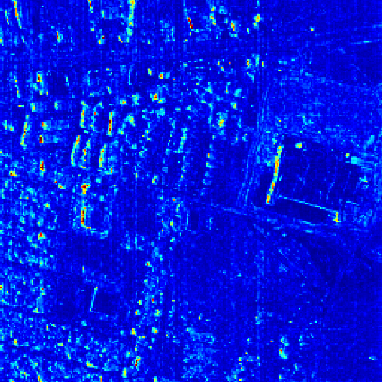}   &
		\includegraphics[width=0.19\figurewidth]{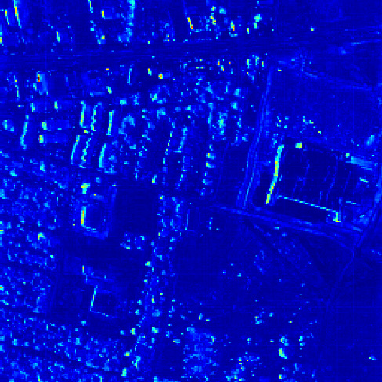} \\
		\multicolumn{1}{c}{\footnotesize{ZSL}}
		&\multicolumn{1}{c}{\footnotesize{STEREO}}
		& \multicolumn{1}{c}{\footnotesize{CSTF}}
		& \multicolumn{1}{c}{\footnotesize{FSTRD}}
		& \multicolumn{1}{c}{\footnotesize{BGS-GTF-HSR}}\\
		\multicolumn{5}{c}{\includegraphics[width=0.6\linewidth]{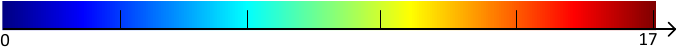}}
	\end{tabular}
	\caption{\label{fig:URBAN visualization} Non-blind fusion results and error maps on the URBAN dataset with IGK. {Pseudo-color is composed of bands 45, 20 and 10.}}
\end{figure}

As for data, we employ the URBAN, Houston2013, and Houston2018
datasets.  The URBAN
dataset\footnote{https://www.erdc.usace.army.mil/Media/Fact-Sheets/Fact-Sheet-Article-View/Article/610433/hypercube/\#}
dataset is a 210-band HSI of $307\times 307$ pixels at $2$-m spatial
resolution. The spectral coverage is $400$--$2,500$\,nm with a $10$-nm
sampling interval. Due to critical water vapor and atmospheric
effects, bands 1--4, 76, 87, 101--111, 136--153, and 198--210 were
discarded. The upper-left corner of the processed image was retained
to obtain a $256\times256\times 162$ SRI.  The Houston2013
dataset\footnote{https://hyperspectral.ee.uh.edu/?page\_id=459} is a
hyperspectral image $349\times1905\times144$ spatial size at $2.5$-m
spatial resolution. The 144 spectral bands cover $380$\,nm to
$1050$,nm. With zero pixels being discarded, a
$322\times1903\times144$ image cube forms the final ground truth; this
is then partitioned into 20 overlapping $256\times256\times144$
subscenes as SRIs. Finally, the Houston2018
dataset\footnote{https://hyperspectral.ee.uh.edu/?page\_id=1075}
consists of a real-world HSI-MSI pair. The HSI is of size $1,202
\times 4,172 \times 48$ and was acquired by an ITRES CASI 1500 sensor,
covering wavelengths $380$--$1,050$\,nm at $1$-m spatial resolution. The image
is cropped to size $500 \times 500 \times 48$ for further processing.
The $12,020 \times 11,920 \times 3$ MSI of the pair was acquired by a DiMAC
ULTRALIGHT+ at a very high spatial resolution of $5$\,cm. An
area of the size $10,000 \times 10,000\times 3$ that is registered
with the HSI is then selected. Considering that the 20-times gap in the spatial
resolution between the HSI and the MSI is too large for current fusion
methods, and a size of $10,000 \times 10,000$ is also too large for the
RAM of our machine, we perform five-times downsampling on the HSI and 25-times
downsampling on the MSI to generate a data pair composed of a
$100\times100\times48$ HSI at $5$-m spatial resolution and a
$400\times400\times3$ MSI at $1.25$-m spatial resolution.
\begin{table}[t]
	\centering
	\renewcommand{\arraystretch}{1.3}
	\tabcolsep=1.55mm
	\caption{Performance on the URBAN dataset}
	\begin{tabular}{c|c|cccc}
		\toprule
		\hline
		\multirow{2}{*}{\textbf{Setup}} & \multirow{2}{*}{\textbf{Methods}} & \multicolumn{4}{c}{\textbf{Quality Indices}} \\\cline{3-6} &\quad
		&  PSNR$\uparrow$                 & RMSE$\downarrow$    & SAM$\downarrow$    & SSIM$\uparrow$       \\\hline
		\multirow{9}{*}{\textbf{Non-blind}} &        Hysure            &  40.0915   &  3.3614   &   2.7795  &   0.9894 \\\quad& SURE 
		&         42.4016          &  2.2118   &   1.8641  &  0.9911      \\
		\quad	&       LTMR             &  44.0306   & 2.3757    & 2.0022    &   0.9918 \\\quad
		&         LRTA           &  42.6566   &  3.2413   &   2.4486  &  0.9852  \\\quad
		&          ZSL          &  42.5588   &  2.3110   &  1.9789   & 0.9921   \\\quad
		&         STEREO           &   41.1537  &  2.7029   &  2.4912   &  0.9833  \\\quad
		&           CSTF         &  44.2222   &   1.8827  &   1.6649  &  0.9918  \\\quad
		&        FSTRD            &  41.7204   &   2.7526  &  2.4355   &  0.9848  \\\quad
		&          BGS-GTF-HSR         &   \textbf{45.4533}  &  \textbf{1.7936}   &  \textbf{1.6218}   & \textbf{0.9939}   \\
		\hline
		\multirow{9}{*}{\textbf{Blind}} &       Hysure            &  40.1689   &  3.3059   &   2.7597  &   0.9895 \\\quad& SURE 
		&         41.8384          &  2.3932   &   2.1031  &  0.9901      \\
		\quad	&       LTMR             &  43.9173   & 2.3898    & 2.0207    &   0.9917 \\\quad
		&         LRTA           &  42.5482   &  3.2605   &   2.4644  &  0.9850  \\\quad
		&          ZSL          &  42.7150   &  2.2951   &  1.9791   & 0.9923   \\\quad
		&         B-STE           &   39.7142  &  3.0233   &  2.7546   &  0.9800  \\\quad
		&           CSTF         &  43.6849   &   \textbf{1.9766}  &   1.8166  &  0.9909  \\\quad
		&        FSTRD            &  41.6403   &   2.7522  &  2.4419   &  0.9838  \\\quad
		&          BGS-GTF-HSR         &   \textbf{44.6240}  &  {2.0684}   &  \textbf{1.7326}   & \textbf{0.9926}   \\
		\hline
		\bottomrule
	\end{tabular}
	\label{tab:URBAN metrics}
\end{table}

\begin{figure}[htbp]
	\centering
	\setlength{\tabcolsep}{0.15mm}
	\begin{tabular}{m{0.19\figurewidth}m{0.19\figurewidth}m{0.19\figurewidth}m{0.19\figurewidth}m{0.19\figurewidth}}
		\includegraphics[width=0.19\figurewidth]{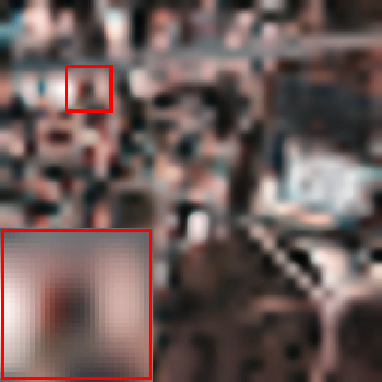} &
		\includegraphics[width=0.19\figurewidth]{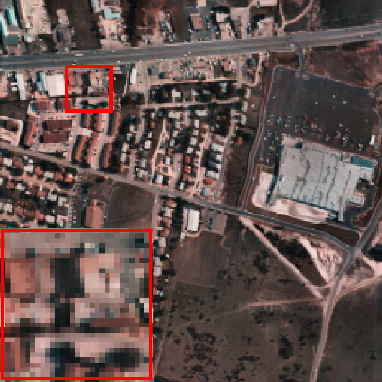}  &
		\includegraphics[width=0.19\figurewidth]{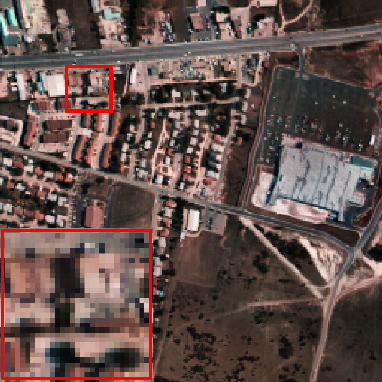}    &
		\includegraphics[width=0.19\figurewidth]{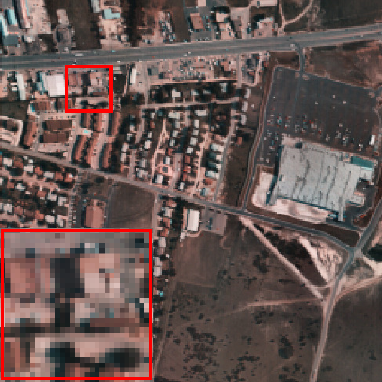}   &
		\includegraphics[width=0.19\figurewidth]{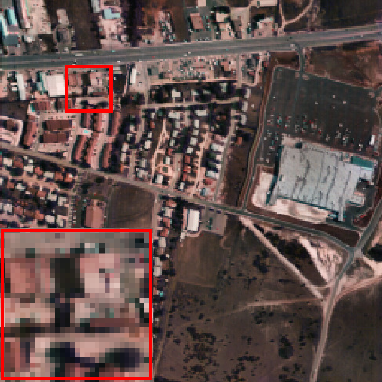}\\
		\includegraphics[width=0.19\figurewidth]{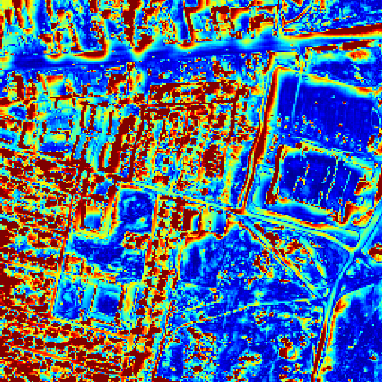} &
		\includegraphics[width=0.19\figurewidth]{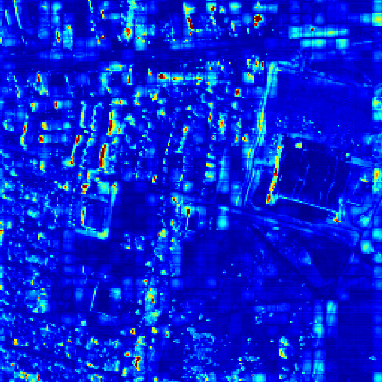}  &
		\includegraphics[width=0.19\figurewidth]{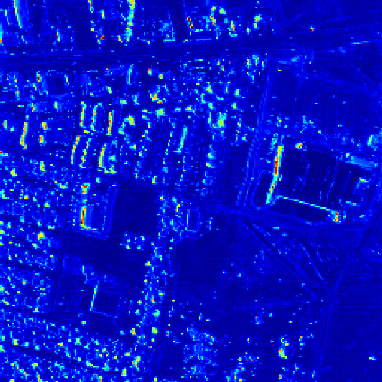}    &
		\includegraphics[width=0.19\figurewidth]{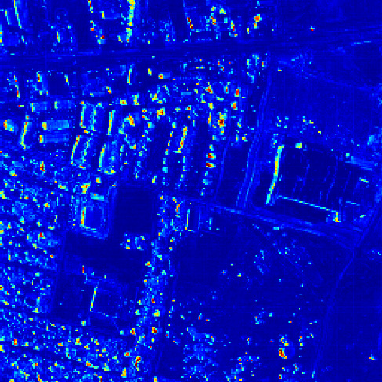}   &
		\includegraphics[width=0.19\figurewidth]{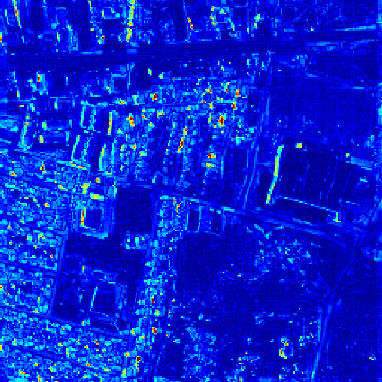} \\
		\multicolumn{1}{c}{\footnotesize{Bicubic}}
		&\multicolumn{1}{c}{\footnotesize{Hysure}}
		& \multicolumn{1}{c}{\footnotesize{{SURE}}}
		& \multicolumn{1}{c}{\footnotesize{LTMR}}
		& \multicolumn{1}{c}{\footnotesize{LRTA}} \\
		\includegraphics[width=0.19\figurewidth]{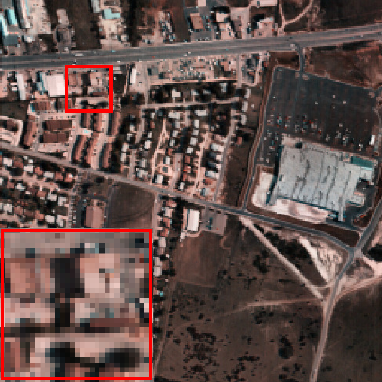} &
		\includegraphics[width=0.19\figurewidth]{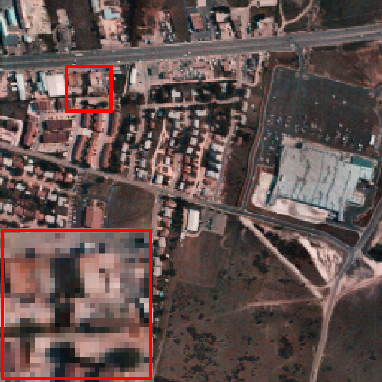}  &
		\includegraphics[width=0.19\figurewidth]{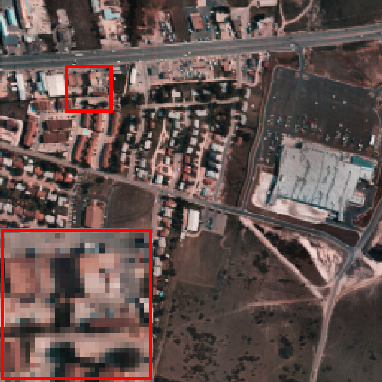}    &
		\includegraphics[width=0.19\figurewidth]{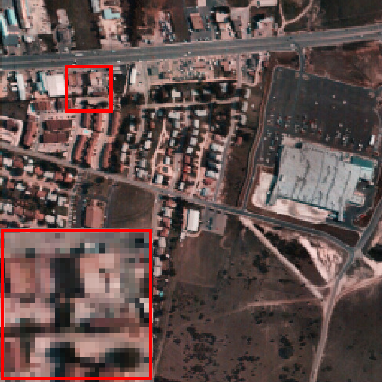}   &
		\includegraphics[width=0.19\figurewidth]{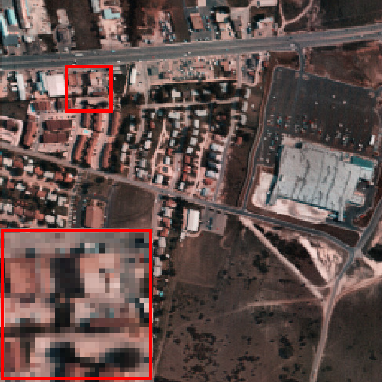}\\
		\includegraphics[width=0.19\figurewidth]{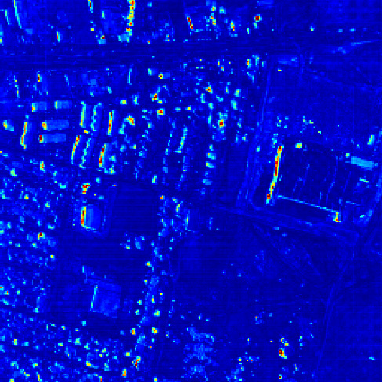} &
		\includegraphics[width=0.19\figurewidth]{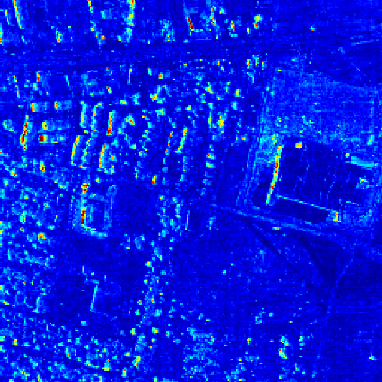}  &
		\includegraphics[width=0.19\figurewidth]{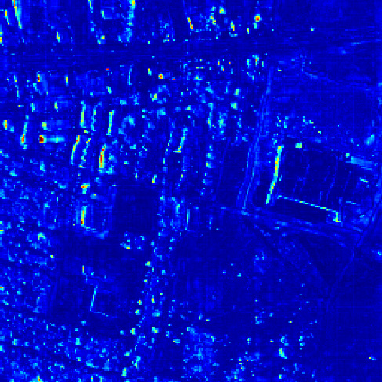}    &
		\includegraphics[width=0.19\figurewidth]{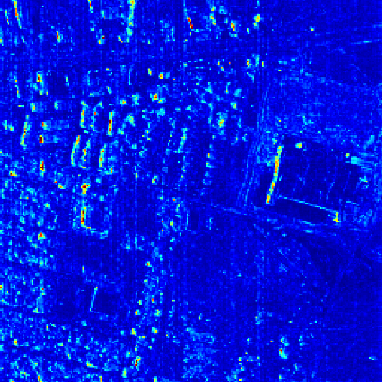}   &
		\includegraphics[width=0.19\figurewidth]{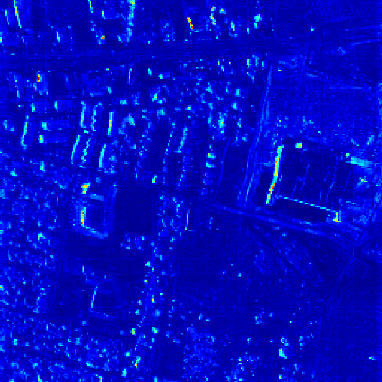} \\
		\multicolumn{1}{c}{\footnotesize{ZSL}}
		&\multicolumn{1}{c}{\footnotesize{B-STE}}
		& \multicolumn{1}{c}{\footnotesize{CSTF}}
		& \multicolumn{1}{c}{\footnotesize{FSTRD}}
		& \multicolumn{1}{c}{\footnotesize{BGS-GTF-HSR}} \\
		\multicolumn{5}{c}{\includegraphics[width=0.6\linewidth]{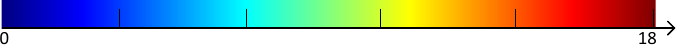}}
	\end{tabular}
	\caption{\label{fig:URBAN blind visualization} Blind fusion results and error maps on URBAN dataset with IGK. {Pseudo-color is composed of bands 45, 20 and 10.}}
\end{figure}

\begin{figure}[t]
	\centering
	\setlength{\tabcolsep}{0.15mm}
	\begin{tabular}{m{0.2375\figurewidth}m{0.2375\figurewidth}m{0.2375\figurewidth}m{0.2375\figurewidth}}
		\includegraphics[width=0.2375\figurewidth]{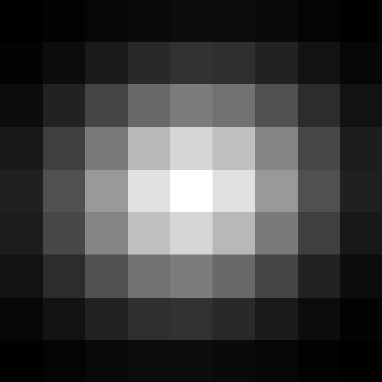}&
		\includegraphics[width=0.2375\figurewidth]{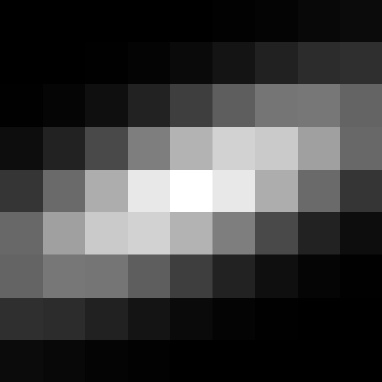}  &
		\includegraphics[width=0.2375\figurewidth]{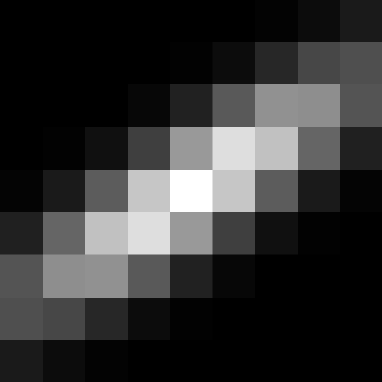}    &
		\includegraphics[width=0.2375\figurewidth]{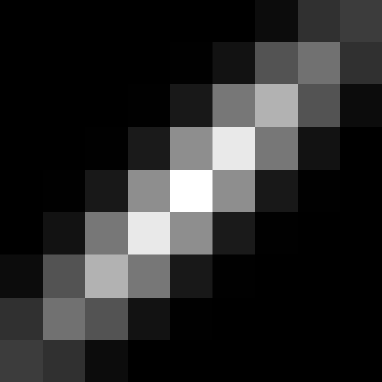}   \\
		\multicolumn{1}{c}{\footnotesize{AGK 1}}                                       &\multicolumn{1}{c}{\footnotesize{AGK 2}}                                       & \multicolumn{1}{c}{\footnotesize{AGK 3}} & \multicolumn{1}{c}{\footnotesize{AGK 4}}\\
		\multicolumn{1}{c}{\footnotesize{{$\theta=\frac{\pi}{16}$}}}                     &\multicolumn{1}{c}{\footnotesize{{$\frac{2\pi}{16}$}}}                                       & \multicolumn{1}{c}{\footnotesize{{$\frac{3\pi}{16}$}}} & \multicolumn{1}{c}{\footnotesize{{$\frac{4\pi}{16}$}}}\\
		\multicolumn{1}{c}{\footnotesize{{$a=0.5$}}}                     &\multicolumn{1}{c}{\footnotesize{{$0.3$}}}                                       & \multicolumn{1}{c}{\footnotesize{{$0.3$}}} & \multicolumn{1}{c}{\footnotesize{{$0.3$}}}\\
		\multicolumn{1}{c}{\footnotesize{{$b=0.6$}}}                     &\multicolumn{1}{c}{\footnotesize{{$0.9$}}}                                       & \multicolumn{1}{c}{\footnotesize{{$1.2$}}} & \multicolumn{1}{c}{\footnotesize{{$1.5$}}}\\
		\multicolumn{1}{c}{\footnotesize{$4.4\times10^{13}$}}                     &\multicolumn{1}{c}{\footnotesize{$2.9\times10^{5}$}}                                       & \multicolumn{1}{c}{\footnotesize{$10^3$}} & \multicolumn{1}{c}{\footnotesize{$40$}}
	\end{tabular}
	\caption{\label{fig:AGK} The four AGK blurring kernels; the numbers
		in the last row are the condition number. {$\theta,a,b$ are defined in Section \ref{sec:AGK}.}}
\end{figure}

\subsection{BGS-GTF-HSR Hyperparameters}

The main hyperparameters concerning the proposed BGS-GTF-HSR include
the latent Tucker-Rank $(L_1,L_2,C)$, the partition parameters $K_1$
and $K_2$ in Alg.~\ref{alg:BGS-TeRF Part 1}, and the shape parameter
$\mathbf{t}=[t_1,t_2,t_3]$ of the proposed B-unfolding deciding the
specific BGS pattern of the core tensor. Since it is well-recognized
that the HSI is not as low rank in the spatial domain as it is in the
spectral domain, we set $L_1=M_1$ and $L_2=M_2$; that is,
$L_1=L_2=256$ for the simulated dataset, and $L_1=L_2=400$ for the
real Houston2018 HSI-MSI pair. Because the spatial information is
largely preserved in the HSI, we set the ratio of $K_i$ to $L_i-K_i$
to be $15:1$ in order to have the atoms extracted from the MSI be
dominant; that is, $K_1=K_2=240$ for the simulated experiments, and
$K_1=K_2=375$ for the real experiments. As for the spectral rank $C$,
we fine-tune it on both the URBAN and Houston2013 datasets as depicted
in Fig.~\ref{fig:Tuning process}, setting $C=12$ in consideration of
both performance and computational efficiency. To determine the BGS
shape parameter $\mathbf{t}$, we note that the size of the elementary
block must be significantly smaller than the overall core tensor, and
the spatial shape must be much larger than the spectral shape in
accordance with the lower rank in the spectral domain. These lead to
the condition that $t_1\ll L_1$, $t_2\ll L_2$ and $t_1, t_2 \gg
t_3$. Thus, we use $\mathbf{t}=[16,16,3]$ and $\mathbf{t}=[20,20,3]$
for the simulated and real experiments, respectively.

\subsection{HSR with IGK}
\begin{figure}[t]
	\centering
	\setlength{\tabcolsep}{0.3mm}
	\begin{tabular}{m{0.19\figurewidth}m{0.19\figurewidth}m{0.19\figurewidth}m{0.19\figurewidth}m{0.19\figurewidth}}
		\includegraphics[width=0.19\figurewidth]{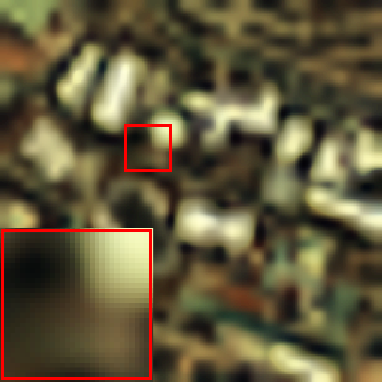}&
		\includegraphics[width=0.19\figurewidth]{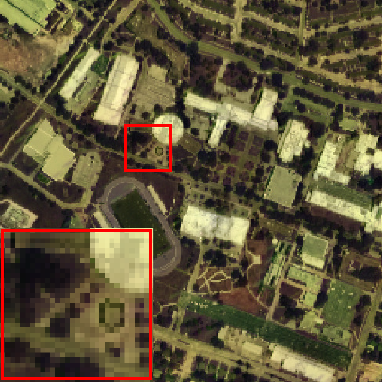}  &
		\includegraphics[width=0.19\figurewidth]{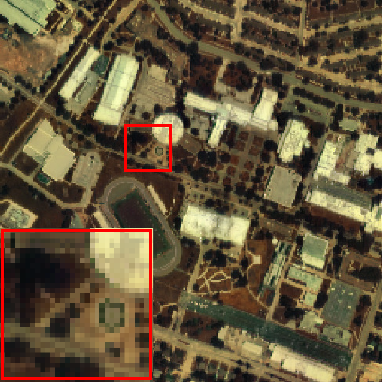}    &
		\includegraphics[width=0.19\figurewidth]{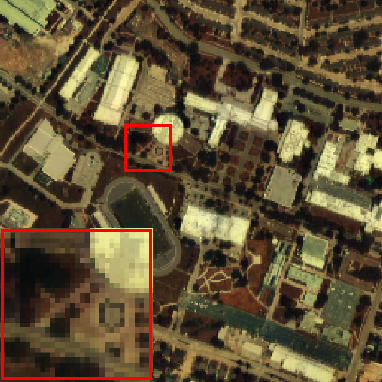}  &
		\includegraphics[width=0.19\figurewidth]{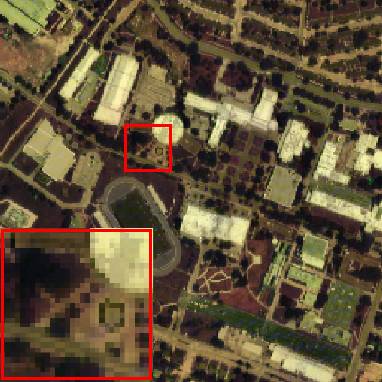}\\
		\includegraphics[width=0.19\figurewidth]{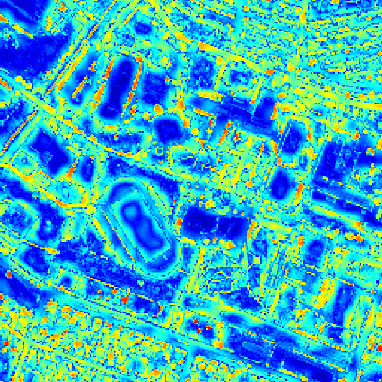} &
		\includegraphics[width=0.19\figurewidth]{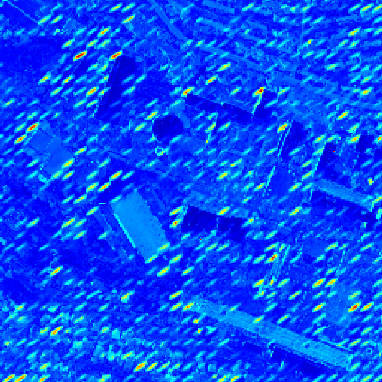}  &
		\includegraphics[width=0.19\figurewidth]{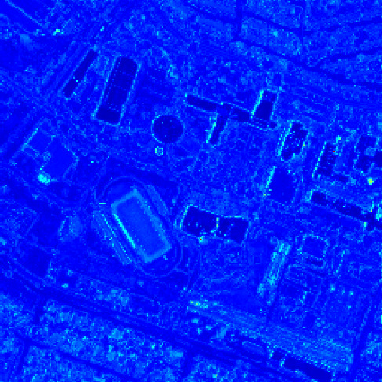}    &
		\includegraphics[width=0.19\figurewidth]{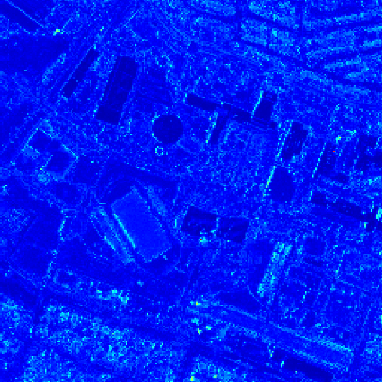}   &
		\includegraphics[width=0.19\figurewidth]{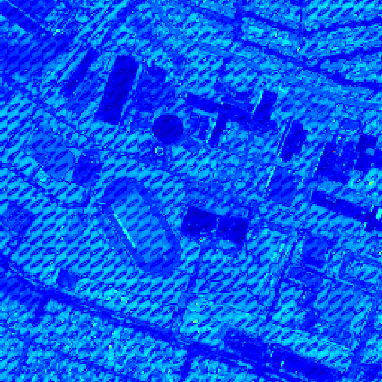} \\
		\multicolumn{1}{c}{\footnotesize{Bicubic}}                                       &\multicolumn{1}{c}{\footnotesize{Hysure}}                                       & \multicolumn{1}{c}{\footnotesize{{SURE}}} & \multicolumn{1}{c}{\footnotesize{LTMR}} & \multicolumn{1}{c}{\footnotesize{LRTA}}
		\\
		\includegraphics[width=0.19\figurewidth]{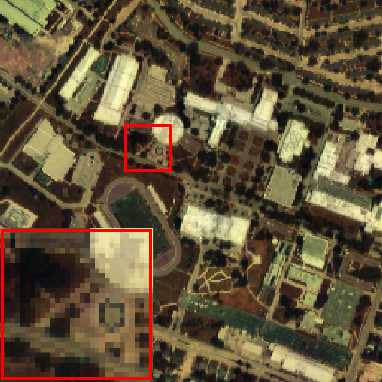} &
		\includegraphics[width=0.19\figurewidth]{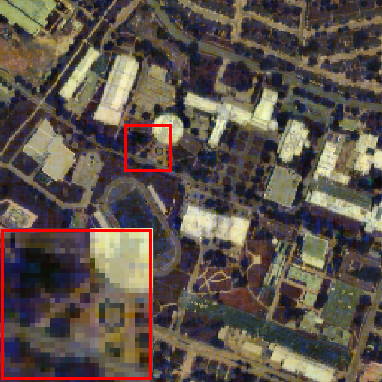}  &
		\includegraphics[width=0.19\figurewidth]{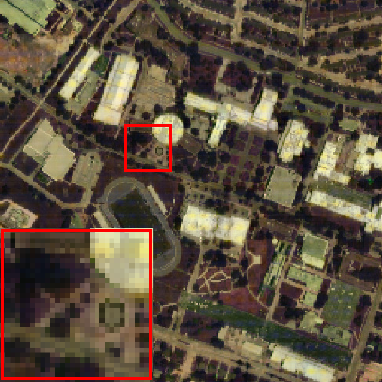}    &
		\includegraphics[width=0.19\figurewidth]{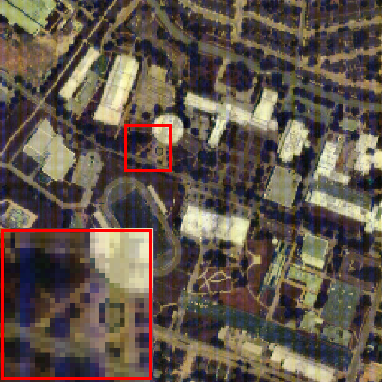}   &
		\includegraphics[width=0.19\figurewidth]{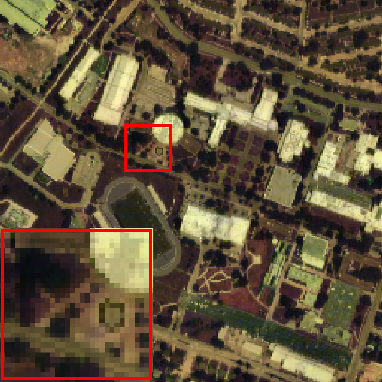}\\
		\includegraphics[width=0.19\figurewidth]{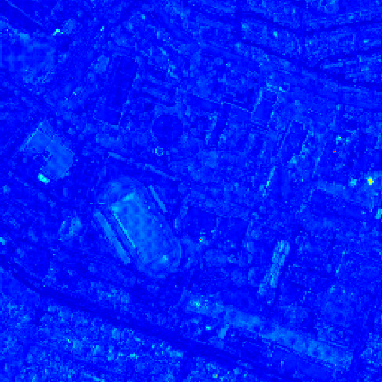} &
		\includegraphics[width=0.19\figurewidth]{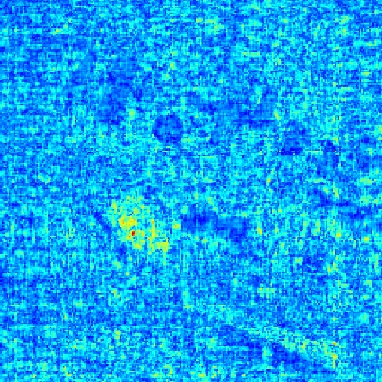}  &
		\includegraphics[width=0.19\figurewidth]{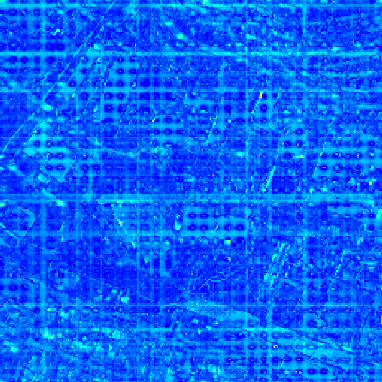}    &
		\includegraphics[width=0.19\figurewidth]{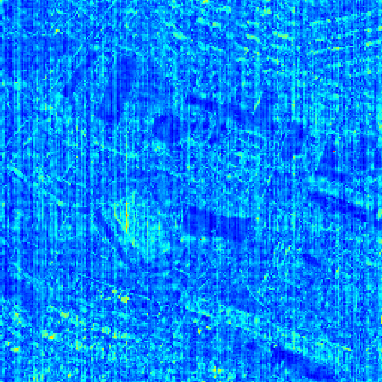}   &
		\includegraphics[width=0.19\figurewidth]{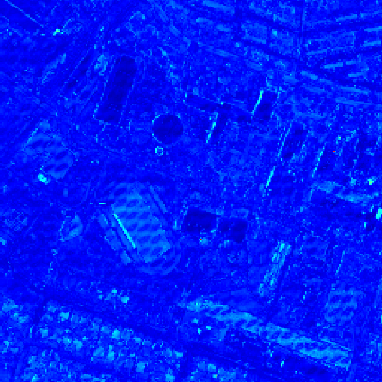} \\
		\multicolumn{1}{c}{\footnotesize{ZSL}}                                       &\multicolumn{1}{c}{\footnotesize{STEREO}}                                       & \multicolumn{1}{c}{\footnotesize{CSTF}} & \multicolumn{1}{c}{\footnotesize{FSTRD}} & \multicolumn{1}{c}{\footnotesize{BGS-GTF-HSR}} \\
		\multicolumn{5}{c}{\includegraphics[width=0.6\linewidth]{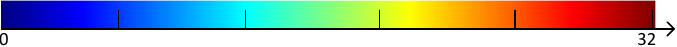}}-
	\end{tabular}
	\caption{\label{fig:Houston visualization 1} Non-blind fusion results and error maps for the Houston2013 dataset with AGK~3 (Scene 2). {Pseudo-color is composed of bands 30, 20 and 10.}}
\end{figure}

\begin{figure}[t]
	\centering
	\setlength{\tabcolsep}{0.15mm}
	\begin{tabular}{m{0.19\figurewidth}m{0.19\figurewidth}m{0.19\figurewidth}m{0.19\figurewidth}m{0.19\figurewidth}}
		\includegraphics[width=0.19\figurewidth]{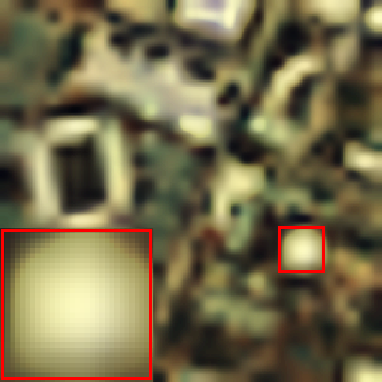} &
		\includegraphics[width=0.19\figurewidth]{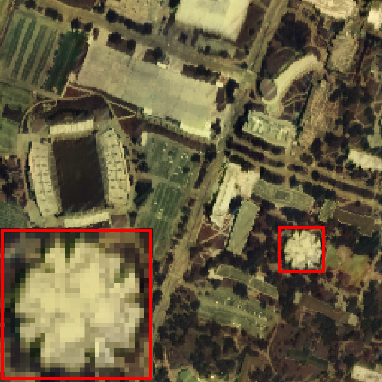}  &
		\includegraphics[width=0.19\figurewidth]{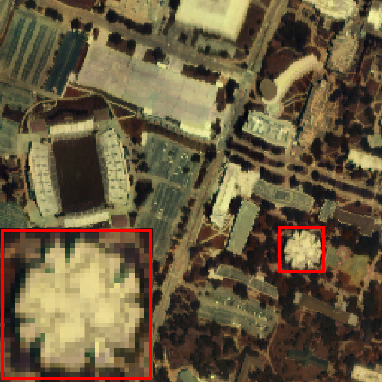}    &
		\includegraphics[width=0.19\figurewidth]{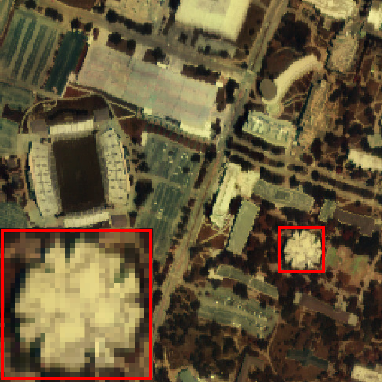}   &
		\includegraphics[width=0.19\figurewidth]{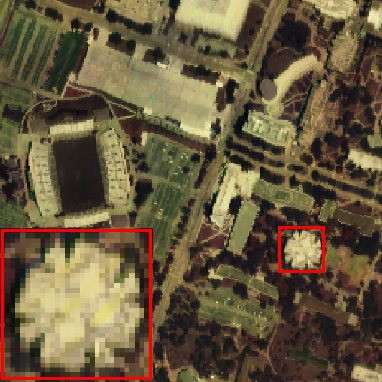}\\
		\includegraphics[width=0.19\figurewidth]{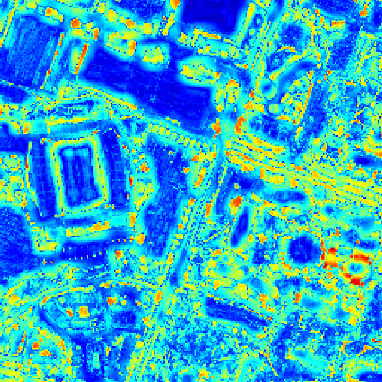} &
		\includegraphics[width=0.19\figurewidth]{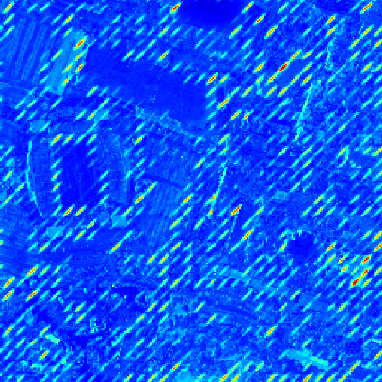}  &
		\includegraphics[width=0.19\figurewidth]{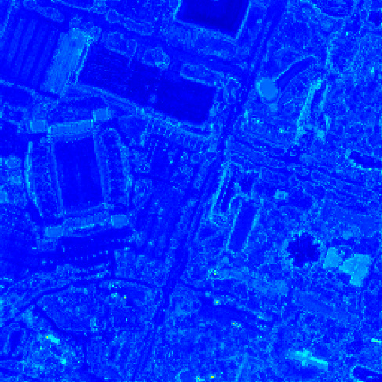}    &
		\includegraphics[width=0.19\figurewidth]{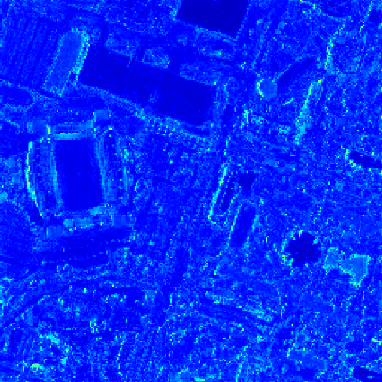}   &
		\includegraphics[width=0.19\figurewidth]{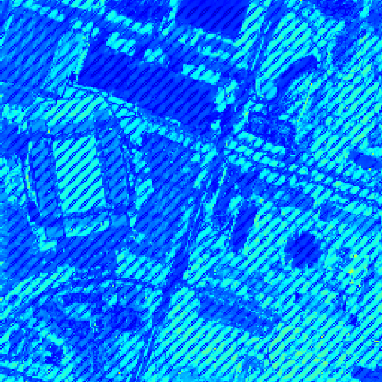} \\
		\multicolumn{1}{c}{\footnotesize{Bicubic}}                                       &\multicolumn{1}{c}{\footnotesize{Hysure}}                                       & \multicolumn{1}{c}{\footnotesize{{SURE}}} & \multicolumn{1}{c}{\footnotesize{LTMR}} & \multicolumn{1}{c}{\footnotesize{LRTA}}
		\\
		\includegraphics[width=0.19\figurewidth]{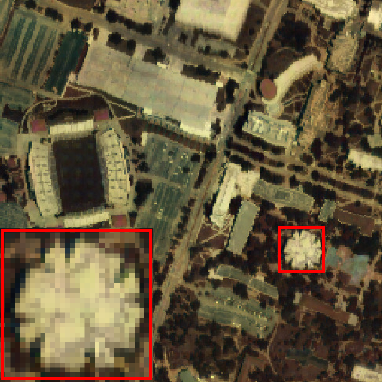} &
		\includegraphics[width=0.19\figurewidth]{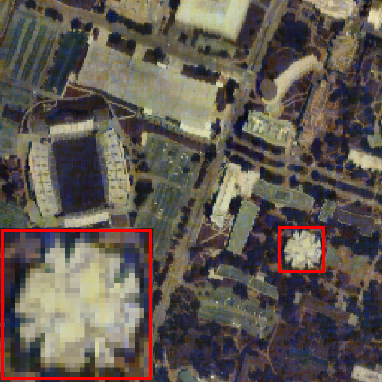}  &
		\includegraphics[width=0.19\figurewidth]{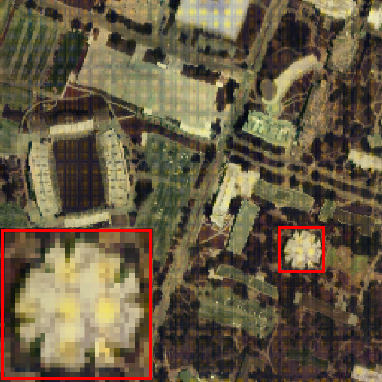}    &
		\includegraphics[width=0.19\figurewidth]{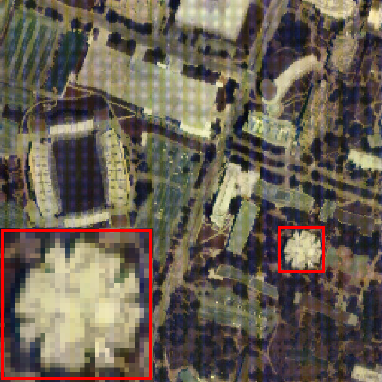}   &
		\includegraphics[width=0.19\figurewidth]{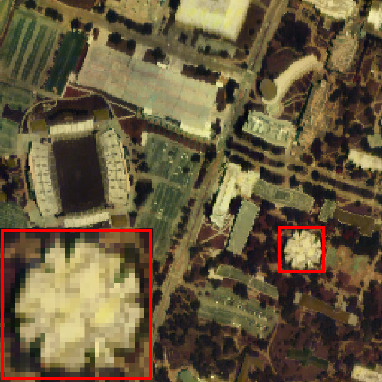}\\
		\includegraphics[width=0.19\figurewidth]{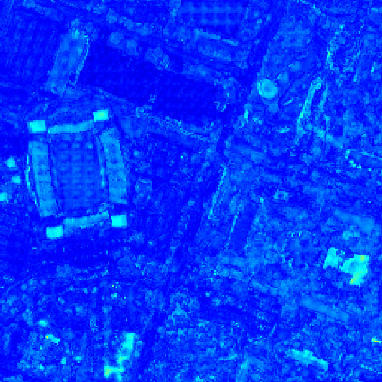} &
		\includegraphics[width=0.19\figurewidth]{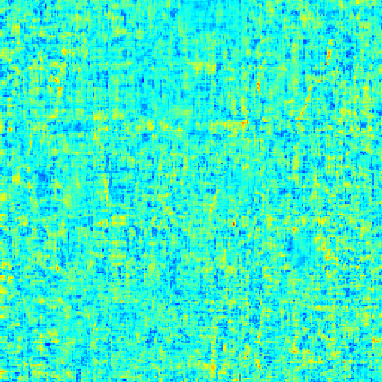}  &
		\includegraphics[width=0.19\figurewidth]{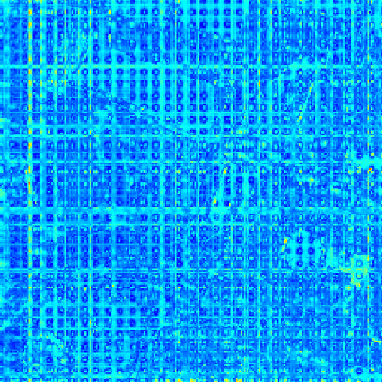}    &
		\includegraphics[width=0.19\figurewidth]{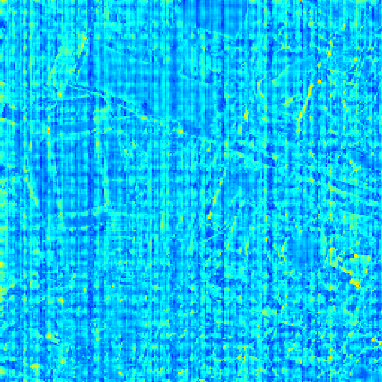}   &
		\includegraphics[width=0.19\figurewidth]{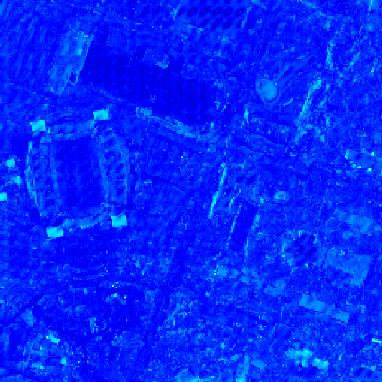} \\
		\multicolumn{1}{c}{\footnotesize{ZSL}}                                       &\multicolumn{1}{c}{\footnotesize{STEREO}}                                       & \multicolumn{1}{c}{\footnotesize{CSTF}} & \multicolumn{1}{c}{\footnotesize{FSTRD}} & \multicolumn{1}{c}{\footnotesize{BGS-GTF-HSR}} \\
		\multicolumn{5}{c}{\includegraphics[width=0.6\linewidth]{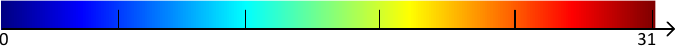}}
	\end{tabular}
	\caption{\label{fig:Houston visualization 2} Non-blind fusion results and error maps for the Houston2013 dataset with AGK 4 (Scene 5). {Pseudo-color is composed of bands 30, 20 and 10.}}
\end{figure}

In the results presented in this section, we employ the URBAN dataset,
using the SRI described in Sec.~\ref{sec:resultssetup} as the ground truth
and generating an HSI and MSI from in artificially. Specifically, we
blurred the SRI with a $9\times 9$ IGK with standard
deviation $3.3973$ and spatially downsampled it by a factor of 8 to
simulate an HSI of size $32\times 32\times 162$. The
$256\times256\times 6$ MSI is generated by averaging the SRI bands
falling into the wavelength between $450$--$520$, $520$--$600$,
$630$--$690$, $760$--$900$, $1,550$--$1750$, and $2,080$--$2,350$\,nm
to simulate the spectral coverage of the USGS/NASA Landsat7 satellite
\cite{LDF2018}.

For non-blind performance, the degradation matrices are assumed to be
known.  Fig.~\ref{fig:URBAN visualization} presents the fusion results
for bands 30, 20, and 10 as pseudo-color images, along with
corresponding error maps generated by pixel-wise SAM between the
results and ground-truth SRI.  Visually, the proposed BGS-GTF-HSR has
the error map with the lowest brightness and least highlighted area,
while SURE, ZSL, and CSTF yield competitive results.  This conclusion
is further confirmed by quantitative evaluation in
Table~\ref{tab:URBAN metrics}---we note that the proposed BGS-GTF-HSR
obtains the best values for all the metrics considered for
non-blind performance.

Finally, we assess blind fusion performance, wherein the degradation
matrices are unknown to the techniques. The corresponding visual
results are given in Fig.~\ref{fig:URBAN blind visualization} while
quantitative performance is tabulated Table~\ref{tab:URBAN
	metrics}. We note that the fusion performance for all the techniques
suffers under the blind scenario, as can be observed in both
Table~\ref{tab:URBAN metrics} and in the error maps of
Fig.~\ref{fig:URBAN blind visualization}. However, the proposed
BGS-GTF-HSR still outperforms other methods since it has the best
values for most quantitative metrics, and its error map is darkest.
\subsection{HSR with AGK}
\begin{figure}[t]
	\centering
	\setlength{\tabcolsep}{0.15mm}
	\begin{tabular}{cm{0.9\figurewidth}}\rotatebox{90}{{PSNR}}&
		\includegraphics[width=0.9\figurewidth]{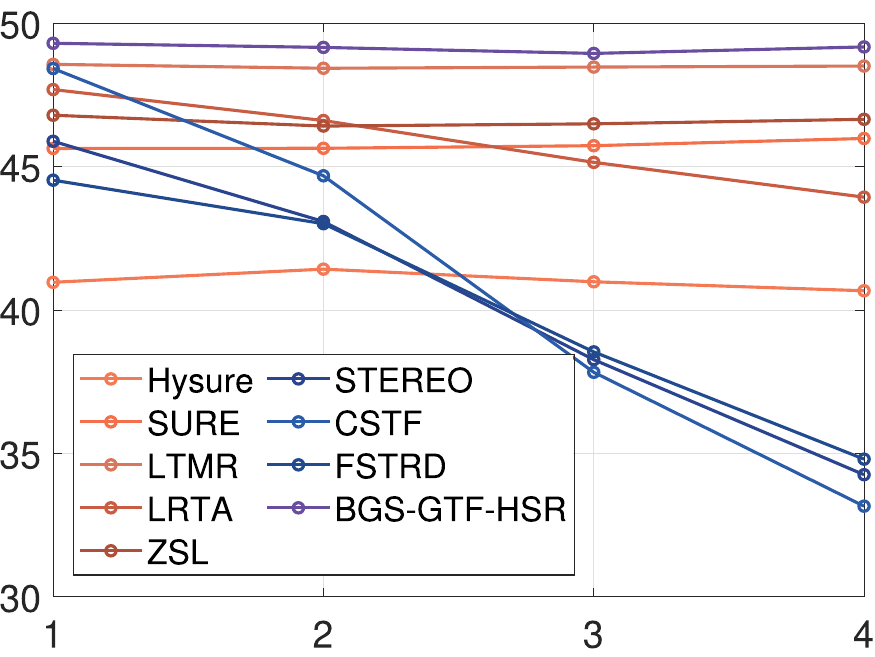} \\\quad&
		\multicolumn{1}{c}{{AGK}}
	\end{tabular}
	\caption{\label{fig:PSNRvsAGK} PSNR performance for varying AGK (1
		through 4) for the Houston2013 dataset, demonstrating how the
		methods respond to the increase in the rank of the blurring
		kernel.}
\end{figure}

\begin{table}[t]
	\centering
	\renewcommand{\arraystretch}{1.3}
	\tabcolsep=1.55mm
	\caption{Performance on the Houston2013 dataset with AGK~1 (averaged over 20 scenes)}
	\begin{tabular}{c|c|cccc}
		\toprule
		\hline
		\multirow{2}{*}{\textbf{Setup}} & \multirow{2}{*}{\textbf{Methods}} & \multicolumn{4}{c}{\textbf{Quality Indices}} \\\cline{3-6} &\quad
		&  PSNR$\uparrow$                 & RMSE$\downarrow$    & SAM$\downarrow$    & SSIM$\uparrow$       \\\hline
		\multirow{9}{*}{\textbf{Non-blind}} &        Hysure           &  40.9721   &  2.6210   &   2.4319  &   0.9801 \\\quad& SURE
		&         45.6402          &  1.5051   &   1.4854  &  0.9877      \\
		\quad	&       LTMR            &  48.5731   & 1.3353    & 1.3054    &   0.9903 \\\quad
		&         LRTA          &  47.6939   &  1.5624   &   1.6141  &  0.9876  \\\quad
		&          ZSL         &  46.7990   &  1.3749   &  1.3571   & \textbf{0.9920}   \\\quad
		&         STEREO          &   45.8905  &  1.8110   &  1.9726   &  0.9786  \\\quad
		&           CSTF        &  48.4274   &   1.3367  &   1.4734  &  0.9893  \\\quad
		&        FSTRD           &  44.5336   &   2.3053  &  2.5005   &  0.9666  \\\quad
		&          BGS-GTF-HSR         &   \textbf{49.3079}  &  \textbf{1.1857}   &  \textbf{1.2251}   & 0.9919   \\
		\hline
		\multirow{9}{*}{\textbf{{Blind}}} &       Hysure           &  41.4202   &  2.5011   &   2.3992  &   0.9819 \\\quad& SURE
		&         45.3860          &  1.5556   &   1.5274  &  0.9871      \\
		\quad	&       LTMR            &  48.2073   & 1.3768    & 1.3314    &   0.9901 \\\quad
		&         LRTA          &  48.3274   &  1.3811   &   1.4247  &  0.9898  \\\quad
		&          ZSL         &  45.4437   &  2.4206   &  2.3740   & 0.9616   \\\quad
		&         B-STE          &   44.4886  &  1.9713   &  2.1890   &  0.9770  \\\quad
		&           CSTF        &  46.8317   &   1.4326  &   1.4425  &  0.9899  \\\quad
		&        FSTRD           &  43.7708   &   2.4287  &  2.5701   &  0.9651  \\\quad
		&          BGS-GTF-HSR         &   \textbf{49.2943}  &  \textbf{1.1832}   &  \textbf{1.2218}   & \textbf{0.9918}   \\
		\hline
		\bottomrule
	\end{tabular}
	\label{tab:Houston metrics1}
\end{table}

\begin{table}[bthp!]
	\centering
	\renewcommand{\arraystretch}{1.3}
	\tabcolsep=1.55mm
	\caption{Performance on the Houston2013 dataset with AGK~2 (averaged over 20 scenes)}
	\begin{tabular}{c|c|cccc}
		\toprule
		\hline
		\multirow{2}{*}{\textbf{Setup}} & \multirow{2}{*}{\textbf{Methods}} & \multicolumn{4}{c}{\textbf{Quality Indices}} \\\cline{3-6} &\quad
		&  PSNR$\uparrow$                 & RMSE$\downarrow$    & SAM$\downarrow$    & SSIM$\uparrow$       \\\hline
		\multirow{9}{*}{\textbf{Non-blind}} &        Hysure           &  41.4307   &  2.4503   &   2.3813  &   0.9805 \\\quad& SURE
		&         45.6455          &  1.4881   &   1.4758  &  0.9882      \\
		\quad	&       LTMR            &  48.4367   & 1.3678    & 1.3320    &   0.9900 \\\quad
		&         LRTA          &  46.6136   &  1.8789   &   1.8861  &  0.9835  \\\quad
		&          ZSL         &  46.4207   &  1.4512   &  1.4595   & 0.9917   \\\quad
		&         STEREO          &   43.0912  &  2.4394   &  2.5337   &  0.9682  \\\quad
		&           CSTF        &  44.6860   &   1.8134  &   1.7637  &  0.9862  \\\quad
		&        FSTRD           &  43.0151   &   2.5345  &  2.6199   &  0.9624  \\\quad
		&          BGS-GTF-HSR         &   \textbf{49.1637}  &  \textbf{{1.2157}}   &  \textbf{1.2544}   & \textbf{0.9918}   \\
		\hline
		\multirow{9}{*}{\textbf{{Blind}}} &       Hysure           &  41.2449   &  2.5472   &   2.4611  &   0.9806 \\\quad& SURE
		&         44.9667          &  1.7049   &   1.5805  &  0.9868      \\
		\quad	&       LTMR            &  47.6852   & 1.5452    & 1.4171    &   0.9890 \\\quad
		&         LRTA          &  47.6156   &  1.6105   &   1.5690  &  0.9882  \\\quad
		&          ZSL         &  46.6373   &  1.4030   &  1.3796   & 0.9908   \\\quad
		&         B-STE          &   44.3812  &  1.9979   &  2.2201   &  0.9764  \\\quad
		&           CSTF        &  46.1899   &   1.6083  &   1.5634  &  0.9881  \\\quad
		&        FSTRD           &  43.2479   &   2.5325  &  2.6075   &  0.9637  \\\quad
		&          BGS-GTF-HSR         &   \textbf{48.9073}  &  \textbf{1.3078}   &  \textbf{1.3065}   & \textbf{0.9911}   \\
		\hline
		\bottomrule
	\end{tabular}
	\label{tab:Houston metrics2}
\end{table}

\begin{table}[t]
	\centering
	\renewcommand{\arraystretch}{1.3}
	\tabcolsep=1.55mm
	\caption{Performance on the Houston2013 dataset with AGK~3 (averaged over 20 scenes)}
	\begin{tabular}{c|c|cccc}
		\toprule
		\hline
		\multirow{2}{*}{\textbf{Setup}} & \multirow{2}{*}{\textbf{Methods}} & \multicolumn{4}{c}{\textbf{Quality Indices}} \\\cline{3-6} &\quad
		&  PSNR$\uparrow$                 & RMSE$\downarrow$    & SAM$\downarrow$    & SSIM$\uparrow$       \\\hline
		\multirow{9}{*}{\textbf{Non-blind}} &        Hysure           &  40.9898   &  2.5609   &   2.3835  &   0.9779 \\\quad& SURE
		&         45.7364          &  1.4738   &   1.4674  &  0.9885      \\
		\quad	&       LTMR            &  48.4784   & 1.3547    & 1.3286    &   0.9901 \\\quad
		&         LRTA          &  45.1557   &  2.4806   &   2.3766  &  0.9753  \\\quad
		&          ZSL         &  46.4985   &  1.4322   &  1.4222   & \textbf{0.9918}   \\\quad
		&         STEREO          &   38.2735  &  4.6097   &  4.2564   &  0.9429  \\\quad
		&           CSTF        &  37.8331   &   3.7743  &   2.6188  &  0.9733  \\\quad
		&        FSTRD           &  38.5440   &   4.0143  &  3.5216   &  0.9406  \\\quad
		&          BGS-GTF-HSR         &   \textbf{49.1722}  &  \textbf{{1.2147}}   &  \textbf{1.2538}   & 0.9917   \\
		\hline
		\multirow{9}{*}{\textbf{{Blind}}} &       Hysure           &  40.2239   &  2.8393   &   2.6308  &   0.9756 \\\quad& SURE
		&         45.7367          &  1.4738   &   1.4674  &  0.9885      \\
		\quad	&       LTMR            &  47.7609   & 1.4880    & 1.3899    &   0.9893 \\\quad
		&         LRTA          &  45.3971   &  2.3589   &   2.1870  &  0.9786  \\\quad
		&          ZSL         &  46.4985   &  1.4322   &  1.4222   & 0.9910   \\\quad
		&         B-STE          &   44.1948  &  2.0406   &  2.2681   &  0.9757  \\\quad
		&           CSTF        &  41.6336   &   2.4721  &   2.0977  &  0.9830  \\\quad
		&        FSTRD           &  40.4948   &   3.1464  &  3.0553   &  0.9566  \\\quad
		&          BGS-GTF-HSR         &   \textbf{48.9543}  &  \textbf{1.2798}   &  \textbf{1.2898}   & \textbf{0.9913}   \\
		\hline
		\bottomrule
	\end{tabular}
	\label{tab:Houston metrics3}
\end{table}

\begin{table}[t]
	\centering
	\renewcommand{\arraystretch}{1.3}
	\tabcolsep=1.55mm
	\caption{Performance on the Houston2013 dataset with AGK~4 (averaged over 20 scenes)}
	\begin{tabular}{c|c|cccc}
		\toprule
		\hline
		\multirow{2}{*}{\textbf{Setup}} & \multirow{2}{*}{\textbf{Methods}} & \multicolumn{4}{c}{\textbf{Quality Indices}} \\\cline{3-6} &\quad
		&  PSNR$\uparrow$                 & RMSE$\downarrow$    & SAM$\downarrow$    & SSIM$\uparrow$       \\\hline
		\multirow{9}{*}{\textbf{Non-blind}} &        Hysure           &  40.6767   &  2.6430   &   2.3746  &   0.9757 \\\quad& SURE
		&         45.9881          &  1.4210   &   1.4318  &  0.9894      \\
		\quad	&       LTMR            &  48.5138   & 1.3441    & 1.3206    &   0.9903 \\\quad
		&         LRTA          &  43.9408   &  3.1656   &   2.9178  &  0.9657  \\\quad
		&          ZSL         &  46.6561   &  1.4023   &  1.3884   & 0.9917   \\\quad
		&         STEREO          &   34.2605  &  7.7812   &  6.7030   &  0.9104  \\\quad
		&           CSTF        &  33.1642   &   6.0793  &   3.6365  &  0.9313  \\\quad
		&        FSTRD           &  34.8040   &   6.1536  &  4.7125   &  0.9042  \\\quad
		&          BGS-GTF-HSR         &   \textbf{49.1804}  &  \textbf{{1.2131}}   &  \textbf{1.2524}   & \textbf{0.9919}   \\
		\hline
		\multirow{9}{*}{\textbf{{Blind}}} &       Hysure           &  39.9899   &  2.9175   &   2.6524  &   0.9742 \\\quad& SURE
		&         45.9881          &  1.4210   &   1.4318  &  0.9894      \\
		\quad	&       LTMR            &  47.8431   & 1.4443    & 1.3709    &   0.9895 \\\quad
		&         LRTA          &  45.4547   &  2.3440   &   2.1988  &  0.9783  \\\quad
		&          ZSL         &  46.6561   &  1.4023   &  1.3884   & 0.9914   \\\quad
		&         B-STE          &   44.1122  &  2.0621   &  2.2928   &  0.9752  \\\quad
		&           CSTF        &  40.5917   &   2.7768  &   2.2238  &  0.9798  \\\quad
		&        FSTRD           &  39.5407   &   3.5931  &  3.5120   &  0.9436  \\\quad
		&          BGS-GTF-HSR         &   \textbf{49.0451}  &  \textbf{1.2468}   &  \textbf{1.2685}   & \textbf{0.9916}   \\
		\hline
		\bottomrule
	\end{tabular}
	\label{tab:Houston metrics4}
\end{table}

\begin{figure}[t]
	\centering
	\setlength{\tabcolsep}{0.3mm}
	\begin{tabular}{m{0.19\figurewidth}m{0.19\figurewidth}m{0.19\figurewidth}m{0.19\figurewidth}m{0.19\figurewidth}}
		\includegraphics[width=0.19\figurewidth]{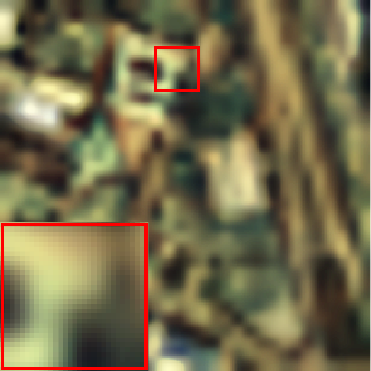} &
		\includegraphics[width=0.19\figurewidth]{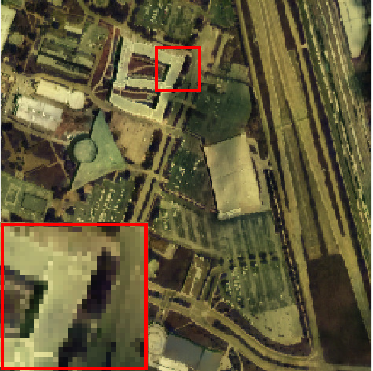}  &
		\includegraphics[width=0.19\figurewidth]{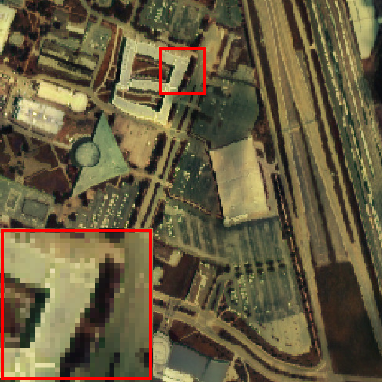}    &
		\includegraphics[width=0.19\figurewidth]{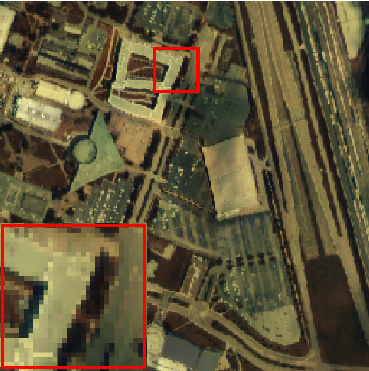}   &
		\includegraphics[width=0.19\figurewidth]{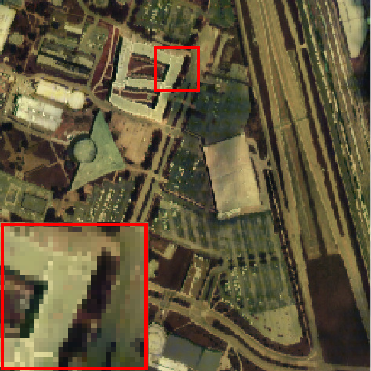}\\
		\includegraphics[width=0.19\figurewidth]{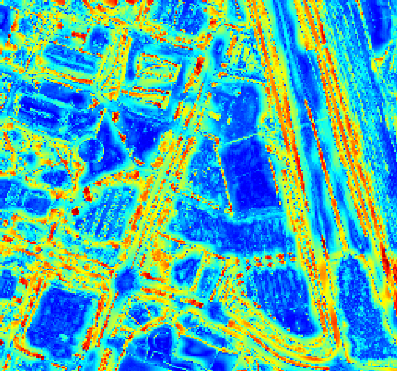} &
		\includegraphics[width=0.19\figurewidth]{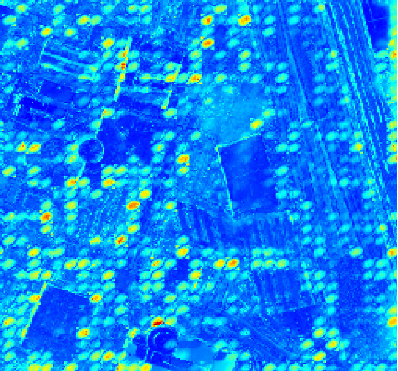}  &
		\includegraphics[width=0.19\figurewidth]{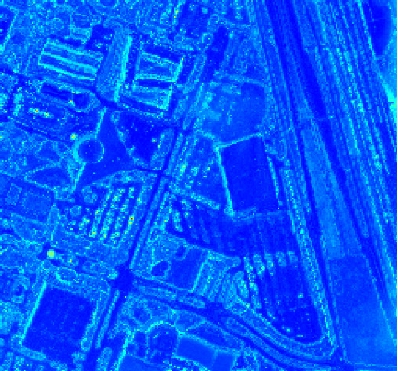}    &
		\includegraphics[width=0.19\figurewidth]{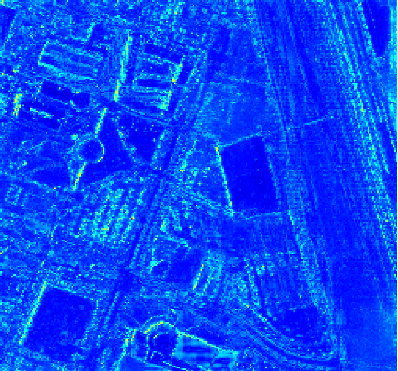}   &
		\includegraphics[width=0.19\figurewidth]{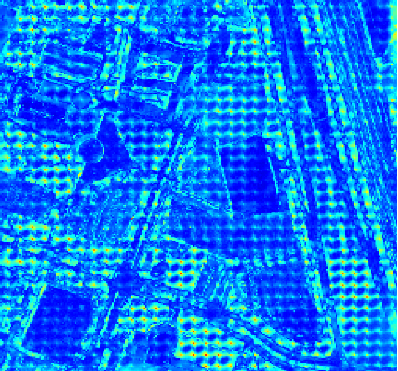} \\
		\multicolumn{1}{c}{\footnotesize{Bicubic}}                                       &\multicolumn{1}{c}{\footnotesize{Hysure}}                                       & \multicolumn{1}{c}{\footnotesize{{SURE}}} & \multicolumn{1}{c}{\footnotesize{LTMR}} & \multicolumn{1}{c}{\footnotesize{LRTA}}
		\\
		\includegraphics[width=0.19\figurewidth]{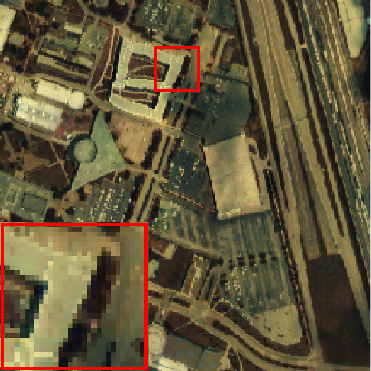} &
		\includegraphics[width=0.19\figurewidth]{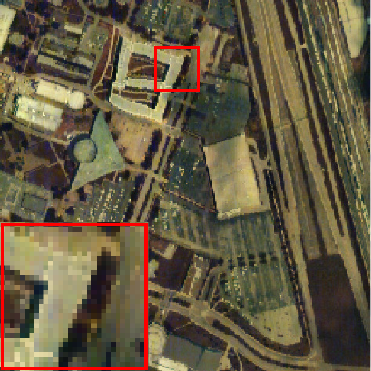}  &
		\includegraphics[width=0.19\figurewidth]{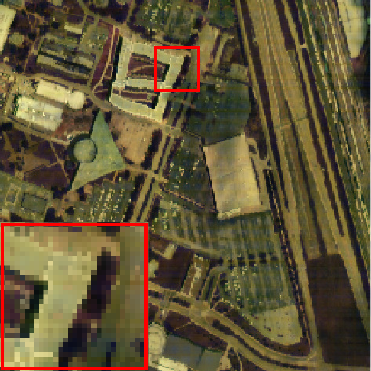}    &
		\includegraphics[width=0.19\figurewidth]{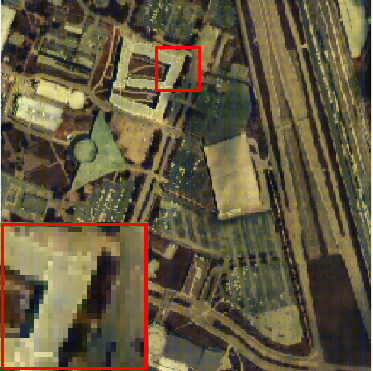}   &
		\includegraphics[width=0.19\figurewidth]{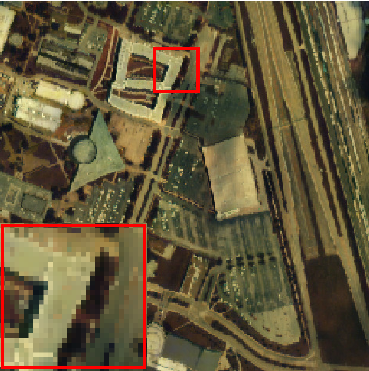}\\
		\includegraphics[width=0.19\figurewidth]{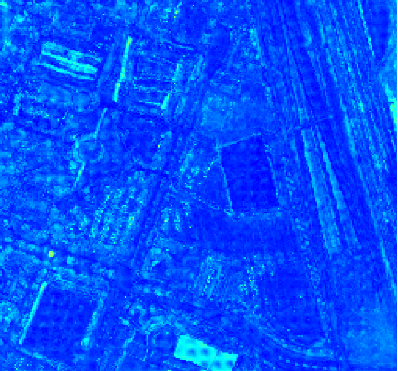} &
		\includegraphics[width=0.19\figurewidth]{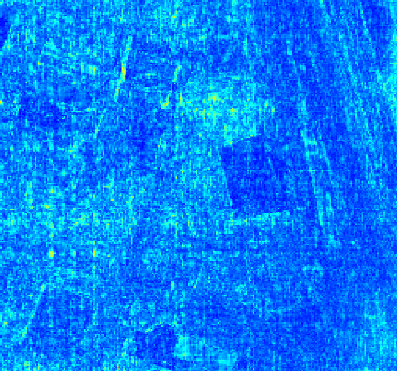}  &
		\includegraphics[width=0.19\figurewidth]{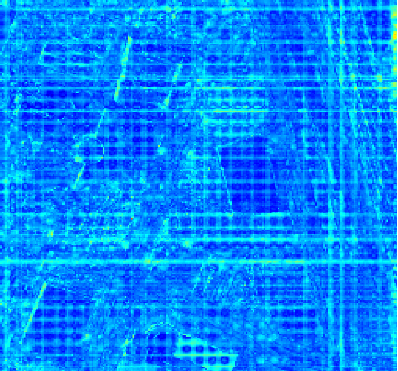}    &
		\includegraphics[width=0.19\figurewidth]{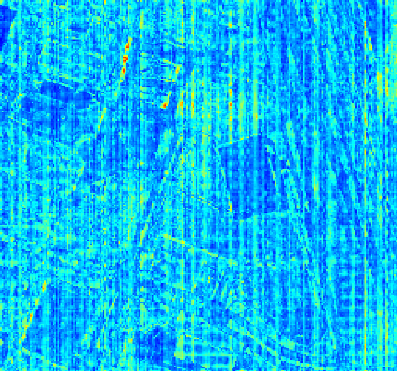}   &
		\includegraphics[width=0.19\figurewidth]{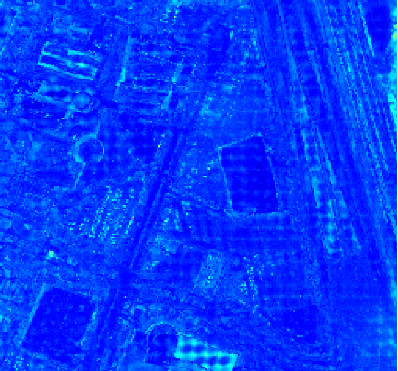} \\
		\multicolumn{1}{c}{\footnotesize{ZSL}}                                       &\multicolumn{1}{c}{\footnotesize{B-STE}}                                       & \multicolumn{1}{c}{\footnotesize{CSTF}} & \multicolumn{1}{c}{\footnotesize{FSTRD}} & \multicolumn{1}{c}{\footnotesize{BGS-GTF-HSR}} \\
		\multicolumn{5}{c}{\includegraphics[width=0.6\linewidth]{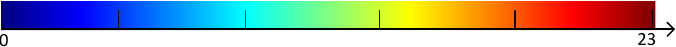}}
	\end{tabular}
	\caption{\label{fig:Houston visualization 3} Blind fusion results and error maps for the Houston2013 dataset with AGK~3 (Scene 17). {Pseudo-color is composed of bands 30, 20 and 10.}}
\end{figure}

\begin{figure}[bthp!]
	\centering
	\setlength{\tabcolsep}{0.15mm}
	\begin{tabular}{m{0.19\figurewidth}m{0.19\figurewidth}m{0.19\figurewidth}m{0.19\figurewidth}m{0.19\figurewidth}}
		\includegraphics[width=0.19\figurewidth]{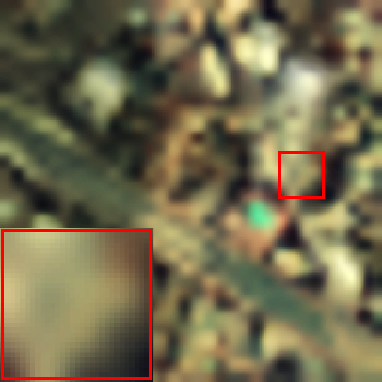} &
		\includegraphics[width=0.19\figurewidth]{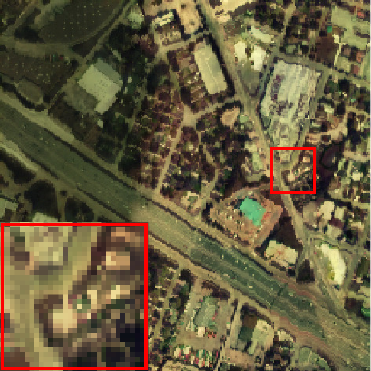}  &
		\includegraphics[width=0.19\figurewidth]{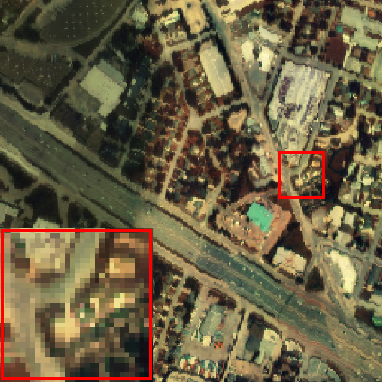}    &
		\includegraphics[width=0.19\figurewidth]{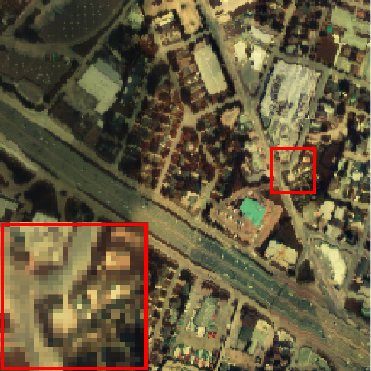}   &
		\includegraphics[width=0.19\figurewidth]{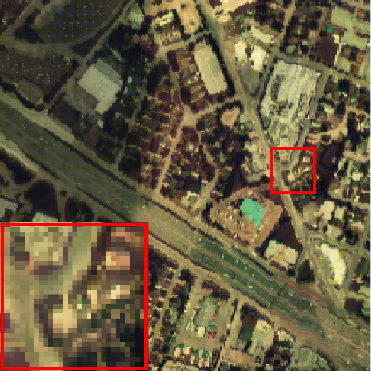}\\
		\includegraphics[width=0.19\figurewidth]{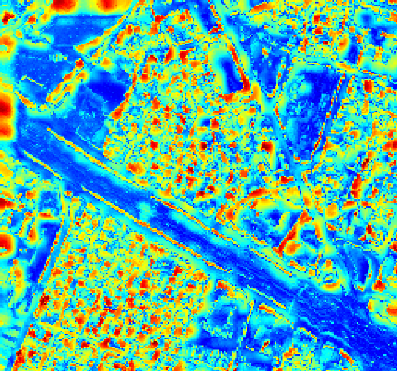} &
		\includegraphics[width=0.19\figurewidth]{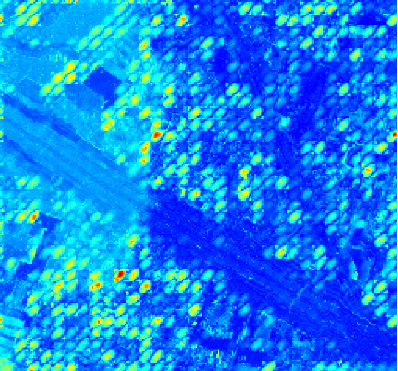}  &
		\includegraphics[width=0.19\figurewidth]{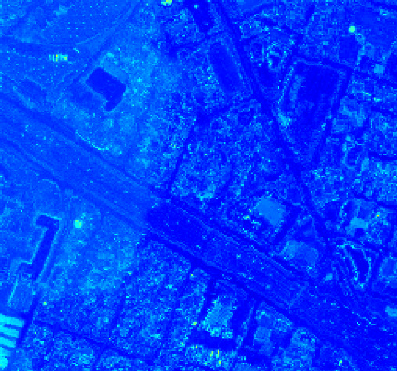}    &
		\includegraphics[width=0.19\figurewidth]{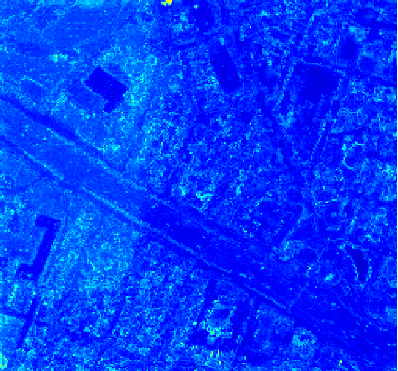}   &
		\includegraphics[width=0.19\figurewidth]{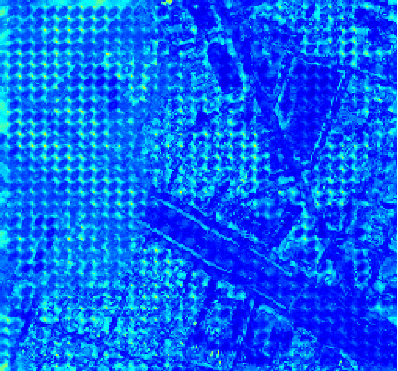} \\
		\multicolumn{1}{c}{\footnotesize{Bicubic}}                                       &\multicolumn{1}{c}{\footnotesize{Hysure}}                                       & \multicolumn{1}{c}{\footnotesize{{SURE}}} & \multicolumn{1}{c}{\footnotesize{LTMR}} & \multicolumn{1}{c}{\footnotesize{LRTA}}
		\\
		\includegraphics[width=0.19\figurewidth]{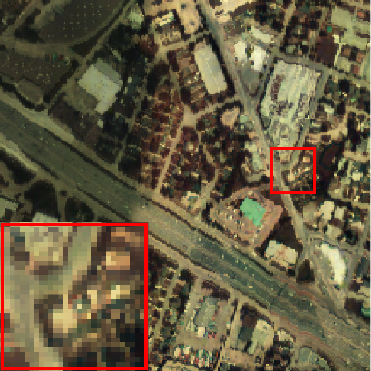} &
		\includegraphics[width=0.19\figurewidth]{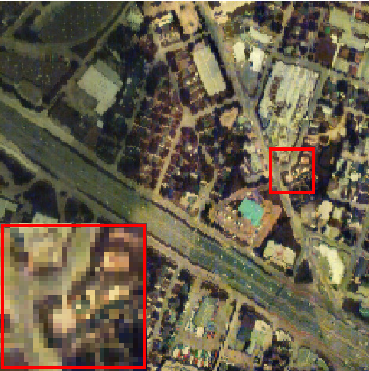}  &
		\includegraphics[width=0.19\figurewidth]{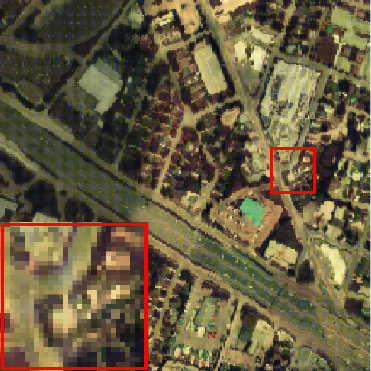}    &
		\includegraphics[width=0.19\figurewidth]{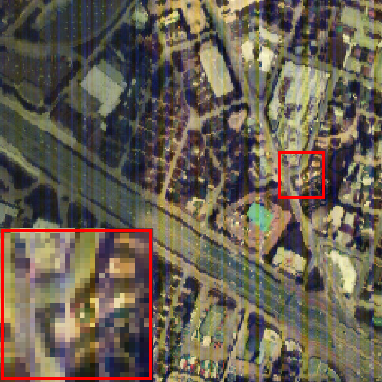}   &
		\includegraphics[width=0.19\figurewidth]{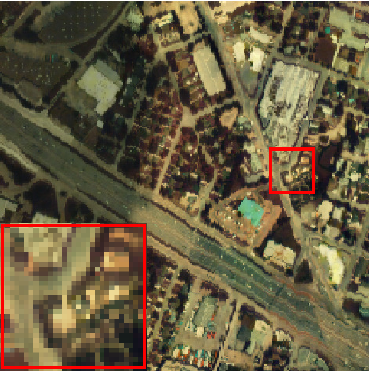}\\
		\includegraphics[width=0.19\figurewidth]{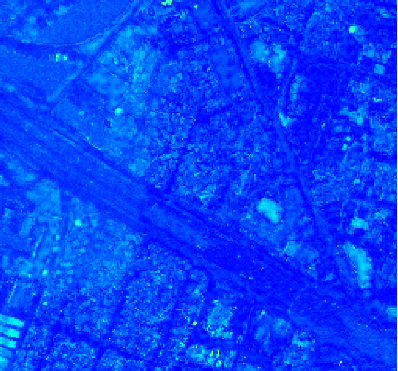} &
		\includegraphics[width=0.19\figurewidth]{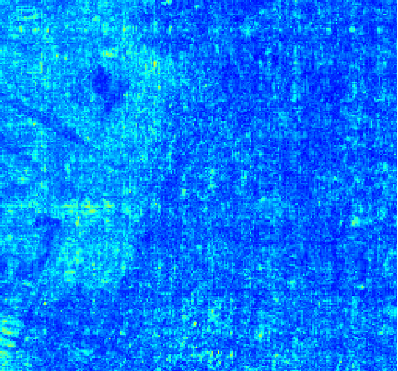}  &
		\includegraphics[width=0.19\figurewidth]{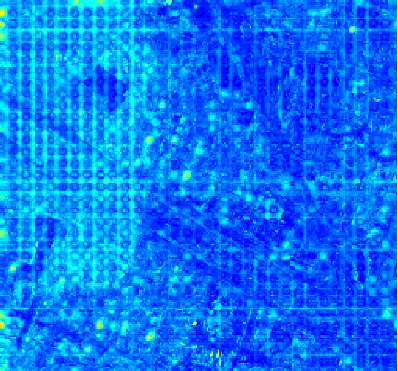}    &
		\includegraphics[width=0.19\figurewidth]{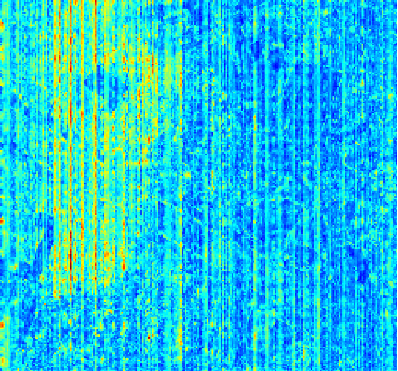}   &
		\includegraphics[width=0.19\figurewidth]{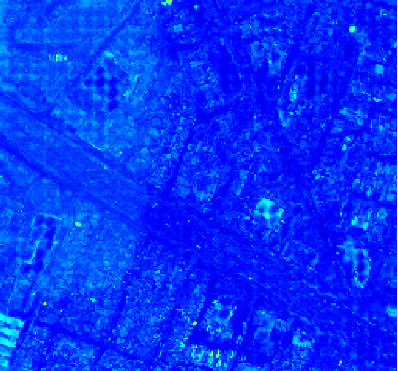} \\
		\multicolumn{1}{c}{\footnotesize{ZSL}}                                       &\multicolumn{1}{c}{\footnotesize{B-STE}}                                       & \multicolumn{1}{c}{\footnotesize{CSTF}} & \multicolumn{1}{c}{\footnotesize{FSTRD}} & \multicolumn{1}{c}{\footnotesize{BGS-GTF-HSR}} \\
		\multicolumn{5}{c}{\includegraphics[width=0.6\linewidth]{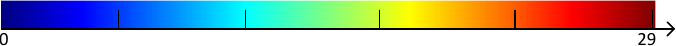}}
	\end{tabular}
	\caption{\label{fig:Houston visualization 4} Blind fusion results and error maps for the Houston2013 dataset with AGK~4 (Scene 5). {Pseudo-color is composed of bands 30, 20 and 10.}}
\end{figure}
We use the Houston2013 dataset to evaluate HSR performance under AGK
spatial blurring.  In order to exam the impact of the blurring
kernel's rank on fusion results, we spatially degenerate each SRI
using four different AGKs of the size $9\times9$, each with a
different condition number, then downsample the blurred images by a
factor of 8 to obtain four $32\times32\times144$ HSIs. The four AGKs
are depicted in Fig.~\ref{fig:AGK} and will be referred to as AGK~1--4
hereafter.

Results for non-blind experiments with AGK~1--4 are presented in
Figs.~\ref{fig:Houston visualization 1} and \ref{fig:Houston
	visualization 2} as well as Tables~\ref{tab:Houston
	metrics1}--\ref{tab:Houston metrics4}.  The most salient phenomenon
in the visual results is the systematic failure of the TF-based
STEREO, CSTF, and FSTRD under anisotropic blurring.  Since the AGK is
not rank~1, the assumptions underlying the TF framework deviate from
the real spatial degradation, resulting in dramatically degraded
performance relative to the MF- and GTF-based approaches.  A similar
conclusion can be drawn from the quantitative results in
Tables~\ref{tab:Houston metrics1}--\ref{tab:Houston metrics4}.
Additionally, Fig.~\ref{fig:PSNRvsAGK} depicts how performance changes
as the rank of the AGK increases. We see that the MF-based techniques
along with the proposed BGS-GTF-HSR are largely resilient to
increasing rank, whereas the TF-based techniques suffer greatly.

The results for blind experiments are presented in
Figs.~\ref{fig:Houston visualization 3} and \ref{fig:Houston
	visualization 4} as well as in Tables \ref{tab:Houston
	metrics1}--\ref{tab:Houston metrics4}.  Again, the TF-based methods
largely fail to adequately handle the anisotropic blurring.
Interestingly, however, the TF-based methods do better in the blind
experiments than they do in the non-blind experiments, particularly for
AGK~3 and AGK~4. While this appears counter-intuitive, these results
imply that the spatial-degradation matrices used in TF-based fusions
do not necessarily have to match the real degradation. Nevertheless,
the proposed BGS-GTF-HSR consistently outperforms all the other
techniques for the blind-fusion scenario.
\subsection{HSR on a Real Dataset}
\begin{figure*}[t]
	\centering
	\setlength{\tabcolsep}{0.3mm}
	\begin{tabular}{m{0.38\figurewidth}m{0.38\figurewidth}m{0.38\figurewidth}m{0.38\figurewidth}m{0.38\figurewidth}}
		\includegraphics[width=0.38\figurewidth]{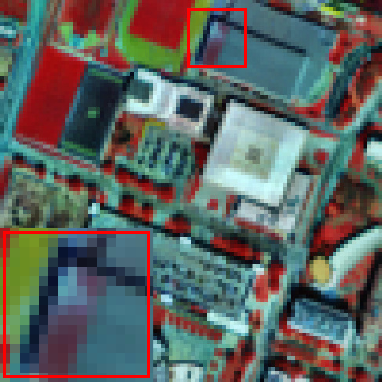} &
		\includegraphics[width=0.38\figurewidth]{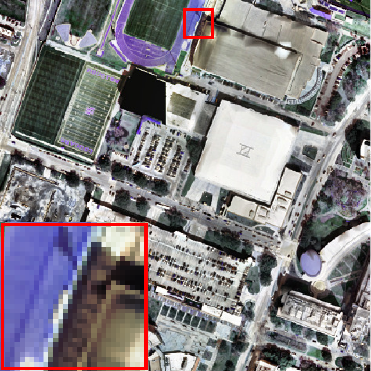} &
		\includegraphics[width=0.38\figurewidth]{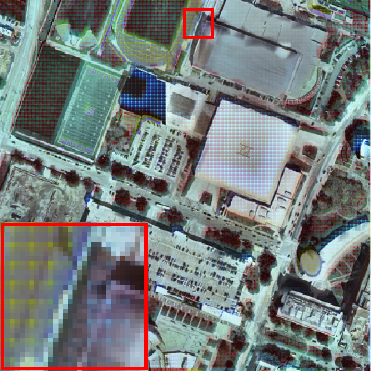}  &
		\includegraphics[width=0.38\figurewidth]{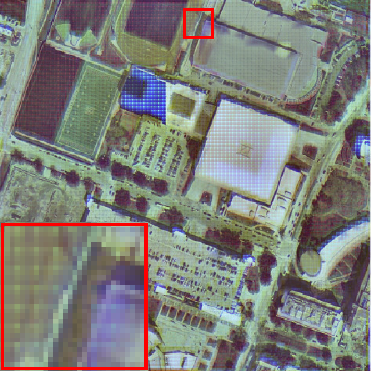}   &
		\includegraphics[width=0.38\figurewidth]{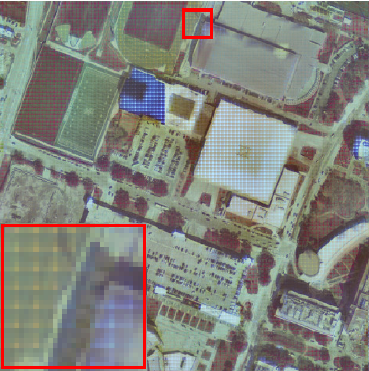}\\
		\multicolumn{1}{c}{HSI}&\multicolumn{1}{c}{MSI}                                       &\multicolumn{1}{c}{\footnotesize{Hysure}}                                       &  \multicolumn{1}{c}{\footnotesize{LTMR}} & \multicolumn{1}{c}{\footnotesize{LRTA}}
		\\
		\includegraphics[width=0.38\figurewidth]{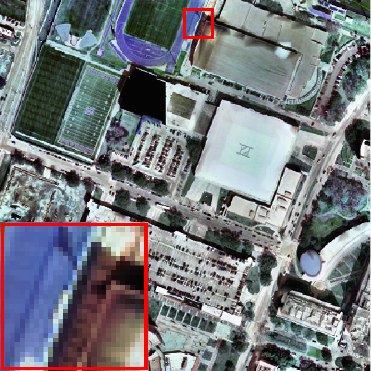} &
		\includegraphics[width=0.38\figurewidth]{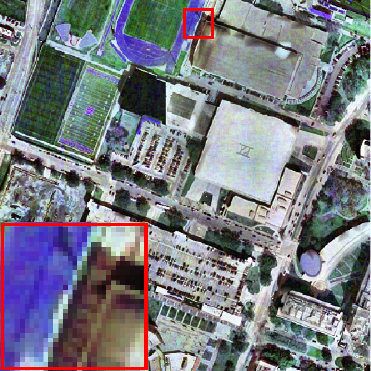}  &
		\includegraphics[width=0.38\figurewidth]{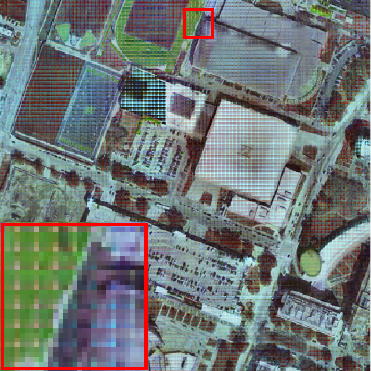}    &
		\includegraphics[width=0.38\figurewidth]{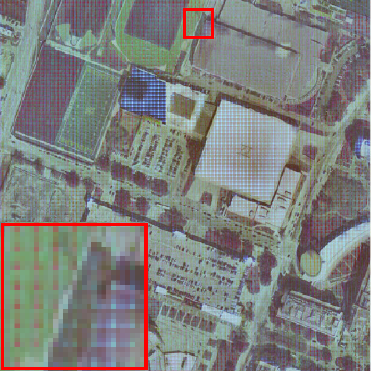}   &
		\includegraphics[width=0.38\figurewidth]{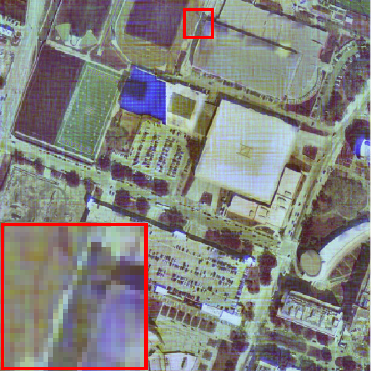} \\
		\multicolumn{1}{c}{\footnotesize{ZSL}}                                       &\multicolumn{1}{c}{\footnotesize{B-STE}}                                       & \multicolumn{1}{c}{\footnotesize{CSTF}} & \multicolumn{1}{c}{\footnotesize{FSTRD}} & \multicolumn{1}{c}{\footnotesize{BGS-GTF-HSR}} 
	\end{tabular}
	\caption{\label{fig:Houston Real visualization} Fusion results on the Houston2018 dataset; pseudo-color is composed of bands 30, 20 and 10.}
\end{figure*}

\begin{figure}[bthp!]
	\centering
	\setlength{\tabcolsep}{0.25mm}
	\begin{tabular}{cm{0.45\figurewidth}cm{0.45\figurewidth}}\rotatebox[origin=c]{90}{\footnotesize{Reflectance}}&\includegraphics[width=0.45\figurewidth]{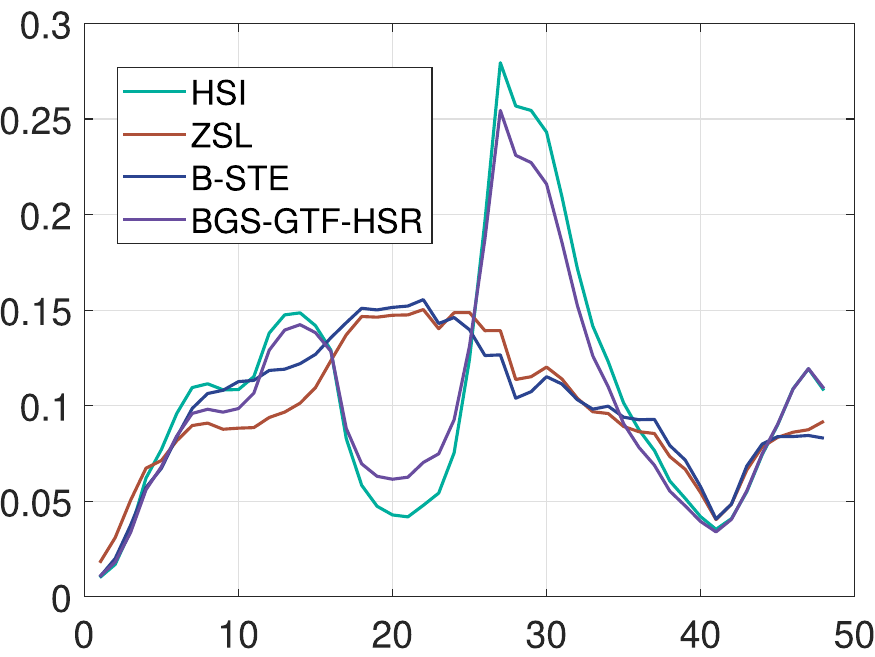}&\rotatebox[origin=c]{90}{\footnotesize{Reflectance}}&
		\includegraphics[width=0.45\figurewidth]{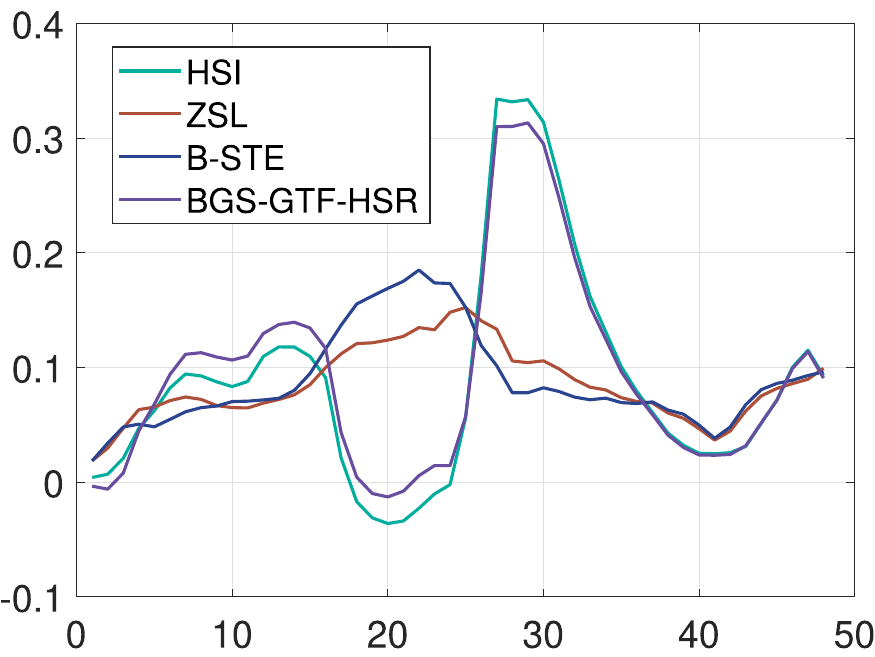} \\\quad&
		\multicolumn{1}{c}{\footnotesize{Band}}&\quad&\multicolumn{1}{c}{\footnotesize{Band}}
	\end{tabular}
	\caption{\label{fig:Spectral curves} Two spectral curves from the Houston2018 dataset.}
\end{figure}

The fusion results are depicted in Fig.~\ref{fig:Houston Real
	visualization} as pseudo-color images for the real Houston2018
HSI/MSI pair. Since there exists no ground truth, we can compare only
the visual results from the perspectives of spatial enhancement and
spectral fidelity. On the one hand, one can observe strong artifacts
from the zoomed-in areas in Fig.~\ref{fig:Houston Real visualization}
for Hysure, LTMR, LRTA, CSTF, and FSTRD, yet, such artifacts are much
less obvious for BGS-GTF-HSR.  If, on the other hand, the pseudo-color
results of ZSL and B-STE appear finer and smoother, it is only because
they have actually overfit the spatial information of the MSI but
failed to maintain the spectral information of the HSI. To show this,
we graph some spectral-pixel curves in Fig.~\ref{fig:Spectral curves}
and see that the curves for ZSL and B-STE severely deviate from those
of the HSI whereas the curves for BGS-GTF-HSR are much closer to their
HSI counterparts.

\section{Conclusion}
\label{sec:conclusions}
In this paper, we proposed a generalization of the TF-based HSR
problem in the form of GTF-HSR. We demonstrated
that this generalization can model arbitrary forms of blurring
kernels, whereas previous TF-HSR approaches cannot handle situations
wherein the blurring kernel was not rank-1. We also addressed the
recoverability of the proposed GTF-HSR, showing that exact recovery is
guaranteed.  To establish an algorithmic framework for practical
HSR, we
proposed a blockwise-sparse regularizer to fully exploit the group
sparsity of higher-order tensors, which is achieved through a
block-based unfolding strategy and the $l_{2,0}$-norm. With the
nonconvex surrogate imposed, the overall problem was optimized via
ADMM in a divide-and-conquer manner to ease the intrinsic difficulty
of large-scale nonconvex optimization for multi-source data
reconstruction. We tested our proposed algorithmic framework, called
BGS-GTF-HSR, on simulated datasets under traditional isotropic
Gaussian blurring as well as more realistic anisotropic Gaussian
blurring and also on real HSI-MSI pairs. Experimental results
demonstrated that our proposed BGS-GTF-HSR outperformed the TF-HSR
methods considered under anisotropic blurring in simulated problems
for both blind as well as non-blind HSR. Additionally, superior
results were observed for BGS-GTF-HSR for real-world data as well.
\allowdisplaybreaks[4]
\section{Supplemental Material}
\subsection{Proofs of Propositions, Lemmas, Theorems, and Corollaries in the Paper}
In this section, we provided detailed proofs of propositions, lemmas,
theorems, and corollaries found in the main text.  In this
supplemental material, $\mathbf{A}^{\dagger}$ represents the
Moore-Penrose pseudo-inverse of $\mathbf{A}$, and $\oplus$ denotes the
direct sum of two spaces. For a $3$-way tensor
$\mathcal{T}\in\mathbb{R}^{M\times N\times L}$, $\mathbf{T}_{M\times
	N,l}$ denotes its $l^{\text{th}}$ frontal slice, and
$\mathcal{T}[m,n,:]$ denotes its pixel vector at spatial location
$(m,n)$. Furthermore, $\mathbf{T}[m,n]$ denotes the $(m,n)$ element of
matrix $\mathbf{T}$, $\mathbf{t}[m]$ denotes the $m^{\text{th}}$
element of vector $\mathbf{t}$, and $\mathbf{t}_m$ denotes the
$m^{\text{th}}$ column vector of matrix $\mathbf{T}$. We use
$\mathbf{I}_N$ and $\mathbf{L}_N$ to indicate the $N \times N$
identity and commutation matrices, respectively.  Finally,
$\mathsf{P}(\cdot)$ stands for the probability function.
\subsection{Proof of Prop.~\ref{Prop:Kr-Rank}}
In proving Prop.~\ref{Prop:Kr-Rank},
we first introduce Lemma~\ref{Le:Lemma for Kronecker}
followed by Corollary \ref{coro:Kr-Rank}.
\begin{lemma}
	\textit{For arbitrary matrices $\mathbf{M}_1\in \mathbb{R} ^{J_1\times K_1}$ and $\mathbf{M}_2\in \mathbb{R} ^{J_2\times K_2}$, there exists a reversible linear mapping $\mathsf{F}(\cdot|J_1,J_2,K_1,K_2)$, such that
		\begin{equation}
			\mathbf{M}_1\otimes\mathbf{M}_2=\mathsf{F}\left(\mathsf{Vec}\left(\mathbf{M}_1\right)\mathsf{Vec}^T\left(\mathbf{M}_2\right)|J_1,J_2,K_1,K_2 \right)
		\end{equation}
		and
		\begin{equation}
			\mathsf{Vec}\left(\mathbf{M}_1\right)\mathsf{Vec}^T\left(\mathbf{M}_2\right)=\mathsf{F}^{-1}\left(\mathbf{M}_1\otimes\mathbf{M}_2|J_1,J_2,K_1,K_2\right)
		\end{equation}
		where
		\begin{align}
			\mathsf{F}&\left(\cdot |J_1,J_2,K_1,K_2\right) =\notag \\&\mathsf{Unv}\left[\left(\mathbf{I}_{K_1}\otimes\mathbf{L}_{K_2J_1}\otimes\mathbf{I}_{J_2}\right)\mathsf{Vec}\left(\left(\cdot\right)^T \right)|J_1J_2,K_1K_2 \right]   , \notag\\
			\mathsf{F}^{-1}&\left(\cdot |J_1,J_2,K_1,K_2\right)=\notag\\&\mathsf{Unv}^T\left[\left(\mathbf{I}_{K_1}\otimes\mathbf{L}_{J_1K_2}\otimes\mathbf{I}_{J_2}\right)\mathsf{Vec}\left(\cdot \right)|J_1K_1,J_2K_2 \right].
		\end{align}
	}
	\label{Le:Lemma for Kronecker}
\end{lemma}
\begin{proof}
	One can see the existence of this map and its reversibility by verifying that
	\begin{align}
		\mathsf{Vec}&\left(\mathbf{M}_1\otimes\mathbf{M}_2\right)\notag\\&=\left(\mathbf{I}_{K_1}\otimes\mathbf{L}_{K_2J_1}\otimes\mathbf{I}_{J_2}\right)\left(\mathsf{Vec}\left(\mathbf{M}_1\right)\otimes\mathsf{Vec}\left(\mathbf{M}_2\right)\right)\nonumber\\&=\left(\mathbf{I}_{K_1}\otimes\mathbf{L}_{K_2J_1}\otimes\mathbf{I}_{J_2}\right)\mathsf{Vec}\left(\mathsf{Vec}\left(\mathbf{M}_2\right)\mathsf{Vec}^T\left(\mathbf{M}_1\right)\right),
	\end{align}
	and its linearity simply follows from the linearity of $\mathsf{Vec}\left(\cdot\right),\,\mathsf{Unv}\left(\cdot\right)$, and matrix multiplication.
\end{proof}
\begin{corollary}
	\textit{For any matrix $\mathbf{W}\in\mathbb{R}^{J_1J_2\times
			K_1K_2}$, its Kronecker rank is
		as
		\begin{equation}
			\mathsf{kr}_{\mathbf{W}}=\operatorname{rank}\left\{\mathsf{F}^{-1}\left(\mathbf{W}|J_1,J_2,K_1,K_2\right)\right\}.
	\end{equation}}
	\label{coro:Kr-Rank}
\end{corollary}
\begin{proof}
	Letting
	$R=\operatorname{rank}\left\{\mathsf{F}^{-1}(\mathbf{W}|J_1,J_2,K_1,K_2)\right\}$,
	there exist $\mathbf{A}\in\mathbb{R}^{J_1K_1\times R}$ and
	$\mathbf{B}\in\mathbb{R}^{ J_2K_2\times R}$ such
	that
	\begin{equation}
		\mathsf{F}^{-1}(\mathbf{W}|J_1,J_2,K_1,K_2)=\mathbf{A}\mathbf{B}^T.
	\end{equation}
	Let
	$\mathbf{M}_1^{(r)}=\mathsf{Unv}\left(\mathbf{a}_r|J_1,K_1\right)$
	and
	$\mathbf{M}_2^{(r)}=\mathsf{Unv}\left(\mathbf{b}_r|J_2,K_2\right)$,
	where $\mathbf{a}_r$ and $\mathbf{b}_r$ are the $r^{\text{th}}$ columns of
	matrices $\mathbf{A}$ and $\mathbf{B}$, respectively. We then have
	\begin{align}
		\mathbf{W}&=\mathsf{F}\left(\mathsf{F}^{-1}(\mathbf{W}|J_1,J_2,K_1,K_2)|J_1,J_2,K_1,K_2\right)\nonumber\\&=\mathsf{F}\left(\mathbf{A}\mathbf{B}^T|J_1,J_2,K_1,K_2\right)\nonumber\nonumber\\&=\mathsf{F}\left(\sum_{r=1}^R\mathbf{a}_r\mathbf{b}_r^T|J_1,J_2,K_1,K_2\right)\nonumber\\&=\sum_{r=1}^R\mathsf{F}\left(\mathbf{a}_r\mathbf{b}_r^T|J_1,J_2,K_1,K_2\right)\nonumber\\&=\sum_{r=1}^R\mathbf{M}_1^{(r)}\otimes\mathbf{M}_2^{(r)} ,
	\end{align}
	which implies that $\mathsf{kr}_{\mathbf{W}}\leq
	\operatorname{rank}\left\{\mathsf{F}^{-1}(\mathbf{W}|J_1,J_2,K_1,K_2)\right\}$.
	On the other hand, supposing the ``tightest'' KD of $\mathbf{W}$ is
	$\mathbf{W}=\sum_{r=1}^{\mathsf{kr}_{\mathbf{W}}}\mathbf{M}_1^{(r)}\otimes\mathbf{M}_2^{(r)}$,
	we then have
	\begin{align}
		\mathsf{F}^{-1}(\mathbf{W}&|J_1,J_2,K_1,K_2)\nonumber\\&=\mathsf{F}^{-1}(\sum_{r=1}^{\mathsf{kr}_{\mathbf{W}}}\mathbf{M}_1^{(r)}\otimes\mathbf{M}_2^{(r)}|J_1,J_2,K_1,K_2)\nonumber\\&=\sum_{r=1}^{\mathsf{kr}_{\mathbf{W}}}\mathsf{F}^{-1}(\mathbf{M}_1^{(r)}\otimes\mathbf{M}_2^{(r)}|J_1,J_2,K_1,K_2)\nonumber\\&=\sum_{r=1}^{\mathsf{kr}_{\mathbf{W}}}\mathsf{Vec}\left(\mathbf{M}_1^{(r)}\right)\mathsf{Vec}^T\left(\mathbf{M}_2^{(r)}\right)
	\end{align}
	which implies that $\mathsf{kr}_{\mathbf{W}}\geq
	\operatorname{rank}\left\{\mathsf{F}^{-1}(\mathbf{W}|J_1,J_2,K_1,K_2)\right\}$.
	Thus,
	it is proven that
	$\mathsf{kr}_{\mathbf{W}}=\operatorname{rank}\left\{\mathsf{F}^{-1}\left(\mathbf{W}|J_1,J_2,K_1,K_2\right)\right\}$.
\end{proof}

\noindent
Next, we start to prove Prop.~\ref{Prop:Kr-Rank}. We first derive
the explicit expression of spatial-degradation matrix $\mathbf{D}$
then complete the proof by taking~Cor.~\ref{coro:Kr-Rank} into
account.
\setcounter{proposition}{0} 
\renewcommand{\theproposition}{\ref{Prop:Kr-Rank}}
\begin{proposition}
	\textit{For the spatial-degradation matrix $\mathbf{D}$ in
		\eqref{eq: Matrix Formulation}, since it is physically modeled as
		\eqref{eq:Physical Expression}, we
		have
		\begin{equation}
			\mathsf{kr}_{\mathbf{D}}=\operatorname{rank}\left\{\mathbf{\Phi}\right\}.
		\end{equation}
	}
\end{proposition}
\begin{proof}
	To prove Prop.~\ref{Prop:Kr-Rank}, we give an
	explicit expression of $\mathbf{D}$ in terms of the blurring kernel
	$\mathbf{\Phi}$ and the specific downsampling strategy. First,
	suppose the rank decomposition of $\mathbf{\Phi}$ is
	\begin{equation}
		\mathbf{\Phi}=\mathbf{U}\mathbf{V}^T=\sum_{r=1}^{R}\mathbf{u}_r\mathbf{v}^T_r ,
	\end{equation}
	where
	$\mathbf{U},\mathbf{V}\subseteq\mathbb{R}^{\phi\times
		R}$ are matrices with full column rank, and
	$R=\operatorname{rank}\left\{\mathbf{\Phi}\right\}$.
	Then, let $\mathbf{u}_r$ and
	$\mathbf{v}_r$ be the $r^{\text{th}}$ column vectors of $\mathbf{U}$ and
	$\mathbf{V}$, respectively. Subsequently, defining
	$\widetilde{\mathcal{Z}}=\mathcal{Z}\ast\mathbf{\Phi}$, we have
	\begin{align}
		\widetilde{\mathcal{Z}}[m,n,:]&=\sum_{l,s=1}^{\phi,\phi}\,\mathbf{\Phi}[l,s]\mathcal{Z}[m-m'+l,n-n'+s,:]\nonumber\\&=\sum_{l,s=1}^{\phi,\phi}\,\sum_{r=1}^{R}\mathbf{u}_r[l]\mathbf{v}_r[s]\mathcal{Z}[m-m'+l,n-n'+s,:]\nonumber\\&=\sum_{r=1}^{R}\sum_{l,s=1}^{\phi,\phi}\,\mathbf{u}_r[l]\mathcal{Z}[m-m'+l,n-n'+s,:]\mathbf{v}_r[s] ,
	\end{align}
	where $m',n'$ are the shift parameters {\cite{KFS2018}}. This actually implies
	that
	\begin{equation}
		\widetilde{\mathcal{Z}}=\mathcal{Z}\ast
		\mathbf{\Phi}=\sum_{r=1}^{R}\mathcal{Z}\ast\mathbf{u}_r\ast\mathbf{v}^T_r=\sum_{r=1}^{R}\mathcal{Z}\times_1\mathbf{T}^{\mathbf{u}_r}\times_2\mathbf{T}^{\mathbf{v}_r} ,
	\end{equation}
	in which
	$\{\mathbf{T}^{\mathbf{u}_r},\mathbf{T}^{\mathbf{v}_r}\}_{r=1}^R$ are
	circulant matrices generated according to
	$\{\mathbf{u}_r,\mathbf{v}_r\}_{r=1}^R$.  More
	concretely,
	\begin{align}
		\mathbf{T}^{\mathbf{u}_r}=\sum_{l=1}^{\phi}\mathbf{u}_r[l]\mathbf{J}_{M_1}^{M_1-m'+l}, \\
		\mathbf{T}^{\mathbf{v}_r}=\sum_{s=1}^{\phi}\mathbf{v}_r[s]\mathbf{J}_{M_2}^{M_2-n'+s} ,
	\end{align}
	and $\mathbf{J}_M$ denotes the $M\times M$ basic circulant
	matrix,
	\begin{equation}
		\mathbf{J}_M=\begin{bmatrix}
			0&1&0&\cdots&0\\
			0&0&1&\cdots&0\\
			\vdots&\vdots&\vdots&\ddots&\vdots\\
			0&0&0&\cdots&1\\
			1&0&0&\cdots&0
		\end{bmatrix}.
	\end{equation}
	Having expressed the blurring aspect of the spatial degradation,
	we now consider the subsequent
	uniform downsampling, yielding an expression for
	the whole spatial-degradation process,
	\begin{equation}
		\mathcal{X}=(\mathcal{Z}\ast\mathbf{\Phi})_{\downarrow}=\sum_{r=1}^{R}\mathcal{Z}\times_1\mathbf{P}^{(r)}_1\times_2\mathbf{P}^{(r)}_2 ,
	\end{equation}
	where
	\begin{align}
		\mathbf{P}^{(r)}_1=\mathbf{S}_1\mathbf{T}^{\mathbf{u}_r},\\
		\mathbf{P}^{(r)}_2=\mathbf{S}_2\mathbf{T}^{\mathbf{v}_r} ,
	\end{align}
	and the row vectors of $\mathbf{S}_1\in\mathbb{R}^{m_1\times M_1}$ and
	$\mathbf{S}_2\in\mathbb{R}^{m_2\times M_2}$ are sampled from those of
	$\mathbf{I}_{M_1}$ and $\mathbf{I}_{M_2}$, respectively. Combining
	with \eqref{eq:Physical Expression}, we now
	have
	\begin{equation}
		\mathbf{X}_{[3]}=\mathbf{Z}_{[3]}\mathbf{D}=\mathbf{Z}_{[3]}\sum_{r=1}^{R}\left(\mathbf{P}_2^{(r)}\otimes\mathbf{P}_1^{(r)}\right)^T.
	\end{equation}
	Since this equation should hold for any $\cal{Z}$, it is concluded
	that
	\begin{equation}
		\mathbf{D}=\sum_{r=1}^{R}\left(\mathbf{P}_2^{(r)}\otimes\mathbf{P}_1^{(r)}\right)^T.
	\end{equation}
	To proceed, from Cor.~\ref{coro:Kr-Rank}, we have
	\begin{equation}
		\begin{aligned}
			\mathsf{kr}_{\mathbf{D}}&=\operatorname{rank}\left\{\mathsf{F}^{-1}(\mathbf{D}|M_2,M_1,m_2,m_1)\right\}\\&=\operatorname{rank}\left\{\mathsf{F}^{-1}(\sum_{r=1}^{R}\left(\mathbf{P}_2^{(r)}\otimes\mathbf{P}_1^{(r)}\right)^T|M_2,M_1,m_2,m_1)\right\}\\&=\operatorname{rank}\left\{\sum_{r=1}^{R}\mathsf{Vec}\left((\mathbf{P}_2^{(r)})^T\right)\mathsf{Vec}^T\left((\mathbf{P}_1^{(r)})^T\right)\right\}.
		\end{aligned}
	\end{equation}
	Denoting $\left\{\mathbf{q}_j^{(r)}\right\}_{j=1}^{m_2}$ and
	$\left\{\mathbf{p}_i^{(r)}\right\}_{i=1}^{m_1}$ to be the column
	vectors of $(\mathbf{P}_2^{(r)})^T$ and $(\mathbf{P}_1^{(r)})^T$,
	respectively, we can further derive that
	\begin{align}
		\sum_{r=1}^{R}\mathsf{Vec}\left((\mathbf{P}_2^{(r)})^T\right)\mathsf{Vec}^T\left((\mathbf{P}_1^{(r)})^T\right)\nonumber\\=\sum_{r=1}^{R}\begin{bmatrix}
			\mathbf{q}_1^{(r)}\\\mathbf{q}_2^{(r)}\\\vdots\\\mathbf{q}_{m_2}^{(r)}
		\end{bmatrix}\begin{bmatrix}
			{\mathbf{p}^{(r)}_1}^T&{\mathbf{p}_2^{(r)}}^T&\cdots&{\mathbf{p}_{m_2}^{(r)}}^T
		\end{bmatrix}.
	\end{align}
	Due to the fact that
	$\left\{{\mathbf{q}_j^{(r)}}^T\right\}_{j=1}^{m_2}$ and
	$\left\{{\mathbf{p}_i^{(r)}}^T\right\}_{i=1}^{m_1}$ are all sampled
	from the rows of circulant matrices, we further have
	\begin{equation}
		\begin{aligned}
			\operatorname{rank}&\left\{\sum_{r=1}^{R}\begin{bmatrix}
				\mathbf{q}_1^{(r)}\\\mathbf{q}_2^{(r)}\\\vdots\\\mathbf{q}_{m_2}^{(r)}
			\end{bmatrix}\begin{bmatrix}
				{\mathbf{p}^{(r)}_1}^T&{\mathbf{p}_2^{(r)}}^T&\cdots&{\mathbf{p}_{m_2}^{(r)}}^T
			\end{bmatrix}\right\}\\&=\operatorname{rank}\left\{\sum_{r=1}^{R}\begin{bmatrix}
				\mathbf{q}_1^{(r)}
			\end{bmatrix}\begin{bmatrix}
				{\mathbf{p}^{(r)}_1}^T&{\mathbf{p}_2^{(r)}}^T&\cdots&{\mathbf{p}_{m_2}^{(r)}}^T
			\end{bmatrix}\right\}\\&=\operatorname{rank}\left\{\sum_{r=1}^{R}
			\mathbf{q}_1^{(r)}
			{\mathbf{p}^{(r)}_1}^T
			\right\}\\&=\operatorname{rank}\left\{\sum_{r=1}^{R}
			\mathbf{v}_r
			\mathbf{u}_r^T
			\right\}\\&=\operatorname{rank}\left\{\mathbf{\Phi}
			\right\} ,
		\end{aligned}
	\end{equation}
	which completes the proof.
\end{proof}
\subsection{Proof of Thm.~\ref{Th:Recovery Guarantee}}
To prove Thm.~\ref{Th:Recovery Guarantee}, we first introduce the
following lemma.
\begin{lemma}
	\textit{Let an arbitrary three-way tensor $\mathcal{T}\in
		\mathbb{R}^{M\times N\times L}$ be decomposed
		as
		\begin{equation}
			\mathcal{T}=\sum_{r=1}^R\mathcal{C}\times_1\mathbf{A}_r\times_2\mathbf{B}_r ,
		\end{equation}
		where $\mathcal{C}\in\mathbb{R}^{I\times J\times L}$ is drawn from
		an absolutely continuous distribution, and
		$\mathbf{A}_r\in\mathbb{R}^{M\times I}$ and
		$\mathbf{B}_r\in\mathbb{R}^{N\times J}$. Suppose that
		\begin{equation}
			\begin{aligned}
				&I\leq JL,\, J\leq IL,\, L\geq 3,\\
				& \operatorname{rank}\left\lbrace \mathbf{A} \right\rbrace=IR,\\
				& \operatorname{rank}\left\lbrace \mathbf{B} \right\rbrace=JR ,
			\end{aligned}
		\end{equation}
		where 
		\begin{equation}
			\begin{aligned}
				\mathbf{A}&=\left[\mathbf{A}_1,\mathbf{A}_2,\cdots,\mathbf{A}_R \right] ,\\
				\mathbf{B}&=\left[\mathbf{B}_1,\mathbf{B}_2,\cdots,\mathbf{B}_R \right].
			\end{aligned}
		\end{equation}
		Then this decomposition of $\mathcal{T}$ in terms of
		$\mathcal{C}$, $\mathbf{A}$, and $\mathbf{B}$ is essentially unique with
		probability~1.  }
	\label{Le: Essential Uniqueness}
\end{lemma}
\begin{remark}
	Here, essential uniqueness means that we can find alternative
	$\hat{\mathcal{C}}\in\mathbb{R}^{I\times J\times L}$,
	$\hat{\mathbf{A}}_r\in\mathbb{R}^{M\times I}$, and
	$\hat{\mathbf{B}}_r\in\mathbb{R}^{N\times J}$ such
	that
	\begin{equation}
		\mathcal{T}=\sum_{r=1}^R\hat{\mathcal{C}}\times_1\hat{\mathbf{A}}_r\times_2\hat{\mathbf{B}}_r,
	\end{equation}
	only if
	\begin{align}
		\hat{\mathbf{A}}_r&=\mathbf{A}_r\mathbf{\Psi}_{\mathbf{A}}^{(r)}, \\
		\hat{\mathbf{B}}_r&=\mathbf{B}_r\mathbf{\Psi}_{\mathbf{B}}^{(r)},
	\end{align}
	and
	\begin{equation}
		\hat{\mathcal{C}}=\mathcal{C}\times_1 \left(\mathbf{\Psi}_{\mathbf{A}}^{(r)}\right)^{-1}\times_2 \left(\mathbf{\Psi}_{\mathbf{B}}^{(r)}\right)^{-1} ,
	\end{equation}
	where $\mathbf{\Psi}_{\mathbf{A}}^{(r)}\in\mathbb{R}^{I\times I}$, and
	$\mathbf{\Psi}_{\mathbf{B}}^{(r)}\in\mathbb{R}^{J\times J}$ are
	nonsingular matrices. We note that Lemma~\ref{Le: Essential
		Uniqueness} can be considered to be a variant of Theorem~6.1 in
	\cite{Del2008}. However, the different supposition made here on the
	genericity of $\mathcal{C}$ requires a separate proof, which we
	present now.
\end{remark}
\begin{proof}
	To begin, it is easy to verify that
	\begin{align}
		& \mathbf{T}_{M\times N,2}\mathbf{T}^{\dagger}_{M\times N,1}=\mathbf{A}(\mathbf{I}_R\otimes  \mathbf{C}_{I\times J,2}\mathbf{C}^{\dagger}_{I\times J,1})\mathbf{A}^{\dagger},\\& \mathbf{T}_{M\times N,3}\mathbf{T}^{\dagger}_{M\times N,1}=\mathbf{A}(\mathbf{I}_R\otimes  \mathbf{C}_{I\times J,3}\mathbf{C}^{\dagger}_{I\times J,1})\mathbf{A}^{\dagger}.
	\end{align}
	Subtracting these two equations, we obtain
	\begin{equation}
		\mathbf{E}_{\mathcal{T}}\triangleq(\mathbf{T}_{M\times N,3}-\mathbf{T}_{M\times N,2})\mathbf{T}^{\dagger}_{M\times N,1}=\mathbf{A}(\mathbf{I}_R\otimes\mathbf{E}_\mathcal{C})\mathbf{A}^{\dagger} ,
	\end{equation}
	where $\mathbf{E}_\mathcal{C}\triangleq(\mathbf{C}_{M\times
		N,3}-\mathbf{C}_{M\times N,2})\mathbf{C}^{\dagger}_{M\times
		N,1}$. Benefiting from the genericity of $\mathcal{C}$,
	$\mathbf{E}_\mathcal{C}$ is also generic. Moreover, it is implied that
	the column space of each $\mathbf{A}_r$ is an invariant subspace of
	$\mathbf{E}_\mathcal{T}$, which means that
	\begin{equation}
		\mathsf{col}\left\lbrace
		\mathbf{E}_\mathcal{T}\right\rbrace
		=\mathsf{col}\left\{\mathbf{A}_1\right\}\oplus\mathsf{col}\left\{\mathbf{A}_2\right\}\oplus\cdots\oplus\mathsf{col}\left\{\mathbf{A}_R\right\}.
	\end{equation}
	Now, if $\mathcal{T}$ can be alternatively decomposed into
	\begin{equation}
		\mathcal{T}=\sum_{r=1}^R\hat{\mathcal{C}}\times_1\hat{\mathbf{A}}_r\times_2\hat{\mathbf{B}}_r ,
	\end{equation}
	where $\hat{\mathcal{C}}\in\mathbb{R}^{I\times J\times L}$,
	$\hat{\mathbf{A}}_r\in\mathbb{R}^{M\times I}$, and
	$\hat{\mathbf{B}}_r\in\mathbb{R}^{N\times J}$, we could similarly
	derive that
	\begin{equation}
		\mathsf{col}\left\lbrace
		\mathbf{E}_\mathcal{T}\right\rbrace
		=\mathsf{col}\left\{\hat{\mathbf{A}}_1\right\}\oplus\mathsf{col}\left\{\hat{\mathbf{A}}_2\right\}\oplus\cdots\oplus\mathsf{col}\left\{\hat{\mathbf{A}}_R\right\}.
	\end{equation}
	Denoting $V^{(r)}_{\hat{r}}$ to indicate the number of columns in
	$\mathbf{A}_r$ that belong to
	$\mathsf{col}\left\{\hat{\mathbf{A}}_{\hat{r}}\right\},
	\hat{r}=1,2,\dots,R,$ we have that $V^{(r)}_{\hat{r}}\geq 0$, and
	$\sum_{\hat{r}=1}^RV^{(r)}_{\hat{r}}=I$. If
	$\max_{\hat{r}}\left\{V^{(r)}_{\hat{r}}\right\}<I$, then, because
	$\mathsf{col}\left\{\hat{\mathbf{A}}_{\hat{r}_1}\right\}\bigcap\mathsf{col}\left\{\hat{\mathbf{A}}_{\hat{r}_2}\right\}=\emptyset$,
	$\forall {\hat{r}_1}\neq{\hat{r}_2}$, it can be concluded that there
	exists at least one zero element in each column of
	$\mathbf{E}_\mathcal{C}$ since
	\begin{equation}
		\mathbf{E}_{\mathcal{T}}\mathbf{A}_r=\mathbf{A}_r\mathbf{E}_\mathcal{C}.
	\end{equation}
	Thus $\left\Vert\mathbf{E}_\mathcal{C}\right\Vert_0\leq I^2-I$, and
	there exist finite mappings to rearrange the elements of
	$\mathbf{E}_\mathcal{C}$ to get $\mathbf{E}^{'}_\mathcal{C}$ whose
	last row is all zero. However, due to the genericity of
	$\mathbf{E}_\mathcal{C}$, $\mathbf{E}^{'}_\mathcal{C}$ is also
	generic. Thus, based on Prop.~2.7 in \cite{WM2022},
	\begin{equation}
		\mathsf{P}\left(\operatorname{rank}\left\{\mathbf{E}^{'}_\mathcal{C}\right\}\leq\left\Vert(\mathbf{E}^{'}_\mathcal{C})^T\right\Vert_{2,0}
		< I\right)=0.
	\end{equation}
	Therefore, we are able to conclude that, with probability~1,
	\begin{equation}
		\max_{\hat{r}}\left\{V^{(r)}_{\hat{r}}\right\}=I,
	\end{equation}
	which implies that, $\forall {\hat{r}}$, $\exists r_{\hat{r}}\in
	\left\{1,2,\cdots,R\right\}$ such that
	\begin{equation}
		\mathsf{col}\left\{\hat{\mathbf{A}}_{\hat{r}}\right\}=\mathsf{col}\left\{{\mathbf{A}}_{r_{\hat{r}}}\right\} ,
	\end{equation}
	and, $\forall \hat{r}_1\neq \hat{r}_2$, we have $r_{\hat{r}_1}\neq
	r_{\hat{r}_2}$.
	Without loss of generality, suppose
	$\mathsf{col}\left\{\hat{\mathbf{A}}_{r}\right\}=\mathsf{col}\left\{{\mathbf{A}}_{r}\right\}$,
	$\forall r\in \left\{1,2,\dots,R\right\}$ (or we could permute the
	order of $\left\{\hat{\mathbf{A}}_{\hat{r}}\right\}_{\hat{r}=1}^R$ to
	have this be true).  Then it follows
	that
	\begin{equation}
		\hat{\mathbf{A}}_{r}={\mathbf{A}}_{r}\mathbf{\Psi}_{\mathbf{A}}^{(r)} ,
	\end{equation}
	where $\mathbf{\Psi}_{\mathbf{A}}^{(r)}\in \mathbb{R}^{I\times I}$ is
	some nonsingular matrix.
	
	Applying the same analysis above on the second dimension of
	$\mathcal{T}$, it can be deduced that, with probability~1, $\forall
	{\hat{r}}$, there exist $r_{\hat{r}}\in \left\{1,2,\dots,R\right\}$ and
	nonsingular $\mathbf{\Psi}_{\mathbf{B}}^{({\hat{r}})}\in
	\mathbb{R}^{J\times J}$ such that
	\begin{equation}
		\hat{\mathbf{B}}_{\hat{r}}={\mathbf{B}}_{r_{\hat{r}}}\mathbf{\Psi}_{\mathbf{B}}^{({\hat{r}})}.
	\end{equation}
	Consequently, we have 
	\begin{align}
		\mathcal{T}=\sum_{\hat{r}=1}^R{\mathcal{C}}\times_1{\mathbf{A}}_{\hat{r}}\times_2{\mathbf{B}}_{\hat{r}}&=\sum_{\hat{r}=1}^R\hat{\mathcal{C}}\times_1\hat{\mathbf{A}}_{\hat{r}}\times_2\hat{\mathbf{B}}_{\hat{r}}\nonumber\\&=\sum_{\hat{r}=1}^R\hat{\mathcal{C}}\times_1{\mathbf{A}}_{\hat{r}}\mathbf{\Psi}_{\mathbf{A}}^{(\hat{r})}\times_2{\mathbf{B}}_{r_{\hat{r}}}\mathbf{\Psi}_{\mathbf{B}}^{({\hat{r}})}\nonumber\\&=\sum_{\hat{r}=1}^R\hat{\mathcal{C}}^{({\hat{r}})}\times_1{\mathbf{A}}_{\hat{r}}\times_2{\mathbf{B}}_{r_{\hat{r}}} ,
	\end{align}
	where $\hat{\mathcal{C}}^{({\hat{r}})}\triangleq
	\hat{\mathcal{C}}\times_1\mathbf{\Psi}_{\mathbf{A}}^{(\hat{r})}\times_2\mathbf{\Psi}_{\mathbf{B}}^{({\hat{r}})}$. Performing
	mode-2 unfolding, we further have
	\begin{align}
		\mathbf{T}_{[2]}&=\begin{bmatrix}\mathbf{B}_1&\mathbf{B}_2&\cdots&\mathbf{B}_R\end{bmatrix}\begin{bmatrix}
			\mathbf{C}_{[2]}&\quad&\quad&\quad\\
			\quad&\mathbf{C}_{[2]}&\quad&\quad\\
			\quad&\quad&\ddots&\quad\\
			\quad&\quad&\quad&\mathbf{C}_{[2]}
		\end{bmatrix}\nonumber\\&\times\left(\mathbf{I}_{LR}\otimes\mathbf{A}\right)^T\nonumber\\&=\begin{bmatrix}\mathbf{B}_{r_{1}}&\mathbf{B}_{r_{2}}&\cdots&\mathbf{B}_{r_{R}}\end{bmatrix}\begin{bmatrix}
			\hat{\mathbf{C}}^{({{1}})}_{[2]}&\quad&\quad&\quad\\
			\quad&\hat{\mathbf{C}}^{{({2})}}_{[2]}&\quad&\quad\\
			\quad&\quad&\ddots&\quad\\
			\quad&\quad&\quad&\hat{\mathbf{C}}^{{({R})}}_{[2]}
		\end{bmatrix}\nonumber\\&\times\left(\mathbf{I}_{LR}\otimes\mathbf{A}\right)^T\nonumber\\&=\begin{bmatrix}\mathbf{B}_1&\mathbf{B}_2&\cdots&\mathbf{B}_R\end{bmatrix}\mathbf{\Pi}_{\mathbf{B}}\begin{bmatrix}
			\hat{\mathbf{C}}^{({{1}})}_{[2]}&\quad&\quad&\quad\\
			\quad&\hat{\mathbf{C}}^{{({2})}}_{[2]}&\quad&\quad\\
			\quad&\quad&\ddots&\quad\\
			\quad&\quad&\quad&\hat{\mathbf{C}}^{{({R})}}_{[2]}
		\end{bmatrix}\nonumber\\&\times\left(\mathbf{I}_{LR}\otimes\mathbf{A}\right)^T ,
	\end{align}
	where $\mathbf{\Pi}_{\mathbf{B}}\in\mathbb{R}^{JR\times JR}$ is a
	block permutation matrix. Because
	$\begin{bmatrix}\mathbf{B}_1&\mathbf{B}_2&\cdots&\mathbf{B}_R\end{bmatrix}$ has full column rank, and $\left(\mathbf{I}_{LR}\otimes\mathbf{A}\right)^T$ is
	of full row rank, we have that
	\begin{equation}
		\begin{bmatrix}
			\mathbf{C}_{[2]}&\quad&\quad&\quad\\
			\quad&\mathbf{C}_{[2]}&\quad&\quad\\
			\quad&\quad&\ddots&\quad\\
			\quad&\quad&\quad&\mathbf{C}_{[2]}
		\end{bmatrix}=\mathbf{\Pi}_{\mathbf{B}}\begin{bmatrix}
			\hat{\mathbf{C}}^{({{1}})}_{[2]}&\quad&\quad&\quad\\
			\quad&\hat{\mathbf{C}}^{{({2})}}_{[2]}&\quad&\quad\\
			\quad&\quad&\ddots&\quad\\
			\quad&\quad&\quad&\hat{\mathbf{C}}^{{({R})}}_{[2]}
		\end{bmatrix}.
	\end{equation}
	Thus, $\mathbf{\Pi}_{\mathbf{B}}$ block-wisely permutes a block diagonal matrix into another block diagonal matrix, which happens if and only if
	\begin{equation}
		\mathbf{\Pi}_{\mathbf{B}}=\mathbf{I}_{JR}.
	\end{equation}
	It then follows directly that $r_{\hat{r}}=\hat{r}$, and, more importantly, 
	\begin{align}
		\hat{\mathbf{A}}_1&=\mathbf{A}_1\mathbf{\Psi}_{\mathbf{A}}^{(1)},\nonumber\\
		\hat{\mathbf{A}}_2&=\mathbf{A}_2\mathbf{\Psi}_{\mathbf{A}}^{(2)},\nonumber\\
		&\vdots\nonumber\\
		\hat{\mathbf{A}}_R&=\mathbf{A}_R\mathbf{\Psi}_{\mathbf{A}}^{(R)},\\
		\hat{\mathbf{B}}_1&=\mathbf{B}_1\mathbf{\Psi}_{\mathbf{B}}^{(1)},\nonumber\\
		\hat{\mathbf{B}}_2&=\mathbf{B}_2\mathbf{\Psi}_{\mathbf{B}}^{(2)},\nonumber\\
		&\vdots\nonumber\\
		\hat{\mathbf{B}}_R&=\mathbf{B}_R\mathbf{\Psi}_{\mathbf{B}}^{(R)},
		\\
		\mathcal{C}&=\hat{\mathcal{C}}\times_1\mathbf{\Psi}_{\mathbf{A}}^{(1)}\times_2\mathbf{\Psi}_{\mathbf{B}}^{(1)}=\hat{\mathcal{C}}\times_1\mathbf{\Psi}_{\mathbf{A}}^{(2)}\times_2\mathbf{\Psi}_{\mathbf{B}}^{(2)}\nonumber\\&\times=\dots=\hat{\mathcal{C}}\times_1\mathbf{\Psi}_{\mathbf{A}}^{(R)}\times_2\mathbf{\Psi}_{\mathbf{B}}^{(R)},
	\end{align}
	which completes the proof.
\end{proof}

With Lemma~\ref{Le: Essential Uniqueness} in hand, we are now ready to
prove Thm.~\ref{Th:Recovery Guarantee}.
\setcounter{theorem}{1} 
\renewcommand{\thetheorem}{\ref{Th:Recovery Guarantee}}
\begin{theorem}
	\textit{Suppose the SRI $\mathcal{Z}\in\mathbb{R}^{M_1\times
			M_2\times S}$, HSI $\mathcal{X}\in \mathbb{R} ^{m_1\times
			m_2\times S}$, and MSI $\mathcal{Y}\in \mathbb{R} ^{M_1\times
			M_2\times s}$ satisfy relationship \eqref{eq: Matrix
			Formulation}. Suppose further that the Tucker decomposition of
		$\mathcal{Z}$ and the KD of spatial-degradation matrix
		$\mathbf{D}\in\mathbb{R}^{M_1M_2\times m_1m_2}$ are 
		\begin{align}
			\mathcal{Z}&=\mathcal{G}\times_1 \mathbf{U}_1\times_2 \mathbf{U}_2\times_3 \mathbf{U}_3,\\
			\mathbf{D}&=\sum_{r = 1}^{\mathsf{kr}_{\mathbf{D}}}  (\mathbf{P}_2^{(r)}\otimes \mathbf{P}_1^{(r)})^T ,
		\end{align}
		where $\mathcal{G}\in \mathbb{R}^{L_1\times L_2\times C}$ is drawn from an absolutely continuous distribution; and
		$\mathbf{U}_1$, $\mathbf{U}_2$, and $\mathbf{U}_3$ have
		full column rank. Then, if it is true that
		\begin{equation}
			\begin{aligned}
				&L_1\leq L_2C, \,L_2\leq L_1C,\, S\geq 3\\ & \operatorname{rank}\left\{\left[\mathbf{P}_1^{(1)}\mathbf{U}_1,\cdots,\mathbf{P}_1^{(\mathsf{kr}_{\mathbf{D}})}\mathbf{U}_1\right]\right\}=L_1\mathsf{kr}_{\mathbf{D}},\\
				& \operatorname{rank}\left\{\left[\mathbf{P}_2^{(1)}\mathbf{U}_2,\cdots,\mathbf{P}_2^{(\mathsf{kr}_{\mathbf{D}})}\mathbf{U}_2\right]\right\}=L_2\mathsf{kr}_{\mathbf{D}},\\ &
				\operatorname{rank}\left\{\mathbf{Y}_{[1]}\right\}=L_1, \\ &
				\operatorname{rank}\left\{\mathbf{Y}_{[2]}\right\}=L_2,
			\end{aligned}
		\end{equation}
		any solution of Tucker-rank at most $(L_1,L_2,S)$ to GTF-HSR
		recovers SRI $\mathcal{Z}$ with probability 1.}
\end{theorem}
\begin{proof}
	Combining \eqref{eq:Reformulation}
	with the Tucker decomposition of $\mathcal{Z}$, we have
	\begin{align}
		\mathcal{X}&=\sum_{r=1}^{\mathsf{kr}_{\mathbf{D}}}\mathcal{G}\times_1 \mathbf{P}_1^{(r)}\mathbf{U}_1\times_2 \mathbf{P}_2^{(r)}\mathbf{U}_2\times_3 \mathbf{U}_3, \label{eq:hsidecomp}\\
		\mathcal{Y}&=\mathcal{G}\times_1 \mathbf{U}_1\times_2 \mathbf{U}_2\times_3 \mathbf{R}\mathbf{U}_3. \label{eq:msidecomp}
	\end{align}
	Denoting that
	\begin{align}
		\mathbf{P}_1&\triangleq \left[\mathbf{P}_1^{(1)}\mathbf{U}_1,\mathbf{P}_1^{(2)}\mathbf{U}_1,\cdots,\mathbf{P}_1^{(R)}\mathbf{U}_1\right],\\
		\mathbf{P}_2&\triangleq \left[\mathbf{P}_2^{(1)}\mathbf{U}_2,\mathbf{P}_2^{(2)}\mathbf{U}_2,\cdots,\mathbf{P}_2^{(R)}\mathbf{U}_2\right],
	\end{align}
	then $\mathbf{P}_1$ and
	$\mathbf{P}_2$ are of full column rank according to our conditions. Considering that
	$\mathbf{U}_3$ has full column rank along with the conditions $L_1\leq
	L_2C$, $L_2\leq L_1C$, and $S\geq 3$, we are now capable of invoking
	Lemma~\ref{Le: Essential Uniqueness} above and Thm~5.1 in \cite{Del2008}
	to conclude the essential uniqueness of the decomposition on
	$\mathcal{X}$ in the form of \eqref{eq:hsidecomp} with respect to
	$\mathcal{G}$, $\left\{\mathbf{P}_1^{(r)}\mathbf{U}_1\right\}_{r=1}^{\mathsf{kr}_{\mathbf{D}}}$, $\left\{\mathbf{P}_2^{(r)}\mathbf{U}_2\right\}_{r=1}^{\mathsf{kr}_{\mathbf{D}}}$, and $\mathbf{U}_3$.
	
	Now, let $\hat{\mathcal{Z}}\in\mathbb{R}^{M_1\times M_2\times S}$ be
	an arbitrary solution of Tucker-rank at most $(L_1,L_2,C)$ to the
	GTF-HSR problem of Def.~\ref{Def: Tensor ReFormulation}. Its Tucker
	decomposition then expands as
	\begin{equation}
		\hat{\mathcal{Z}}=\hat{\mathcal{G}}\times_1 \hat{\mathbf{U}}_1\times_2 \hat{\mathbf{U}}_2\times_3 \hat{\mathbf{U}}_3 ,
		\label{eq:expand}
	\end{equation}
	where $\mathcal{G}\in \mathbb{R}^{\hat{L}_1\times \hat{L}_2\times
		\hat{C}}$ is of Tucker-rank $\left(\hat{L}_1,\hat{L}_2,
	\hat{C}\right)$, and $\hat{\mathbf{U}}_1\in\mathbb{R}^{M_1\times
		\hat{L_1}}$, $\hat{\mathbf{U}}_2\in\mathbb{R}^{M_2\times
		\hat{L}_2}$, and $\hat{\mathbf{U}}_3\in\mathbb{R}^{S\times
		\hat{C}}$ all have full column rank. Substituting
	\eqref{eq:expand} into \eqref{eq:Reformulation}, we have
	\begin{align}
		\mathcal{X}&=\sum_{r=1}^{\mathsf{kr}_{\mathbf{D}}}\hat{\mathcal{G}}\times_1 \mathbf{P}_1^{(r)}\hat{\mathbf{U}}_1\times_2 \mathbf{P}_2^{(r)}\hat{\mathbf{U}}_2\times_3 \hat{\mathbf{U}}_3,
		\label{eq: HSI Decom Solution}\\
		\mathcal{Y}&=\hat{\mathcal{G}}\times_1 \hat{\mathbf{U}}_1\times_2 \hat{\mathbf{U}}_2\times_3 \mathbf{R}\hat{\mathbf{U}}_3.
		\label{eq: MSI Decom Solution}
	\end{align}
	Then we notice that, from \eqref{eq: HSI Decom Solution}, we can
	have only that $\hat{L}_1=L_1$, $\hat{L}_2=L_2$, and $\hat{C}=C$ which
	would otherwise contradict the uniqueness of \eqref{eq:hsidecomp} in
	terms of $\mathcal{G}$,
	$\left\{\mathbf{P}_1^{(r)}\mathbf{U}_1\right\}_{r=1}^{\mathsf{kr}_{\mathbf{D}}}$,
	$\left\{\mathbf{P}_2^{(r)}\mathbf{U}_2\right\}_{r=1}^{\mathsf{kr}_{\mathbf{D}}}$,
	and $\mathbf{U}_3$. It then follows from such uniqueness
	that,
	$\forall\,r\in\left\{1,2,\dots,\mathsf{kr}_{\mathbf{D}}\right\}$, there exist
	$\mathbf{\Psi}_{1}^{(r)}\in\mathbb{R}^{L_1\times
		L_1}$ and $\mathbf{\Psi}_{2}^{(r)}\in\mathbb{R}^{L_2\times L_2}$ that
	are nonsingular matrices such that
	\begin{align}
		\mathbf{P}_1^{(r)}\hat{\mathbf{U}}_1&=\mathbf{P}_1^{(r)}\mathbf{U}_1\mathbf{\Psi}_{1}^{(r)},\\
		\mathbf{P}_2^{(r)}\hat{\mathbf{U}}_2&=\mathbf{P}_2^{(r)}\mathbf{U}_2\mathbf{\Psi}_{2}^{(r)}, \\
		\hat{\mathbf{U}}_3&=\mathbf{U}_3\mathbf{\Psi}_{3}\\
		\hat{\mathcal{G}}&=\mathcal{G}\times_1 \left(\mathbf{\Psi}_{1}^{(r)}\right)^{-1}\times_2 \left(\mathbf{\Psi}_{2}^{(r)}\right)^{-1}\times_3 \left(\mathbf{\Psi}_{3}\right)^{-1} ,
	\end{align}
	where $\mathbf{\Psi}_3\in\mathbb{R}^{C\times C}$ is also
	nonsingular. Subsequently, unfolding \eqref{eq:msidecomp} and
	\eqref{eq: MSI Decom Solution}, we have
	\begin{align}
		\mathbf{Y}_{[1]}&=\mathbf{U}_1\mathbf{G}_{[1]}\left(\mathbf{P}_{3}^{(r)}\mathbf{U}_3\otimes \mathbf{U}_2\right)^T, \\
		\mathbf{Y}_{[2]}&=\mathbf{U}_2\mathbf{G}_{[2]}\left(\mathbf{P}_{3}^{(r)}\mathbf{U}_3\otimes \mathbf{U}_1\right)^T,\\
		\mathbf{Y}_{[1]}&=\hat{\mathbf{U}}_1\hat{\mathbf{G}}_{[1]}\left(\mathbf{P}_{3}^{(r)}\hat{\mathbf{U}}_3\otimes \hat{\mathbf{U}}_2\right)^T, \\
		\mathbf{Y}_{[2]}&=\hat{\mathbf{U}}_2\hat{\mathbf{G}}_{[2]}\left(\mathbf{P}_{3}^{(r)}\hat{\mathbf{U}}_3\otimes \hat{\mathbf{U}}_1\right)^T.
	\end{align}
	Since $\operatorname{rank}\left\{\mathbf{Y}_{[1]}\right\}=L_1$ and
	$\operatorname{rank}\left\{\mathbf{Y}_{[2]}\right\}=L_2$, it is
	concluded
	that
	\begin{align}
		\mathsf{col}\left(\mathbf{Y}_{[1]}\right)&=\mathsf{col}\left(\mathbf{U}_{1}\right)=\mathsf{col}\left(\hat{\mathbf{U}}_{1}\right), \\
		\mathsf{col}\left(\mathbf{Y}_{[2]}\right)&=\mathsf{col}\left(\mathbf{U}_{2}\right)=\mathsf{col}\left(\hat{\mathbf{U}}_{2}\right) ,
	\end{align}
	which indicates the existence of nonsingular
	$\mathbf{Q}_1\in\mathbb{R}^{L_1\times L_1}$ and
	$\mathbf{Q}_2\in\mathbb{R}^{L_2\times L_2}$ such that
	\begin{align}
		\hat{\mathbf{U}}_1&=\mathbf{U}_1\mathbf{Q}_1, \\
		\hat{\mathbf{U}}_2&=\mathbf{U}_2\mathbf{Q}_2.
	\end{align}
	We then have that
	\begin{align}
		\mathbf{P}_1^{(r)}\hat{\mathbf{U}}_1&=\mathbf{P}_1^{(r)}\mathbf{U}_1\mathbf{Q}_1=\mathbf{P}_1^{(r)}\mathbf{U}_1\mathbf{\Psi}_{1}^{(r)},\\\mathbf{P}_2^{(r)}\hat{\mathbf{U}}_2&=\mathbf{P}_2^{(r)}\mathbf{U}_2\mathbf{Q}_2=\mathbf{P}_2^{(r)}\mathbf{U}_2\mathbf{\Psi}_{2}^{(r)}.
	\end{align}
	Since $\mathbf{P}_1^{(r)}\mathbf{U}_1$ has full column rank, it is
	then true that, $\forall
	r\in\left\{1,2,\dots,\mathsf{kr}_{\mathbf{D}}\right\}$,
	\begin{align}
		\mathbf{Q}_1&=\mathbf{\Psi}_{1}^{(r)}, \\
		\mathbf{Q}_2&=\mathbf{\Psi}_{2}^{(r)} .
	\end{align}
	It finally follows that
	\begin{align}
		\hat{\mathcal{Z}}&=\hat{\mathcal{G}}\times_1 \hat{\mathbf{U}}_1\times_2 \hat{\mathbf{U}}_2\times_3 \hat{\mathbf{U}}_3\nonumber\\
		&=\hat{\mathcal{G}}\times_1 {\mathbf{U}}_1\mathbf{Q}_1\times_2 {\mathbf{U}}_2\mathbf{Q}_2\times_3 {\mathbf{U}}_3\mathbf{\Psi}_3\nonumber\\
		&=\mathcal{G}\times_1 \left(\mathbf{\Psi}_{1}^{(r)}\right)^{-1}\times_2 \left(\mathbf{\Psi}_{2}^{(r)}\right)^{-1}\times_3 \left(\mathbf{\Psi}_{3}\right)^{-1}\nonumber\\&\quad\quad\times_1 {\mathbf{U}}_1\mathbf{Q}_1\times_2 {\mathbf{U}}_2\mathbf{Q}_2\times_3 {\mathbf{U}}_3\mathbf{\Psi}_3\nonumber\\
		&=\mathcal{G}\times_1 \left(\mathbf{Q}_{1}\right)^{-1}\times_2 \left(\mathbf{Q}_{2}\right)^{-1}\times_3 \left(\mathbf{\Psi}_{3}\right)^{-1}\nonumber\\&\quad\quad\times_1 {\mathbf{U}}_1\mathbf{Q}_1\times_2 {\mathbf{U}}_2\mathbf{Q}_2\times_3 {\mathbf{U}}_3\mathbf{\Psi}_3\nonumber\\
		&=\mathcal{G}\times_1 {\mathbf{U}}_1\mathbf{Q}_1\left(\mathbf{Q}_{1}\right)^{-1}\times_2 {\mathbf{U}}_2\mathbf{Q}_2\left(\mathbf{Q}_{2}\right)^{-1}\times_3{\mathbf{U}}_3\mathbf{\Psi}_3 \left(\mathbf{\Psi}_{3}\right)^{-1}\nonumber\\
		&=\mathcal{G}\times_1 \mathbf{U}_1\times_2 \mathbf{U}_2\times_3 \mathbf{U}_3\nonumber\\
		&=\mathcal{Z} ,
	\end{align}
	which completes the proof.
\end{proof}
\subsection{Proof of Cor.~\ref{Cro:TeF Inability}}
\setcounter{corollary}{0} 
\renewcommand{\thecorollary}{3.\arabic{corollary}}
\begin{corollary}
	\textit{Under the conditions of Thm.~\ref{Th:Recovery Guarantee},
		if it is true that
		\begin{equation}
			\begin{aligned}
				&L_1\leq L_2C, \,L_2\leq L_1C,\, S\geq 3,\\
				& \operatorname{rank}\left\{\left[\mathbf{P}_1^{(1)}\mathbf{U}_1,\cdots,\mathbf{P}_1^{(\mathsf{kr}_{\mathbf{D}})}\mathbf{U}_1\right]\right\}=L_1\mathsf{kr}_{\mathbf{D}},\\
				& \operatorname{rank}\left\{\left[\mathbf{P}_2^{(1)}\mathbf{U}_2,\cdots,\mathbf{P}_2^{(\mathsf{kr}_{\mathbf{D}})}\mathbf{U}_2\right]\right\}=L_2\mathsf{kr}_{\mathbf{D}},
			\end{aligned}
		\end{equation}
		then any solution to TF-HSR recovers SRI $\mathcal{Z}$ with
		probability 0 when $\mathsf{kr}_{\mathbf{D}}>1$.}
\end{corollary}
\begin{proof}
	We first recall that, under these conditions, we have the essential
	uniqueness of the decomposition of the HSI $\mathcal{X}$ in the form
	of
	\begin{equation}
		\mathcal{X}=\sum_{r=1}^{\mathsf{kr}_{\mathbf{D}}}\mathcal{G}\times_1
		\mathbf{P}_1^{(r)}\mathbf{U}_1\times_2
		\mathbf{P}_2^{(r)}\mathbf{U}_2\times_3 \mathbf{U}_3
	\end{equation}
	almost surely. Thus, if SRI $\mathcal{Z}$ solves TF-HSR, there would
	exist an alternative decomposition of $\mathcal{X}$ in the form
	of
	\begin{equation}
		\mathcal{X}=\mathcal{G}\times_1 \mathbf{P}_1\mathbf{U}_1\times_2
		\mathbf{P}_2\mathbf{U}_2\times_3\mathbf{U}_3,
	\end{equation}
	which would contradict the uniqueness above. As such, we assert that
	SRI $\cal Z$ solves the TF-HSR problem with probability~0, thereby
	completing the proof.
	
\end{proof}
\section{Detailed ADMM Derivations for Algs.~\ref{alg:BGS-TeRF Part 1} and \ref{alg:BGS-TeRF Part 2}}
\subsection{ADMM for Alg.~\ref{alg:BGS-TeRF Part 1}}
ADMM is used in Alg.~\ref{alg:BGS-TeRF Part 1} to solve the
optimization in \eqref{eq:SubIden}, namely,
\begin{equation}
	\begin{aligned}
		\min_{\mathbf{U}_i^\mathcal{X},\mathbf{A}_i}\,\Big\lVert \mathbf{X}_{[i]}-
		&\begin{bmatrix}\mathbf{P}_i^{(1)} & \cdots & \mathbf{P}_i^{(\mathsf{kr}_{\mathbf{D}})}\end{bmatrix}\left(\mathbf{I}_{\mathsf{kr}_{\mathbf{D}}}\otimes \begin{bmatrix}\mathbf{U}_i^\mathcal{Y} & \mathbf{U}_i^\mathcal{X}\end{bmatrix}\right)\mathbf{B}_{i}\Big\rVert_F^2\\&+\mu\left\lVert \mathbf{A}_i\right\rVert_1 \\
		&\text{s.t.} \quad \mathbf{A}_i = \mathbf{B}_i ,
	\end{aligned}
\end{equation}
The augmented Lagrangian function is 
\begin{equation}
	\begin{aligned}
		\mathsf{L}&\left(\mathbf{U}_i^\mathcal{X},\mathbf{B}_{i},\mathbf{A}_i,\mathbf{M}_i\right)\\&\triangleq\left\lVert \mathbf{X}_{[i]}-\left[\mathbf{P}_i^{(1)},\cdots,\mathbf{P}_i^{(\mathsf{kr}_{\mathbf{D}})}\right]\left(\mathbf{I}_{\mathsf{kr}_{\mathbf{D}}}\otimes \left[\mathbf{U}_i^\mathcal{Y},\mathbf{U}_i^\mathcal{X}\right]\right)\mathbf{B}_{i}\right\rVert_F^2\\&+\mu\left\lVert \mathbf{A}_i\right\rVert_1 +\left\langle\mathbf{M}_i,\mathbf{A}_i-\mathbf{B}_i\right\rangle+\frac{\rho}{2}\left\Vert\mathbf{A}_i-\mathbf{B}_i\right\Vert_F^2\\&=\left\lVert \mathbf{X}_{[i]}-\left[\mathbf{P}_i^{(1)},\cdots,\mathbf{P}_i^{(\mathsf{kr}_{\mathbf{D}})}\right]\left(\mathbf{I}_{\mathsf{kr}_{\mathbf{D}}}\otimes \left[\mathbf{U}_i^\mathcal{Y},\mathbf{U}_i^\mathcal{X}\right]\right)\mathbf{B}_{i}\right\rVert_F^2\\&+\mu\left\lVert \mathbf{A}_i\right\rVert_1 +\frac{\rho}{2}\left\Vert\mathbf{A}_i-\mathbf{B}_i+\frac{\mathbf{M}_i}{\rho}\right\Vert_F^2-\frac{\left\Vert\mathbf{M}_i\right\Vert_F^2}{2\rho} ,
	\end{aligned}
\end{equation}
where the auxiliary variables $\mathbf{M}_i$ are of the same size as
$\mathbf{A}_i$.
\begin{itemize} 
	\item
	The $\mathbf{U}_i^\mathcal{X}$ subproblem:\newline
	Solving for $\mathbf{U}_i^\mathcal{X}$ proceeds by solving
	\begin{equation}
		\begin{aligned}
		\min_{\mathbf{U}_i^\mathcal{X}}\mathsf{L}_{\mathbf{U}_i^\mathcal{X}}\triangleq\Big\lVert \mathbf{X}_{[i]}-\begin{bmatrix}\mathbf{P}_i^{(1)}&\cdots&\mathbf{P}_i^{(\mathsf{kr}_{\mathbf{D}})}\end{bmatrix}\\\times\left(\mathbf{I}_{\mathsf{kr}_{\mathbf{D}}}\otimes \begin{bmatrix}\mathbf{U}_i^\mathcal{Y}&\mathbf{U}_i^\mathcal{X}\end{bmatrix}\right)\mathbf{B}_{i}\Big\rVert_F^2.
		\end{aligned}
		\label{eq:subprob1}
	\end{equation}
	Partitioning $\mathbf{B}_i$ into
	\begin{equation}
		\begin{aligned}
			\mathbf{B}_i=\Big[&
				(\mathbf{B}_i^{(1)\mathcal{Y}})^T 
				,(\mathbf{B}_i^{(1)\mathcal{X}})^T 
				,(\mathbf{B}_i^{(2)\mathcal{Y}})^T 
				,(\mathbf{B}_i^{(2)\mathcal{X}})^T 
				,\\&\cdots 
				(\mathbf{B}_i^{(\mathsf{kr}_{\mathbf{D}})\mathcal{Y}})^T 
				,(\mathbf{B}_i^{(\mathsf{kr}_{\mathbf{D}})\mathcal{X}})^T
			\Big]^T
		\end{aligned}
	\end{equation}
	where $\mathbf{B}_i^{(r)\mathcal{Y}}\in\mathbb{R}^{K_i\times
		m_{3-i}S}$ and
	$\mathbf{B}_i^{(r)\mathcal{X}}\in\mathbb{R}^{(L_i-K_i)\times
		m_{3-i}S}$, $\mathsf{L}_{\mathbf{U}_i^\mathcal{X}}$ is then
	reorganized as
	\begin{align}
		\mathsf{L}_{\mathbf{U}_i^\mathcal{X}}&=\Big\lVert \mathbf{X}_{[i]}-\begin{bmatrix}\mathbf{P}_i^{(1)}&\cdots&\mathbf{P}_i^{(\mathsf{kr}_{\mathbf{D}})}\end{bmatrix}\nonumber\\&\quad\times\left(\mathbf{I}_{\mathsf{kr}_{\mathbf{D}}}\otimes \begin{bmatrix}\mathbf{U}_i^\mathcal{Y}&\mathbf{U}_i^\mathcal{X}\end{bmatrix}\right)\mathbf{B}_{i}\Big\rVert_F^2\nonumber\\&=\left\lVert \mathbf{X}_{[i]}-\sum_{r=1}^{\mathsf{kr}_{\mathbf{D}}}\mathbf{P}_i^{(r)}\begin{bmatrix}\mathbf{U}_i^\mathcal{Y}&\mathbf{U}_i^\mathcal{X}\end{bmatrix}\begin{bmatrix}
			\mathbf{B}_i^{(r)\mathcal{Y}}\nonumber\\
			\mathbf{B}_i^{(r)\mathcal{X}}
		\end{bmatrix}\right\rVert_F^2\nonumber\\&=\Bigg\lVert \mathbf{X}_{[i]}-\sum_{r=1}^{\mathsf{kr}_{\mathbf{D}}}\mathbf{P}_i^{(r)}\mathbf{U}_i^\mathcal{Y}
		\mathbf{B}_i^{(r)\mathcal{Y}}\nonumber\\&\quad-\sum_{r=1}^{\mathsf{kr}_{\mathbf{D}}}\mathbf{P}_i^{(r)}\mathbf{U}_i^\mathcal{X}
		\mathbf{B}_i^{(r)\mathcal{X}}\Bigg\rVert_F^2.
	\end{align}
	Thus, the gradient
	$\nabla\mathsf{L}_{\mathbf{U}_i^\mathcal{X}}$ is becomes
	\begin{equation}
		\begin{aligned}
			\nabla\mathsf{L}_{\mathbf{U}_i^\mathcal{X}}&=2\Bigg(\sum_{r_1=1}^{\mathsf{kr}_{\mathbf{D}}}\sum_{r_2=1}^{\mathsf{kr}_{\mathbf{D}}}\left(\mathbf{P}_i^{r_1}\right)^T\mathbf{P}_i^{r_2}\mathbf{U}_i^\mathcal{X}\mathbf{B}_i^{(r_2)\cal X}\\&\times\left(\mathbf{B}_i^{(r_1)\cal X}\right)^T-\sum_{r=1}^{\mathsf{kr}_{\mathbf{D}}}\left(\mathbf{P}_i^{r}\right)^T\widetilde{\mathbf{X}}_{[i]}\left(\mathbf{B}_i^{(r)\cal X}\right)^T\Bigg) ,
		\end{aligned}
	\end{equation}
	where
	$\widetilde{\mathbf{X}}_{[i]}\triangleq\mathbf{X}_{[i]}-\sum_{r=1}^{\mathsf{kr}_{\mathbf{D}}}\mathbf{P}_i^{(r)}\mathbf{U}_i^\mathcal{Y}
	\mathbf{B}_i^{(r)\mathcal{Y}}.$ Then \eqref{eq:subprob1}
	is solved by setting
	$\nabla\mathsf{L}_{\mathbf{U}_i^\mathcal{X}}=0$ and applying CG \cite{GV2013}.
	
	\item
	The $\mathbf{B}_i$ subproblem:\newline
	Solving for $\mathbf{B}_i$ proceeds by solving
	\begin{equation}
		\begin{aligned}
			\min_{\mathbf{B}_{i}}\Bigg\lVert& \mathbf{X}_{[i]}-\begin{bmatrix}\mathbf{P}_i^{(1)}&\cdots&\mathbf{P}_i^{(\mathsf{kr}_{\mathbf{D}})}\end{bmatrix}\left(\mathbf{I}_{\mathsf{kr}_{\mathbf{D}}}\otimes \begin{bmatrix}\mathbf{U}_i^\mathcal{Y}&\mathbf{U}_i^\mathcal{X}\end{bmatrix}\right)\\&\times\mathbf{B}_{i}\Bigg\rVert_F^2+\frac{\rho}{2}\left\Vert\mathbf{A}_i-\mathbf{B}_i+\frac{\mathbf{M}_i}{\rho}\right\Vert_F^2.
		\end{aligned}
	\end{equation}
	The corresponding objective function is strongly convex
	and has the unique solution
	\begin{equation}
		\begin{aligned}				\mathbf{B}_i=&\left(\mathbf{D}_i^T\mathbf{D}_i+\frac{\rho}{2}\mathbf{I}_{L_i\mathsf{kr}_{\mathbf{D}}}\right)^{-1}\\&\times\left(\mathbf{D}_i^T\mathbf{X}_{[i]}+\frac{\rho}{2}\left(\mathbf{A}_i+\frac{\mathbf{M}_i}{\rho}\right)\right) ,
		\end{aligned}
	\end{equation}
	where
	$$\mathbf{D}_i\triangleq\begin{bmatrix}\mathbf{P}_i^{(1)}&\cdots&\mathbf{P}_i^{(\mathsf{kr}_{\mathbf{D}})}\end{bmatrix}\left(\mathbf{I}_{\mathsf{kr}_{\mathbf{D}}}\otimes
	\begin{bmatrix}\mathbf{U}_i^\mathcal{Y}&\mathbf{U}_i^\mathcal{X}\end{bmatrix}\right).$$
	\item
	The $\mathbf{A}_i$subproblem:\newline
	The solution to 
	\begin{equation}
		\min_{\mathbf{A}_i}\mu\left\lVert \mathbf{A}_i\right\rVert_1 +\frac{\rho}{2}\left\Vert\mathbf{A}_i-\mathbf{B}_i+\frac{\mathbf{M}_i}{\rho}\right\Vert_F^2
	\end{equation}
	is well-known to be
	\begin{equation}
		\mathbf{A}_i=\mathtt{soft}\left(\mathbf{B}_i-\frac{\mathbf{M}_i}{\rho},\frac{\mu}{\rho}\right).
	\end{equation}
	
	\item
	Updating $\mathbf{M}_i$:\newline
	The final step is the updating of the auxiliary variables $\mathbf{M}_i$
	which is done as
	\begin{equation}
		\mathbf{M}_i\leftarrow\mathbf{M}_i+\rho(\mathbf{A}_i-\mathbf{B}_i).
	\end{equation}
	
\end{itemize}

\subsection{ADMM for Alg.~\ref{alg:BGS-TeRF Part 2}}
ADMM is used in Alg.~\ref{alg:BGS-TeRF Part 2} to solve the
optimization in \eqref{eq:BGS subproblem}, namely
\begin{equation}
	\begin{gathered}  
		\min_{{\left\{\mathcal{G}_r\right\}_{r=1}^{\mathsf{kr}_{\mathbf{D}}}
			},\mathcal{G},\hat{\mathbf{G}}} \left\lVert
		\hat{\mathbf{G}}\right\rVert_{2,\gamma} \\
		\begin{aligned}
			\text{s.t.}\quad \mathcal{X} & =\sum_{r =
				1}^{\mathsf{kr}_{\mathbf{D}}}
			\mathcal{G}_r\times_1\mathbf{P}_1^{(r)}\mathbf{U}_1\times_2\mathbf{P}_2^{(r)}\mathbf{U}_2\times_3\mathbf{U}_3,
			\\ \mathcal{Y} &
			=\mathcal{G}\times_1\mathbf{U}_1\times_2\mathbf{U}_2\times_3\mathbf{R}\mathbf{U}_3,
			\\ \hat{\mathbf{G}} & =\mathbf{G}_{[\mathbf{t}]}, \\ \mathcal{G}
			& =\mathcal{G}_r, \,\, r=1,2,\dots,\mathsf{kr}_{\mathbf{D}}.
		\end{aligned}
	\end{gathered}
\end{equation}
The augmented Lagrangian function is
\begin{align}
	\mathsf{L}&\Big({\left\{\mathcal{G}_r\right\}_{r=1}^{\mathsf{kr}_{\mathbf{D}}} },\mathcal{G},\hat{\mathbf{G}},\mathcal{P}^\mathcal{X},\mathcal{P}^\mathcal{Y},\mathbf{W},\left\{\mathcal{P}_r\right\}_{r=1}^{\mathsf{kr}_{\mathbf{D}}}\Big)\nonumber\\&\triangleq {}\left\lVert \hat{\mathbf{G}}\right\rVert_{2,\gamma}+\left\langle\mathcal{P}^\mathcal{X},\mathcal{X}-\sum_{r = 1}^{\mathsf{kr}_{\mathbf{D}}}  \mathcal{G}_r\times_1\mathbf{P}_1^{(r)}\mathbf{U}_1\times_2\mathbf{P}_2^{(r)}\mathbf{U}_2\times_3\mathbf{U}_3\right\rangle\nonumber\\&\quad+\left\langle\mathcal{P}^\mathcal{Y},\mathcal{Y}-\mathcal{G}\times_1\mathbf{U}_1\times_2\mathbf{U}_2\times_3\mathbf{R}\mathbf{U}_3\right\rangle+\left\langle\mathbf{W},\hat{\mathbf{G}}-\mathbf{G}_{[B;\mathbf{t}]}\right\rangle\nonumber\\&\quad+\sum_{r=1}^{\mathsf{kr}_{\mathbf{D}}}\left\langle\mathcal{P}_r,\mathcal{G}-\mathcal{G}_r\right\rangle\nonumber\\&\quad+\frac{\rho}{2}\Bigg(\left\Vert\mathcal{X}-\sum_{r = 1}^{\mathsf{kr}_{\mathbf{D}}}  \mathcal{G}_r\times_1\mathbf{P}_1^{(r)}\mathbf{U}_1\times_2\mathbf{P}_2^{(r)}\mathbf{U}_2\times_3\mathbf{U}_3\right\Vert_F^2\nonumber\\&\quad\quad\quad+\left\Vert\mathcal{Y}-\mathcal{G}\times_1\mathbf{U}_1\times_2\mathbf{U}_2\times_3\mathbf{R}\mathbf{U}_3\right\Vert_F^2\nonumber\\&\quad\quad\quad+\left\Vert\hat{\mathbf{G}}-\mathbf{G}_{[B;\mathbf{t}]}\right\Vert_F^2+\sum_{r = 1}^{\mathsf{kr}_{\mathbf{D}}}\left\Vert\mathcal{G}-\mathcal{G}_r\right\Vert_F^2\Bigg)={}\left\lVert \hat{\mathbf{G}}\right\rVert_{2,\gamma}\nonumber\\&\quad+\frac{\rho}{2}\Bigg(\left\Vert\mathcal{X}-\sum_{r = 1}^{\mathsf{kr}_{\mathbf{D}}}  \mathcal{G}_r\times_1\mathbf{P}_1^{(r)}\mathbf{U}_1\times_2\mathbf{P}_2^{(r)}\mathbf{U}_2\times_3\mathbf{U}_3+\frac{\mathcal{P}^{\cal{X}}}{\rho}\right\Vert_F^2\nonumber\\&\quad\quad\quad+\left\Vert\mathcal{Y}-\mathcal{G}\times_1\mathbf{U}_1\times_2\mathbf{U}_2\times_3\mathbf{R}\mathbf{U}_3+\frac{\mathcal{P}^{\cal{Y}}}{\rho}\right\Vert_F^2\nonumber\\&\quad\quad\quad+\left\Vert\hat{\mathbf{G}}-\mathbf{G}_{[B;\mathbf{t}]}+\frac{\mathbf{W}}{\rho}\right\Vert_F^2+\sum_{r = 1}^{\mathsf{kr}_{\mathbf{D}}}\left\Vert\mathcal{G}-\mathcal{G}_r+\frac{\mathcal{P}_r}{\rho}\right\Vert_F^2\Bigg)\nonumber\\&\quad\quad\quad-\frac{1}{2\rho}\Bigg(\left\Vert\mathcal{P}^\mathcal{X}\right\Vert_F^2+\left\Vert\mathcal{P}^\mathcal{Y}\right\Vert_F^2+\left\Vert\mathbf{W}\right\Vert_F^2+\sum_{r = 1}^{\mathsf{kr}_{\mathbf{D}}}\left\Vert\mathcal{P}_r\right\Vert_F^2\Bigg)
\end{align}
where $\mathcal{P}^\mathcal{X}$, $\mathcal{P}^\mathcal{Y}$,
$\mathbf{W}$, and
$\left\{\mathcal{P}_r\right\}_{r=1}^{\mathsf{kr}_{\mathbf{D}}}$ are
auxiliary variables.
\begin{itemize}
	\item
	The $\mathcal{G}_r$ subproblem:\newline
	Solving for $\mathcal{G}_r$ proceeds by solving
	\begin{align}
		\min_{\mathcal{G}_r}\,\mathsf{L}_{\mathcal{G}_r}&\triangleq\Bigg\Vert\mathcal{X}-\sum_{r^* = 1}^{\mathsf{kr}_{\mathbf{D}}}  \mathcal{G}_{r^*}\times_1\mathbf{P}_1^{(r^*)}\mathbf{U}_1\times_2\mathbf{P}_2^{(r^*)}\mathbf{U}_2\nonumber\\&\quad\times_3\mathbf{U}_3+\frac{\mathcal{P}^{\cal{X}}}{\rho}\Bigg\Vert_F^2+\left\Vert\mathcal{G}-\mathcal{G}_r+\frac{\mathcal{P}_r}{\rho}\right\Vert_F^2\nonumber\\&=\Bigg\Vert\mathcal{X}+\frac{\mathcal{P}^{\cal{X}}}{\rho}-\sum_{r^*\neq r}  \mathcal{G}_{r^*}\times_1\mathbf{P}_1^{(r^*)}\mathbf{U}_1\nonumber\\&\quad\times_2\mathbf{P}_2^{(r^*)}\mathbf{U}_2\times_3\mathbf{U}_3-\mathcal{G}_{r}\times_1\mathbf{P}_1^{(r)}\mathbf{U}_1\nonumber\\&\quad\times_2\mathbf{P}_2^{(r)}\mathbf{U}_2\times_3\mathbf{U}_3\Bigg\Vert_F^2\nonumber\\&+\left\Vert\mathcal{G}_r-\left(\mathcal{G}+\frac{\mathcal{P}_r}{\rho}\right)\right\Vert_F^2 .
		\label{eq:grsub}
	\end{align}
	By defining
	\begin{align}
		\mathcal{H}&=\mathcal{X}+\frac{\mathcal{P}^{\cal{X}}}{\rho}\nonumber\\&\quad-\sum_{r^*\neq r}  \mathcal{G}_{r^*}\times_1\mathbf{P}_1^{(r^*)}\mathbf{U}_1\times_2\mathbf{P}_2^{(r^*)}\mathbf{U}_2\times_3\mathbf{U}_3,\\
		\mathbf{Q}_1&=\mathbf{P}_1^{(r)}\mathbf{U}_1,\\
		\mathbf{Q}_2&=\mathbf{P}_2^{(r)}\mathbf{U}_2,\\
		\mathbf{Q}_3&=\mathbf{U}_3,\\\mathcal{S}&=\mathcal{G}_r,\\
		\mathcal{K}&=\mathcal{G}+\frac{\mathcal{P}_r}{\rho},
	\end{align}
	\eqref{eq:grsub} then falls into the form
	\begin{equation}
		\min_{\mathcal{S}}\,\left\Vert\mathcal{H}-\mathcal{S}\times_1\mathbf{Q}_1\times_2\mathbf{Q}_2\times_3\mathbf{Q}_3\right\Vert_F^2+\tau\left\Vert\mathcal{S}-\mathcal{K}\right\Vert_F^2 ,
		\label{eq:grsub2}
	\end{equation}
	where $\tau=1$.  To optimize \eqref{eq:grsub2}, we first denote the
	eigenvalue decompositions of $\mathbf{Q}_1^T\mathbf{Q}_1$,
	$\mathbf{Q}_2^T\mathbf{Q}_2$, and $\mathbf{Q}_3^T\mathbf{Q}_3$
	as
	\begin{align}
		\mathbf{Q}_1^T\mathbf{Q}_1&=\mathbf{V}_1\mathbf{\Sigma}_1\mathbf{V}_1^T\\
		\mathbf{Q}_2^T\mathbf{Q}_2&=\mathbf{V}_2\mathbf{\Sigma}_2\mathbf{V}_2^T\\
		\mathbf{Q}_3^T\mathbf{Q}_3&=\mathbf{V}_3\mathbf{\Sigma}_3\mathbf{V}_3^T,
	\end{align}
	respectively.  Then, letting
	$\mathcal{T}=\mathcal{H}\times_1\mathbf{Q}_1^T\times_2\mathbf{Q}_2^T\times_3\mathbf{Q}_3^T+\tau\mathcal{K}$
	and
	$\mathcal{T}'=\mathcal{T}\times_1\mathbf{V}_1^T\times_2\mathbf{V}_2^T\times_3\mathbf{V}_3^T$,
	the optimal solution is obtained via
	\begin{equation}
		\mathcal{S}=\mathcal{T}{''}\times_1\mathbf{V}_1\times_2\mathbf{V}_2\times_3\mathbf{V}_3 ,
	\end{equation}
	where
	$$\mathsf{Vec}\left(\mathcal{T}''\right)=\left(\mathbf{\Sigma}_3\otimes\mathbf{\Sigma}_2\otimes\mathbf{\Sigma}_1+\tau\mathbf{I}_{L_1L_2C}\right)^{-1}\mathsf{Vec}\left(\mathcal{T}'\right).$$
	
	\item
	The $\hat{\mathbf{G}}$ subproblem:\newline
	Solving for $\hat{\mathbf{G}}$ requires solving
	\begin{equation}
		\min_{\hat{\mathbf{G}}}\,\frac{\rho}{2}\left\Vert\hat{\mathbf{G}}-\mathbf{G}_{[B;\mathbf{t}]}+\frac{\mathbf{W}}{\rho}\right\Vert_F^2+{}\left\lVert \hat{\mathbf{G}}\right\rVert_{2,\gamma} ,
	\end{equation}
	which is a nonconvex, sparsity-inducing problem. We resort to the recently developed GAI \cite{ZZW2023} for an iterative solution.
	
	\item
	The $\mathcal{G}$ subproblem:\newline
	To find $\mathcal{G}$, we solve
	\begin{equation}
		\begin{aligned}
			\min_\mathcal{G}\,&\left\Vert\mathcal{Y}-\mathcal{G}\times_1\mathbf{U}_1\times_2\mathbf{U}_2\times_3\mathbf{R}\mathbf{U}_3+\frac{\mathcal{P}^{\cal{Y}}}{\rho}\right\Vert_F^2\nonumber\\&+\left\Vert\hat{\mathbf{G}}-\mathbf{G}_{[B;\mathbf{t}]}+\frac{\mathbf{W}}{\rho}\right\Vert_F^2+\sum_{r = 1}^{\mathsf{kr}_{\mathbf{D}}}\left\Vert\mathcal{G}-\mathcal{G}_r+\frac{\mathcal{P}_r}{\rho}\right\Vert_F^2.
		\end{aligned}
		\label{eq:G-subproblem}
	\end{equation}
	Introducing the variables
	\begin{align}
		\mathcal{H}&=\mathcal{Y}+\frac{\mathcal{P}^{\cal{Y}}}{\rho},\\
		\mathcal{S}&=\mathcal{G},\\
		\mathbf{Q}_1&=\mathbf{U}_1,\\
		\mathbf{Q}_2&=\mathbf{U}_2,\\
		\mathbf{Q}_3&=\mathbf{R}\mathbf{U}_3,\\
		\mathcal{K}&=\frac{\left(\mathcal{G}^{\mathbf{W}}+\sum_{r=1}^{\mathsf{kr}_{\mathbf{D}}}\mathcal{G}_r-\frac{\mathcal{P}_r}{\rho}\right)}{{\mathsf{kr}_{\mathbf{D}}}+1},\\
		\tau&={\mathsf{kr}_{\mathbf{D}}}+1 ,
	\end{align}
	where
	$\mathbf{G}^{\mathbf{W}}_{[\mathbf{t}]}\triangleq\hat{\mathbf{G}}+\frac{\mathbf{W}}{\rho}$
	is the B-unfolding of $\mathcal{G}^{\mathbf{W}}$, problem
	\eqref{eq:G-subproblem} is equivalent to \eqref{eq:grsub2}
	and can be optimized similarly.
	
	\item
	Updating $\mathcal{P}^\mathcal{X}$, $\mathcal{P}^\mathcal{Y}$,
	$\mathbf{W},\left\{\mathcal{P}_r\right\}_{r=1}^{\mathsf{kr}_{\mathbf{D}}}$:\newline
	Updating the auxiliary variables is done as
	\begin{align}
		\mathcal{P}^\mathcal{X}&\leftarrow\mathcal{P}^\mathcal{X}+\rho\Bigg(\mathcal{X}-\sum_{r = 1}^{\mathsf{kr}_{\mathbf{D}}}  \mathcal{G}_r\nonumber\\&\quad\times_1\mathbf{P}_1^{(r)}\mathbf{U}_1\times_2\mathbf{P}_2^{(r)}\mathbf{U}_2\times_3\mathbf{U}_3\Bigg),\\
		\mathcal{P}^\mathcal{Y}&\leftarrow\mathcal{P}^\mathcal{Y}+\rho\left(\mathcal{Y}-\mathcal{G}\times_1\mathbf{U}_1\times_2\mathbf{U}_2\times_3\mathbf{R}\mathbf{U}_3\right),\\
		\mathbf{W}&\leftarrow\mathbf{W}+\rho\left(\hat{\mathbf{G}}-\mathbf{G}_{[\mathbf{t}]}\right),\\
		\mathcal{P}_r&\leftarrow\mathcal{P}_r+\rho\left(\mathcal{G}-\mathcal{G}_r\right),\quad r=1,2,\dots,{\mathsf{kr}_{\mathbf{D}}}.
	\end{align}
	
\end{itemize}

\bibliographystyle{IEEEtran}
\bibliography{fowler,bibfile}

\end{document}